\documentclass[%
reprint,
superscriptaddress,
frontmatterverbose, 
preprintnumbers,
longbibliography,
 amsmath,amssymb,
 aps,
 pra,
notitlepage,
nofootinbib,
twocolumn
]{revtex4-1}

\usepackage{lipsum}

\usepackage{graphicx}
\usepackage{dcolumn}
\usepackage{bm}
\usepackage{hyperref}
\usepackage{subfigure}
\hypersetup{
colorlinks=true,
linkcolor=blue,
filecolor=blue,
citecolor=blue,  
urlcolor=blue,
}

\newcommand{\mycomment}[1]{}

\usepackage{amsfonts}

\usepackage{comment}

\usepackage{thmtools}
\usepackage{thm-restate}

\usepackage{enumerate} 

\usepackage{amssymb}
\usepackage{mathtools}

\usepackage{mathrsfs}
\usepackage{multirow}
\usepackage{bbm}

\usepackage[dvipsnames]{xcolor}

\usepackage{amsmath,empheq}

\definecolor{sanddune}{rgb}{0.59, 0.44, 0.09}
\definecolor{darkblue}{RGB}{0,0,102}
\definecolor{darkred}{rgb}{0.5,0.,0.}
\definecolor{BlueViolet}{RGB}{138,43,226}
\definecolor{SkyBlue}{RGB}{30,144,255}
\definecolor{DarkGreen}{RGB}{0,100,0}

\usepackage{amsthm}
\usepackage{amsmath}
\theoremstyle{plain}
\newtheorem{thm}{Theorem}
\newtheorem{lem}[thm]{Lemma}
\newtheorem{prop}[thm]{Proposition}
\newtheorem{cor}[thm]{Corollary}

\newtheorem{fact}[thm]{Fact}

\theoremstyle{definition}
\newtheorem{defn}{Definition}
\newtheorem{conj}{Conjecture}
\newtheorem{exmp}{Example}

\newtheorem*{rem}{Remark}

\newcommand{\ket}[1]{|#1\rangle}
\newcommand{\bra}[1]{\langle #1|}
\newcommand{\bracket}[2]{\langle #1|#2\rangle}
\newcommand{\ketbra}[2]{|#1\rangle\langle #2|}

\newcommand{\abs}[1]{\left|#1\right|}

\newcommand{\mc}{\mathcal}

\newcommand{\mbb}{\mathbb}

\newcommand{\ba}{\begin{eqnarray}}
\newcommand{\ea}{\end{eqnarray}}

\DeclareMathOperator{\Tr}{Tr}

\newcommand{\wt}{\mathrm{wt}}

\newcommand{\supp}{\operatorname{supp}}

\newcommand{\T}{\mathsf{T}}

\usepackage{pifont}
\newcommand{\cmark}{\ding{51}}  
\newcommand{\xmark}{\ding{55}}  

\usepackage{makecell}
\usepackage{multirow}

\newcommand{\norm}[1]{\left\| #1 \right\|}

\newcommand{\twonorm}[1]{\left\| #1 \right\|_2}
\newcommand{\bignorm}[1]{\big\| #1 \big\|}

\begin{document}

\title{Long-range nonstabilizerness and quantum codes, phases, and complexity}

\author{Fuchuan Wei}
\affiliation{Yau Mathematical Sciences Center, Tsinghua University, Beijing 100084, China}
\affiliation{Department of Mathematical Sciences, Tsinghua University, 100084 Beijing, China}

\author{Zi-Wen Liu}
\affiliation{Yau Mathematical Sciences Center, Tsinghua University, Beijing 100084, China}

\date{\today}

\begin{abstract}
As a necessary resource for quantum computational advantage, quantum magic (nonstabilizerness) is of fundamental importance in the study of quantum computation and physics.
We develop a systematic theory of \emph{long-range magic (LRM)}---nonstabilizerness that cannot be erased by shallow unitary circuits---and demonstrate its broad relevance.
By bridging LRM with fault-tolerant logical gate theory, we show the emergence of LRM families from quantum error-correcting codes and devise a simple yet powerful method for testing transversal logical gates.
Further, we introduce and characterize \emph{LRM phases} in which all ground states exhibit LRM, and identify certain non-Abelian topological orders as representative examples.
Then, adopting a complexity theory perspective, we demonstrate the classicality of non-LRM systems in e.g.~preparation and learning settings, and present a ``no low-energy trivial magic'' (NLTM) conjecture with key motivation in the quantum PCP context, for which our LRM results suggest a promising route.
Additionally, we demonstrate how to diagnose LRM with correlation functions.
The concepts and results admit substantive extensions to approximate (robust) and nongeometric scenarios.
Our LRM theory illuminates profound new connections among quantum resources, computational advantage, error correction and fault tolerance, many-body physics, and complexity theory.
\end{abstract}

\maketitle


\section{Introduction}

Understanding the nature of quantum systems that enable computational advantages over classical means is a central pursuit of quantum information science.
As implied by the Gottesman–Knill theorem~\cite{gottesman1998heisenberg,Aaronson2004Improved,NielsenChuang}, quantum computation composed solely of stabilizer elements is classically efficiently simulable, despite permitting arbitrarily large entanglement. Consequently, nonstabilizerness (magic)~\cite{BravyiKitaev,Veitch2014resource} emerges as a necessary resource for quantum computational advantages, with its characterizations and meanings extensively studied from e.g.~resource theory~\cite{Veitch2014resource,Howard2017application} and classical simulation~\cite{PhysRevX.6.021043,Bravyi2016Improved,Bravyi2019simulationofquantum,PRXQuantum.2.010345,PhysRevLett.123.170502} perspectives over the years.
Recently, nonstabilizerness in many-body quantum systems has attracted substantial interest from both quantum information and physics communities, with extensive research efforts devoted to understanding its interplay with entanglement (see e.g.~Refs.~\cite{Liu2022manybody,Lorenzo2022stabilizer,PhysRevB.107.035148,Lami2023Nonstabilizerness,Tarabunga2024Nonstabilizerness,Frau2024Nonstabilizerness,gu2024magicinducedcomputationalseparationentanglement,zejun2025Nonequilibrium}) and connecting it with many-body physics (see e.g.~Refs.~\cite{Sarkar2020Characterization,White2021Conformal,Liu2022manybody,Tarabunga2023ManyBody,gu2024magicinducedcomputationalseparationentanglement,Mircea2024Dynamical}).

In the development of condensed matter physics, the concept of long-range entanglement~\cite{ChenGuWen10:lu,wen2013topological,HastingsWen:quasi} has played a vital role, underpinning intrinsically quantum notions of orders and phases that extend beyond the traditional Landau paradigm.
Here ``long-range'' captures the global, topological nature of the feature, requiring it to be preserved under shallow local quantum circuits representing ``smooth'' deformations that do not alter fundamentally global properties.

We adopt a computational perspective and study \emph{long-range magic (LRM)}~\cite{White2021Conformal,Ellison2021symmetryprotected}, defined as a form of nonstabilizerness that cannot be erased by shallow circuits, which captures the ``topological'' kind of nonstabilizerness embedded in global structures and robust against local perturbations. 
Several insights are worth noting. First, it is evident that LRM is a notion strictly stronger than long-range entanglement, hence representing a refined, computationally grounded characterization of long-range quantumness.
Second, from a computational complexity viewpoint, LRM can be understood as high-complexity nonstabilizerness that is intrinsically far from the classical regime, thus naturally holding importance in quantum complexity theory.
Third, as we will demonstrate, LRM bridges numerous fields including quantum computation, error correction and fault tolerance, complexity, and many-body physics, in new interesting ways, drawing on insights from these fields to inform one another.
Building on these perspectives, a rigorous understanding and constructions of LRM would have broad applications and spur new directions across these important domains.

In this work, we fulfill this goal by developing a systematic theory of LRM as well as its robust extension that accommodates highly permissive approximations, and studying its construction and significance from diverse perspectives. 
Note that proving LRM feature essentially amounts to the fundamentally challenging problem of setting nontrivial circuit lower bounds for erasing nonstabilizerness. 
We first demonstrate the manifestation of robust LRM in topological codes and orders, showcasing how quantum computational resources emerge from prominent physical systems.
Here, a key insight is that shallow circuits underpin fault-tolerant logical operations in the quantum coding context and the characterization of quantum phases of matter, enabling us to bridge LRM with these fields.
In particular, leveraging fault-tolerant gate theory we show that topological codes protect LRM in logical nonstabilizer states, along the way establishing various new results such as a robust generalization of the Bravyi--König theorem that may find independent interest in quantum error correction.
The link between LRM and logical gate theory also leads us to a novel application for logical gates: we devise a simple and widely useful method for testing whether a transversal implementation of certain non-Clifford gates is possible on stabilizer codes, e.g., ruling out transversal $T$ gates on the gross code~\cite{Bravyi2024High,yoder2025tourgrossmodularquantum}.
Furthermore, we formalize the theory of \emph{LRM phases} which capture ``inherently magical'' phases of matter and thereby offer a new paradigm for quantum phase classification and identification based on computational characteristics.
{We rigorously link the LRM-ness of a topological order to its realizability as a stabilizer model and derive clear ground state degeneracy conditions for its certification, which enable us to establish profound connections between the non-Abelian and LRM features of topological orders and identify explicit examples of LRM phases such as the doubled Fibonacci topological order.}
These LRM phases are also mathematically remarkable in that they naturally evade all stabilizer states that pervade the Hilbert space, highlighting LRM as a robust emergent property with computational significance.
Lastly, we pivot to complexity theory through the lens of LRM. We first elucidate the quantum complexity significance of LRM from widespread perspectives including preparation, description, and learning complexities.
Then, building upon our LRM phase results, we formulate a ``no low-energy trivial
magic'' (NLTM) conjecture, which unifies the trivial entanglement and nonstabilizerness types of classical proofs and therefore strengthens the recently proven no low-energy trivial state (NLTS) theorem~\cite{freedman2013quantumsystemsnonkhyperfinitecomplexes,Anshu2023NLTS}, potentially yielding important advances towards quantum PCP. We explain how our LRM phases suggest a compelling path towards fully proving it.
Additionally, we present a general and practical scheme for diagnosing LRM with correlation functions, which we use to show e.g.,~$\mathrm{C}^{n-1}Z$ and GHZ-type states are LRM. Note that most of our LRM results are robust even against $O(1/n)$ approximation and do not rely on geometric locality.

\section{Preliminaries}
 
As a foundational setup, consider $n$-qubit quantum systems\footnote{For understanding the essence of the theory, it is largely sufficient and convenient to consider qubit systems. Cases of general dimensions are useful but technically more involved; rigorous details for qudit cases can be found in the appendix.}.
The Clifford group $\mathcal{C}_n$ is defined as all unitaries $U$ such that $U P U^\dagger \in \mathcal{P}_n$ for all $P \in \mathcal{P}_n$, where $\mathcal{P}_n$ is the $n$-qubit Pauli group; equivalently, $\mathcal{C}_n$ is the group generated by Hadamard, phase, and CNOT gates. Then the states that can be generated from $\ket{0}^{\otimes n}$ by acting an $n$-qubit Clifford unitary are stabilizer states.
Our discussion builds on the seminal observation noted earlier that ``magic'' (nonstabilizerness) is a necessary quantum computational resource due to the efficient classical simulability of any process composed of stabilizer-preserving elements~\cite{gottesman1998heisenberg,Aaronson2004Improved,NielsenChuang}.

Henceforth, by \emph{shallow circuits} we mean families of local quantum circuits consisting of a constant ($O(1)$ as the system size grows)\footnote{Note that extensions of subsequent definitions and results to higher circuit depths are feasible and sometimes meaningful 
(see Footnote \ref{fn:depth} for more information); for ease of exposition we stick to constant depth in the general definitions, which suffices for conveying the essence.} number of layers of disjoint local unitary gates which act only on finitely many neighboring (as determined by an underlying connectivity/adjacency graph) subsystems.

\begin{defn}
We say a family of states $\{\ket{\psi_n}\}$ composed of $n$-partite states $\ket{\psi_n}$ exhibits \emph{long-range magic (LRM)} if for any family of shallow circuits $\{U_n\}$, $U_n\ket{\psi_n}$ remains a magic state for sufficiently large $n$.
Furthermore, we say $\{\ket{\psi_n}\}$ exhibits \emph{strong LRM} if its LRM is robust against even $\Omega(1/n)$ error, i.e., $\big\|U_n\ketbra{\psi_n}{\psi_n}U_n^\dagger-\ketbra{S_n}{S_n}\big\|_1=\Omega(1/n)$ for any family of shallow circuits $\{U_n\}$ and any family of stabilizer states $\{\ket{S_n}\}$.
\end{defn}

Conversely, we say $\{\ket{\psi_n}\}$ has \emph{short-range magic (SRM)} if there exists a family of shallow circuits $\{U_n\}$ such that $\{U_n\ket{\psi_n}\}$ are stabilizer states.
We may sometimes simply say some state has LRM while implicitly referring to its naturally associated state family.

Some elucidating remarks are in order. Importantly, LRM is a strictly stronger condition than long-range entanglement (cannot be transformed into product states) and therefore characterizes a more refined property.
Also, the presence or absence of LRM does not depend on the choice of local bases.
Furthermore, this $\Omega(1/n)$ robustness in our notion of strong LRM is highly lenient and bears fundamental physical and information-theoretic significance, arising as a universal critical error scaling in connection to various essential principles including symmetry, order and phase classification~\cite{faist20,liu2022approximate,Yi_2024}. For example, physically,
an $o(1/n)$ error ensures bounded energy deviation for natural systems so this $\Omega(1/n)$ robustness is naturally protected by an energy gap.

\section{Long-range magic and QEC codes}\label{sec:LRM_qec}

Our first insight is the fundamental connection between LRM and quantum error correction (QEC),
which yields new understanding on both the information protection and processing properties of QEC codes.
We elucidate this insight through topological stabilizer codes (TSCs), which play a central role across quantum computing and many-body physics, encompassing toric~\cite{KITAEV2003anyons}, surface~\cite{Topological2022Dennis}, color codes~\cite{Bombin2006TopologicalQuantumDistillation} and their variants as prominent examples. Specifically, we bridge logical nonstabilizerness to the intensively studied logical non-Clifford gate problem, and thereby establish LRM through limitations on fault-tolerant  gates~\cite{Bravyi2013Classification} (the results extend to much broader kinds of codes even without geometry~\cite{nogohgp}).

Formally, now consider $D$-dimensional TSC families characterized by the following properties: the physical qubits are arranged on a lattice embedded in a suitable $D$-dimensional manifold with {Pauli} checks (stabilizer generators) locally supported on bounded-size regions, and the code distance grows macroscopically with the lattice size.

\begin{thm}\label{thm:LRM_from_TSC}
For any $D$-dimensional ($D\ge2$) TSC family, if $\ket{\phi}\notin\big\{V^\dagger\ket{0^k}\mid V\in \mc{C}_k^{(D)}\big\}$\footnote{The Clifford hierarchy is a nested tower of unitaries defined in the following recursive manner:  for an $n$-qubit system, the first level is given by the Pauli group $\mc{C}_n^{(1)}:=\mc{P}_n$, and then the $l$-th level ($l\ge2$) is defined as $\mc{C}_n^{(l)}:=\{U\in\operatorname{U}(2^n)\mid UPU^\dagger\in\mc{C}_n^{(l-1)},\forall P\in\mc{P}_n\}$, where $\operatorname{U}(2^n)$ denotes the unitary group. In particular, $\mc{C}_n^{(2)}$ is the Clifford group. A specific $\mc{C}_n^{(3)}$ gate of outstanding importance is $T=\operatorname{diag}(1,e^{i\pi/4})$.} (where $k$ is the number of logical qubits), then the corresponding logical state family $\ket{\overline{\phi}}$ exhibits strong LRM.
\end{thm}

Particularly for 2D, $\{V^\dagger\ket{0^k}\mid V\in \mc{C}_k^{(2)}\}$ is the set of $k$-qubit stabilizer states and Theorem \ref{thm:LRM_from_TSC} implies that any logical nonstabilizer state exhibits LRM.

With many details and rigorous proofs given in Appendix~\ref{app:LRM_TSC}, here we highlight the key physical mechanism.
Heuristically, shallow circuits capture ``smooth'', topologically trivial deformations that do not alter global phase properties. 
A physical insight behind the theorem is that erasing the magic of certain logical states requires deforming the underlying TSC into a distinct phase, which is beyond the capability of shallow circuits.

More concretely, suppose that a shallow circuit could erase the magic of a code state $\ket{\overline{\phi}}$, i.e.~map it to a stabilizer state.
This shallow circuit induces a deformed code that remains in the same phase and contains a stabilizer code state.
As a critical lemma, we prove that all such deformed codes must themselves be a TSC (see Appendix~\ref{app:proof_of_LRM_from_TSC}).
Note that, despite its conceptual simplicity, making this statement rigorous entails substantial insight and effort.
These insights also underpin the LRM phase results later (Section~\ref{sec:LRMphase}).

Next, observe that these
shallow circuits induce fault-tolerant operations mapping between the original and deformed codes,  bridging our problem to the rich theory of fault-tolerant logical gates. Since both codes are TSCs as demonstrated above, we can apply the fundamental result of Bravyi and König on fault-tolerant gates of TSCs~\cite{Bravyi2013Classification}, which states that the Clifford hierarchy level of shallow-circuit-implementable logical actions is bounded by the spatial dimension. It follows that our initial assumption cannot hold except for a restricted, discrete set of $\ket{\overline{\phi}}$'s (corresponding to the Clifford hierarchy restriction), implying that any state outside this set exhibits LRM, which yields the theorem statement.

To promote this to \emph{strong} LRM, we make every step of the above argument error-robust. In particular, we establish the following robust extension of the Bravyi--König theorem (detailed proof in Appendix~\ref{app:robust_bk}) which is potentially of independent interest in QEC and fault tolerance.

\begin{thm}[Robust Bravyi--König]
Let $\Pi_1=J_1J_1^\dagger$ and $\Pi_2=J_2J_2^\dagger$  be the code projectors of two $D$-dimensional ($D\ge2$) TSCs defined on the same lattice, each encoding $k$ logical qubits.
If a shallow circuit $U$ satisfies $\big\|U\Pi_1U^{\dagger}-\Pi_2\big\|_2=o(1)$, then the induced logical unitary satisfies
\begin{equation}
\min\nolimits_{Q\in\mc{C}_k^{(D)}}\bignorm{J_2^\dagger UJ_1-Q}_2=o(1).
\end{equation}
\end{thm}

To provide more intuitive insight into how LRM emerges from topology as unveiled by Theorem~\ref{thm:LRM_from_TSC}, we present an explicit proof for a basic {toric code} $\ket{\overline{T0}}$ example in Appendix~\ref{app:t0}.
Interestingly, this example exhibits a state-preparation task achievable by the ``shallow circuit then Clifford'' model but not the ``Clifford then shallow circuit'' model; as a related insight, LRM can be generated by only one non-Clifford gate (first apply a $T$ gate and then perform the Clifford encoding for the toric code) rather than a macroscopic amount as one may naturally expect.
We also note an intriguing property of topological nonstabilizerness of TSCs: the erasure of any union of contractable regions on the torus does not reduce the amount of nonstabilizerness in $\ket{\overline{T0}}$ since it is correctable by Clifford recovery operations.

\subsection*{Logical gate and symmetry testing}

The above connection between magic erasure and logical gates not only helps establish LRM, but also leads to interesting applications to the understanding of logical operations (which also correspond to code symmetries) which is crucial to quantum computing and physics,  as we will now demonstrate.

Here we consider transversal gates, the most important and simple setting for fault-tolerant logical operations.
As a variant of the above definition of SRM/LRM, we say an $n$-qubit state $\ket{\psi}$ has $\mathrm{SRM}_0$ (or $\mathrm{LRM}_0$) if it can (or cannot) be mapped to a stabilizer state by transversal unitaries taking the form of a tensor product of single-qubit unitaries.
Notably, $\mathrm{LRM}_0$ essentially captures local-basis-independent magic, which merits separate study.

Now consider an $[\![n,k]\!]$ stabilizer code: if a $k$-qubit unitary $U$ admits a transversal implementation, then for any $k$-qubit state $\ket{\phi}\in\{U\ket{S}:\ket{S}\text{ a $k$-qubit stabilizer state}\}$, the corresponding $n$-qubit logical state $\ket{\overline{\phi}}$ has $\mathrm{SRM}_0$ because its magic can be removed by the transversal implementation of logical $U^\dagger$.
Consequently, the $\mathrm{LRM}_0$-ness of $\ket{\overline{\phi}}$ rules out any transversal implementation of $U$.

Crucially, we can formulate a simple test for $\mathrm{LRM}_0$ based on Pauli expectation values.
For any region $R\subset[n]$, the quantity
\begin{equation}
f(\ket{\psi},R)=\sum_{P\in\mc{P}_n^+:\,\supp(P)=R} \big|\bra{\psi}P\ket{\psi}\big|^2
\end{equation}
is invariant under transversal unitaries (Lemma~\ref{lemma:pauli_support_invariance} in Appendix~\ref{app:logical_gate_testing}), where $\mc{P}_n^+=\{\mbb{I},X,Y,Z\}^{\otimes n}$.
Since for any stabilizer state $\ket{S}$ we have $\bra{S}P\ket{S}\in\{-1,0,1\}$, it follows that $f(\ket{S},R)$ must be an integer.
Therefore, if for some region $R$, $f(\ket{\psi},R)$ is non-integer, then $\ket{\psi}$ exhibits $\mathrm{LRM}_0$.

Combining the above insights, we obtain the following result (detailed proof in Appendix~\ref{app:logical_gate_testing}).

\begin{thm}[Transversal gate/symmetry testing (informal)]
Consider an $[\![n,k]\!]$ stabilizer code. Let $U$ be a $k$-qubit unitary. If for some $\ket{\phi}\in\{U\ket{S}:\ket{S}\text{ is a $k$-qubit stabilizer state}\}$, the $n$-qubit code state $\ket{\overline{\phi}}$ satisfies
\begin{equation}\label{eq:transversal_gate_testing_main}
\sum_{P\in\mc{P}_n^+:\,\supp(P)=R} 
\big|\bra{\overline{\phi}}P\ket{\overline{\phi}}\big|^2\notin \mathbb{Z},
\end{equation}
for some subset $R\subset[n]$, then $\ket{\overline{\phi}}$ has $\mathrm{LRM}_0$, and logical $U$ does not admit any strictly transversal implementation.
\end{thm}

This criterion is simple and widely applicable. It only entails the computation of Pauli expectation values on code states, does not require any geometric or structural assumptions (such as CSS structure or specific gates), and it applies to both asymptotic and finite-size codes. Notably, it enables us to handle general-dimensional cases and codes with long-range check operators, yielding results that are not accessible via existing logical gate results such as the Bravyi--König bound~\cite{Bravyi2013Classification}.
Here we give two representative examples:
\begin{exmp}
All $D$-dimensional ($D\ge3$) toric codes~\cite{Kubica2015Unfolding} do not admit strictly transversal  $T$ gate on any logical qubit.
\end{exmp}
\begin{exmp}
The $[\![144,12,12]\!]$ gross code~\cite{Bravyi2024High,yoder2025tourgrossmodularquantum} does not admit strictly transversal $T$ gate on any logical qubit.
\end{exmp}
Detailed derivations and proofs can be found in Appendix~\ref{app:logical_gate_testing}.
They showcase applications to asymptotic and finite-size codes respectively, and remarkably, neither result follows from known bounds including Bravyi--König.

\section{Topological order and long-range magic phases of matter}\label{sec:LRMphase}

Fundamentally characterized by long-range entanglement, topological order represents an intrinsically quantum notion of phase of matter and has been of central interest in modern quantum many-body physics. A major motivation for our LRM theory is to incorporate computational complexity into the understanding and classification of quantum phases. This leads to a natural question: are there phases in which every ground state exhibits LRM throughout the phase, thereby intrinsically harboring robust quantum computational resources?

We answer this in the affirmative and demonstrate that LRM, as a strictly stronger notion than long-range entanglement, provides a novel lens into phases of matter, in particular yielding a sharp dichotomy between LRM and non-LRM topological orders (see Table~\ref{tab:LRE_LRM_phase} for representative examples).

\begin{table}[htbp]
\centering
\resizebox{0.49\textwidth}{!}{
\begin{tabular}{c|c|c|c}
\hline
\hline
Examples &  Ground states & LRE phase & LRM phase \\
\hline
\hline
Repetition code & Some SRE & \xmark & \xmark \\
\hline
Toric code & All LRE, some SRM & \cmark & \xmark \\
\hline
Doubled Fibonacci & All LRM & \cmark & \cmark \\
\hline 
\end{tabular}
}
\caption{
Phase classification by entanglement and magic.
Short-range entangled (SRE) phases (e.g., repetition code) contain at least one SRE ground state, and long-range entangled (LRE) phases consist entirely of LRE ground states and are closely related to the notion of topological order. LRE phases are further stratified into subclasses distinguished by magic: non-LRM phases (e.g., toric code, 2D Abelian topological orders) admit at least one stabilizer ground state, whereas LRM phases (e.g., doubled Fibonacci, non-Abelian orders) contain only LRM ground states.
}
\label{tab:LRE_LRM_phase}
\end{table}

To formalize this new classification, note that topological orders admit local Hamiltonian realizations, such as the quantum double~\cite{KITAEV2003FaultTolerant} and string-net~\cite{Levin2005StringNet} models, which characterize these phases through their ground spaces. These Hamiltonians are typically defined on generic qudits, for which the corresponding stabilizer formalism becomes more intricate than in the standard qubit setting (see Appendix~\ref{app:qudit_stab_formalism} for details).
Furthermore, since the stabilizer formalism and magic theory depend on the dimension and tensor structure of the Hilbert space, studying magic in topological orders requires explicitly specifying the \emph{local configuration}---a tuple of integers $(q_1,\cdots,q_m)$---such that each local site on which the Hamiltonian acts is associated with a Hilbert space $\bigotimes_{j=1}^m \mathbb{C}^{q_j}$ (integers $q_j$ need not be prime).
For instance, a qubit corresponds to $q_1=2$, while a 6-dimensional qudit could have a configuration of $q_1=6$ or $(q_1,q_2)=(2,3)$.

\begin{defn}\label{def:LRM_SRM_phase}
Given some local configuration, we call a topological order a \emph{(strong) LRM phase}
if, for every compatible local Hamiltonian realization, all ground states exhibit (strong) LRM\footnote{\label{fn:depth}In this definition, (strong) LRM is discussed with respect to geometrically local circuits, following physics conventions; however, the (strong) LRM phase results here naturally hold for shallow circuits with all-to-all connectivity. Moreover, the circuit-depth lower bound for eliminating magic can be strengthened from super-constant to $\Omega(\log n)$ in the all-to-all-connectivity setting, and to $\Omega(\text{system size})$ in the geometrically local setting, which are considered particularly strong bounds.}.
\end{defn}

Note that ``LRM phase" is a stronger notion than the emergence of LRM in TSCs discussed in Section~\ref{sec:LRM_qec}. While certain LRM logical state families emerge from TSCs due to their topological protection, the code themselves are defined by commuting Pauli operators so that they necessarily admit at least one stabilizer state within their code/ground space. That is, topological orders that are realizable by TSCs cannot constitute LRM phases. For example, when considering local configuration $q\in\mbb{Z}$, the $\mathbb{Z}_q$ topological order can be realized by the $\mathbb{Z}_q$ toric code, so it is not considered an LRM phase.

Remarkably, we establish the converse implication: any topological order that is not an LRM phase admits a TSC realization.
This yields a profound correspondence between LRM phases and the impossibility of a stabilizer realization, leading to the following theorem:

\begin{thm}\label{thm:LRM_phase_strongLRM_phase}
Given a certain local configuration:
\begin{enumerate}[i)]
\item A topological order is an LRM phase iff it cannot be realized by any TSC with this local configuration.
\item If a topological order cannot be approximated up to $o(1)$ error in projectors' distance by any TSC with this local configuration, then it is a strong LRM phase.
\end{enumerate}
\end{thm}

Again, despite their simple and general form, the rigorous proofs of these results are nontrivial and have interesting consequences as we will demonstrate. We sketch the key insights here and relegate the full proofs to  Appendix~\ref{app:TopologicalOrder}.

For (i), the key is to show that any topological order possessing a stabilizer ground state $\ket{S}$ (since shallow circuits do not alter phase, local Hamiltonians with SRM ground states can be turned into local Hamiltonians with stabilizer ground states) admits a TSC realization.
According to the TQO condition~\cite{LiebRobinson2006Bravyi,Bravyi2010TQO}, the ground space of a topological order constitutes a QEC code with a macroscopic distance, so all ground states are locally indistinguishable. Consequently, the local properties of a single ground state fully characterize the entire ground space.
Specifically, the ground space consists entirely of states that are locally indistinguishable from $\ket{S}$, forcing it to be a TSC governed exactly by the local stabilizers of $\ket{S}$. This argument also underlies the proof of Theorem~\ref{thm:LRM_from_TSC}.

Then for (ii), we need to further invoke robustness and show that any topological order containing a ground state $o(1/n)$ close to a stabilizer state $\ket{S}$ in trace distance admits a TSC realization with $o(1)$ error.
Crucially, such robustness is naturally guaranteed by the constant energy gap of conventional topological orders: the TSC defined by local stabilizers of $\ket{S}$ has a basis with $o(1)$ excitation energy relative to the local Hamiltonian, thereby ensuring the proximity of the TSC to the ground space.

\subsection*{Dimension criterion for strong LRM phases}

Remarkably, Theorem~\ref{thm:LRM_phase_strongLRM_phase} yields a concrete and practical condition for identifying strong LRM phases based simply on the ground space degeneracy or dimension of the topological orders.

\begin{cor}\label{corollary:dimension_mismatch}
Given a local configuration $(q_1,\cdots,q_m)$, if the ground space degeneracy of a topological order contains a prime factor that does not divide $\prod_{i=1}^mq_i$, then it constitutes a strong LRM phase.
\end{cor}

This dimension criterion immediately allows us to identify representative strong LRM phases.
For a 2D bosonic topological order on a closed surface of genus $g$, the ground state degeneracy is given by~\cite{Barkeshli2009Structure,Kong2020Relation}
\begin{equation}\label{eq:degeneracy_formula}
\Big(\sum_{i=0}^{N-1} d_i^2 \Big)^{g-1} \sum_{i=0}^{N-1} d_i^{-2(g-1)},
\end{equation}
where $d_0,\cdots,d_{N-1}$ are the quantum dimensions of the anyons in the theory.
For instance, consider the doubled Fibonacci topological order, which can be realized by a Levin–Wen Hamiltonian on qubit systems~\cite{KOENIG20102707TuraevViro}.
For the qubit configuration, Corollary~\ref{corollary:dimension_mismatch} implies that a topological order with the ground space dimension that is not a power of two is a strong LRM phase.
The doubled Fibonacci topological order possesses anyons with quantum dimensions $1,\tau,\tau,\tau+1$~\cite{KOENIG20102707TuraevViro} where $\tau = (1+\sqrt{5})/2$ is the golden ratio. According to \eqref{eq:degeneracy_formula}, the ground space degeneracy on a genus-$g$ closed surface is given by $5^{g-1}(\tau^{g-1}+\tau^{1-g})^{2}$ (which yields 4, 25, 225, 2500 for $g=1, 2, 3,4$, respectively).
Since the term $(\tau^{g-1}+\tau^{1-g})^{2}$ is always an integer (see Appendix~\ref{app:number_theory_lemmas}), the degeneracy for $g \ge 2$ contains 5 as a prime factor.
Hence, Corollary~\ref{corollary:dimension_mismatch} immediately indicates:
\begin{exmp}
Qubit doubled Fibonacci topological order living on a closed surface of genus $g\ge2$ is a strong LRM phase.
Consequently, for the standard Levin--Wen Hamiltonian realization of this order on such manifolds, every family of ground states exhibits LRM.
\end{exmp}

As another important example, consider the $S_3$ quantum double model, which has eight anyon types with quantum dimensions $1,1,2,3,3,2,2,2$~\cite{Chen2025universal}. Applying \eqref{eq:degeneracy_formula}, the genus-$g$ ground space degeneracy is $2\cdot6^{2g-2}+4\cdot3^{2g-2}+2\cdot2^{2g-2}$. For $g\ge2$, it has a prime factor that does not divide 6 (see Appendix~\ref{app:number_theory_lemmas}), so by Corollary~\ref{corollary:dimension_mismatch} we obtain:
\begin{exmp}
The $S_3$ quantum double canonically defined on qudits of dimension $\abs{S_3} = 6$ and a closed surface of genus $g\ge2$ is a strong LRM phase.
\end{exmp}

\subsection*{Non-Abelian topological orders and intrinsic LRM phases}

While the presence of magic in a topological order generally depends on the specific choice of the local Hilbert space, some phases have entanglement structures that are fundamentally incompatible with stabilizer states.
To capture this strong notion, we call a topological order an \emph{intrinsic LRM phase} if it constitutes an LRM phase for \emph{any} local configuration. 
By Theorem~\ref{thm:LRM_phase_strongLRM_phase}, a topological order is an intrinsic LRM phase iff it cannot be realized by any TSC (regardless of the local configuration).

To illustrate the relevance of this distinction, consider the double semion phase. On qubits, it constitutes an LRM phase given the standard belief that it does not admit a TSC realization; this follows from the understanding that all translation-invariant 2D qubit TSCs belong to the same phase as decoupled copies of toric codes~\cite{Bombin2012Universal,Bombín2014Structure,Haah2021Classification}. However, interestingly, the double semion phase is known to admit a TSC realization when defined on 4-dimensional qudits~\cite{Ellison2022Pauli}. Because its magic can be resolved by tuning the local dimension, it is not considered an intrinsic LRM phase.

We point out that intrinsic LRM phases naturally arise from non-Abelian topological orders. For 2D topological orders, existing arguments relate the possibility of  TSC realizations to Abelian versus non-Abelian nature of the order~\cite{Potter2016Symmetry,Ellison2022Pauli}, by which we conclude:
\begin{rem}
Any 2D non-Abelian topological order constitutes an intrinsic LRM phase, whereas any 2D Abelian topological order that admits a gapped boundary does not.
\end{rem}

Crucially, since the concept of intrinsic LRM phases can be defined for topological orders with connectivity more general than 2D, it extends the Abelian/non-Abelian classification beyond 2D topological orders, encompassing higher-dimensional topological orders~\cite{WalkerWang2012,Tian2018Classification,Tian2019Classification} and orders with general connectivity~\cite{rakovszky2023physicsgoodldpccodes,rakovszky2024physicsgoodldpccodes,deroeck2024ldpcstabilizercodesgapped,placke2024topologicalquantumspinglass}.

\section{Quantum complexity of long-range magic and NLTM conjecture}

We now shift to a complexity-theoretic perspective.
A critical observation is that any SRM state family $\{\ket{\psi_n}=U_n\ket{S_n}\}$ admits a succinct classical description specified directly by the shallow circuit $U_n$ and the stabilizer tableaux of $\ket{S_n}$.
Furthermore, SRM ensures efficient classical tractability of all local information: any $O(1)$-body reduced density matrix (hence expectation values of local Hamiltonians) can be computed in polynomial time by tracking the backward lightcone of $U_n$ and evaluating the corresponding reduced state on $\ket{S_n}$.
This observation unveils a fundamental connection between our theory and the quantum probabilistically checkable proofs (qPCP) program, a central open problem in quantum complexity theory (see the appendix and e.g.~Ref.~\cite{aharonov2013quantumpcpconjecture} for background),
indicating that SRM states can serve as classical witnesses for qPCP Hamiltonians. In particular, the recent landmark progress towards qPCP known as  the NLTS theorem~\cite{freedman2013quantumsystemsnonkhyperfinitecomplexes,Anshu2023NLTS} rules out shallow-circuit-type low-energy proofs, while SRM states substantially broaden the landscape of such proofs that must be altogether precluded in order to prove qPCP.
Based on this, we formulate the following no low-energy trivial magic (NLTM) conjecture, which is a necessary consequence of qPCP (detailed in Appendix~\ref{app:NLTM_qPCP}) strictly stronger than NLTS and would bring us significantly closer to qPCP (also mentioned in a talk~\cite{Nirkhe_2023}):

\begin{conj}[NLTM]
There exists a constant $\epsilon>0$ and a family of $O(1)$-local Hamiltonians $\{H^{(n)}=\sum_{i=1}^{m^{(n)}}H_i^{(n)}\}$ acting on $n$ qubits with $m^{(n)}=\Theta(n)$ and $\|H_i^{(n)}\|_\infty\le1$, such that any family of states $\{\ket{\psi_n}\}$ satisfying $\bra{\psi_n}H^{(n)}\ket{\psi_n}\le\lambda_{\min}(H^{(n)})+\epsilon n$ exhibits LRM.
\end{conj}

A key conceptual advance of NLTM is that it unifies disparate types of efficient witnesses originating from trivial entanglement and classical simulation algorithms (see also Refs.~\cite{gharibianlegall,coble2023hamiltonians}), thereby elevating them into a more coherent condition for the qPCP conjecture.
Hence, the investigation of NLTM is expected to yield more solid evidence and provide new insights and techniques for qPCP.
Our LRM phase results above can be interpreted as ``no lowest-energy trivial magic" theorems and point to a compelling route towards a full resolution of the NLTM conjecture: namely, to extend LRM phases, e.g.,~non-Abelian topological orders, to general connectivity, in light of the role of good qLDPC codes in proving NLTS~\cite{Anshu2023NLTS}.

We further offer a baseline discussion from gate and learning complexity perspectives, which are of independent interest and merit further study.
First note that SRM states have $O(n^2)$ gate complexity: any $n$-qubit stabilizer state can be prepared using $O(n^2)$ gates~\cite{Aaronson2004Improved,Nash2020optimizations} and the additional shallow circuit contributes at most $O(n)$ more gates.
Therefore, states with $\omega(n^2)$ gate complexity must exhibit LRM.
An immediate corollary is that almost all quantum states exhibit LRM since the typical gate complexity is $\Omega(2^n)$~\cite{Plesch2011quantum,sun2023asymptotically} with probability $\rightarrow1$ as $n\rightarrow\infty$.
A more concrete random construction of LRM states can be realized by applying 2-qubit Haar-random gates at random locations in an all-to-all connected  system. In this setting, the gate complexity of the resulting state grows linearly with the number of gates~\cite[Corollary 1.9]{chen2024incompressibility}, and applying $\omega(n^6)$ such random gates generates a state with gate complexity $\omega(n^2)$ (hence LRM) with probability approaching $1$.
Another consequence of the $O(n^2)$ complexity is that all SRM states are sample-efficiently learnable: they can be reconstructed to $\epsilon$-precision with high probability using $\widetilde{O}(n^2\epsilon^{-2})$ copies~\cite{Zhao2024Learning}.

\section{Diagnosing long-range magic with magical correlation}

Having discussed various theoretical perspectives, a natural practical question is whether one can easily diagnose LRM. 
Here, we present a powerful scheme for diagnosing LRM (which applies to arbitrary connectivity) based on two-point ``magical'' correlations and explicitly demonstrate it on two prominent entangled state families.

\begin{figure}[t]
\centering
\includegraphics[width=0.33\textwidth]{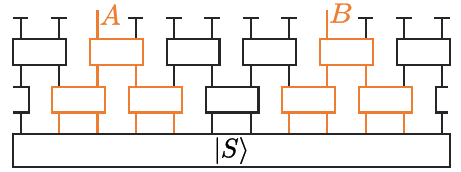}
\caption{
Correlations in SRM systems. For two distant subsystems $A$ and $B$ with nonintersecting backward lightcones, the bipartite reduced state can only exhibit correlations mediated by a stabilizer state $\ket{S}$. This constraint can be used to diagnose LRM. (Here we depict the 1D case for simplicity but the same reasoning applies to arbitrary connectivity.)
}
\label{fig:SRM_lightcone}
\end{figure}

More specifically,
observe that separated pairs of qubits in an SRM state can only exhibit rigid, stabilizer-type correlations (see Fig.~\ref{fig:SRM_lightcone}), which leads to the following theorem (detailed proof given in Appendix~\ref{app:SRM_2qubits_cor_by_EPRs}).
\begin{thm}\label{thm:SRM_2qubits_cor_by_EPRs}
For arbitrary connectivity, suppose $\{\ket{\psi_n}\}$ is an SRM state family given by $U_n\ket{S_n}$ where $U_n$ is a shallow circuit, then for any two qubits sufficiently well separated (such that their backward lightcones specified by $U_n$ do not intersect), the reduced state of these two qubits can be obtained by local operations on $O(1)$ EPR pairs.
\end{thm}

Consequently, we can certify LRM by identifying behaviors of correlation functions incompatible with the EPR constraints which can be concretely characterized. 

For illustration, we first examine the canonical hypergraph state family $\ket{\mathrm{C}^{n-1}Z}:=\mathrm{C}^{n-1}Z\ket{+^n}$ with $\mathrm{C}^{n-1}Z=\operatorname{diag}(1,\cdots,1,-1)$ being the $n$-qubit multi-controlled-$Z$ gate, which is of major interest in quantum computing~\cite{Kim2022iToffoli,Nguyen2024Floquet,Bluvstein2024Logical,Wang_2024,wei2024noise,PhysRevResearch.3.013118,WillsHsiehYamasaki2025ConstantOverhead}.   

\begin{exmp}
$\ket{\mathrm{C}^{n-1}Z}$ family exhibits LRM.
\end{exmp}

The proof goes specifically as follows, with an illustration in Fig.~\ref{fig:bc_correlation} and complete technical details in Appendix~\ref{app:LRM_from_correlation}.
First, we show that a two-qubit state $\rho$ with $\langle X\otimes X\rangle_\rho\equiv b$ and $\langle X\otimes \mbb{I}\rangle_\rho=\langle \mbb{I}\otimes X\rangle_\rho \equiv c$ can be generated by local operations on $K$ EPR pairs  
if and only if $(b,c)$ satisfies a nontrivial set of constraints corresponding to $K$ (a special case of Proposition~\ref{prop:correlations_EPR} in Appendix~\ref{app:LRM_from_correlation}). We illustrate the $K=2,3$ cases in  Fig.~\ref{fig:bc_correlation}:  the feasible regions of $(b,c)$ given by Proposition~\ref{prop:correlations_EPR} are highlighted in blue.
Now, consider the correlations in $\ket{\mathrm{C}^{n-1}Z}$:  for any two qubits $\{k,l\}\subset[n]$ of $\ket{\mathrm{C}^{n-1}Z}$, we have $\langle X_kX_l\rangle:=\bra{\mathrm{C}^{n-1}Z}X_kX_l\ket{\mathrm{C}^{n-1}Z}=1-2^{2-n}$ and  $\langle X_k\rangle=\langle X_l\rangle=1-2^{2-n}$.
That is, any two-qubit subsystem has $(b,c)=(1-2^{2-n},1-2^{2-n})$, which lie on the orange lines and moves towards the top-right corner as $n$ increases in Fig.~\ref{fig:bc_correlation}.  
It is straightforward to see that for any $K$, sufficiently large $n$ takes $(1-2^{2-n},1-2^{2-n})$ outside the feasible region.
It follows that the corresponding two-qubit reduced states cannot be obtained from finitely many EPR pairs by local operations. By Theorem~\ref{thm:SRM_2qubits_cor_by_EPRs}, this indicates that $\ket{\mathrm{C}^{n-1}Z}$ has LRM.

\begin{figure}[t]
\centering
\includegraphics[width=0.48\textwidth]{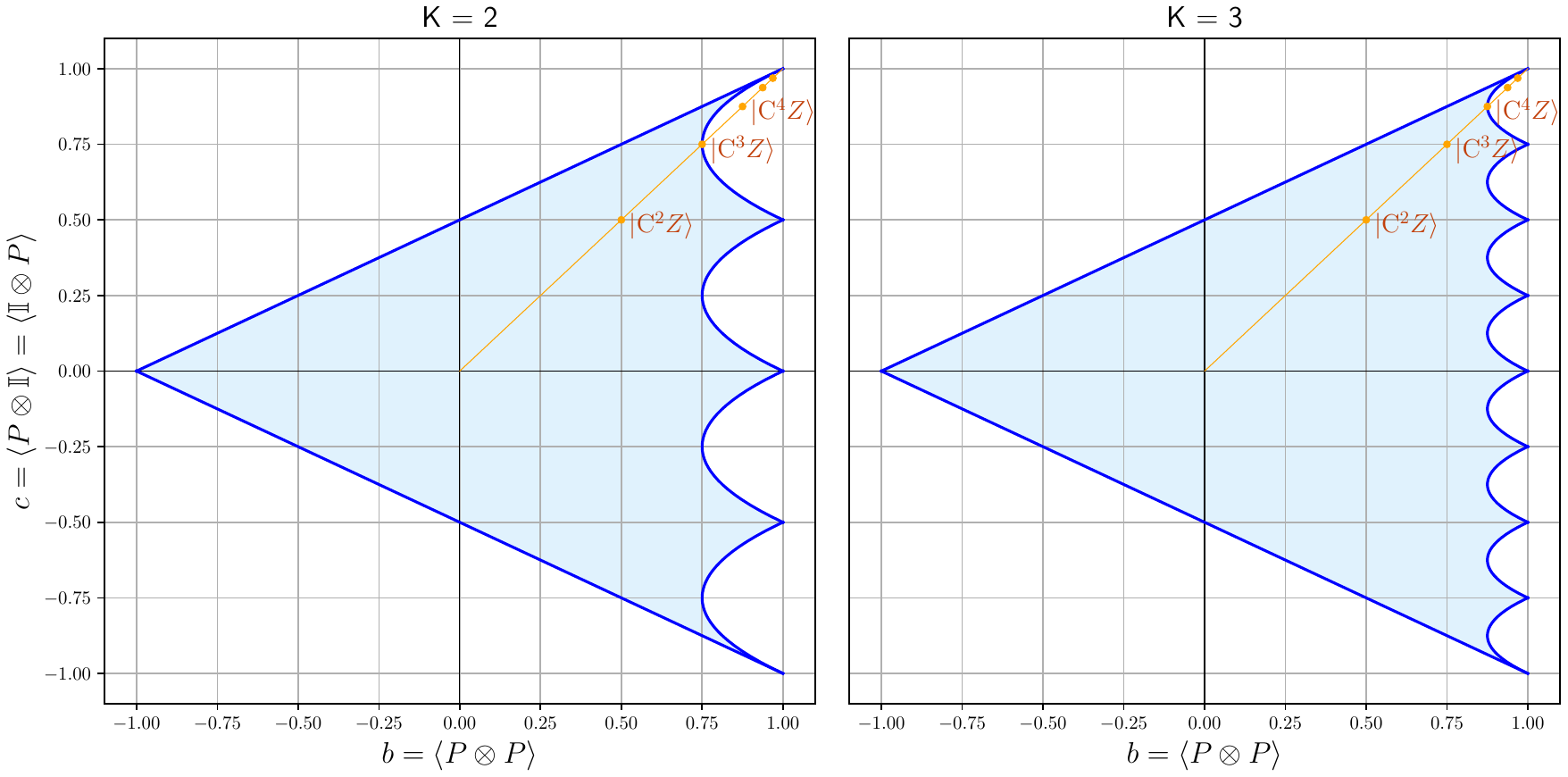}
\caption{Diagnosing LRM with two-point correlation functions. 
Here we plot the $K=2,3$ cases for illustration; the general cases are analogous.
$P\in\{X,Y,Z\}$ is a single-qubit Pauli operator.
The blue areas enclosed by dark blue boundaries are ``feasible regions" of $(b,c)$ given by Proposition~\ref{prop:correlations_EPR} in Appendix~\ref{app:LRM_from_correlation} (the boundaries consist of two linear functions and $2^K$ quadratic functions):
if a two-qubit state $\rho$ has $(b,c)$ lying outside the feasible region for $K$, then $\rho$ cannot be obtained by local operations on $K$ EPR pairs.
For the $\ket{\mathrm{C}^{n-1}Z}$ example,
choosing $P=X$, the corresponding $(b,c)=(1-2^{2-n},1-2^{2-n})$ lie on the orange line $b=c$ and exit the feasible region for any fixed $K$ upon increasing $n$, signifying LRM.}
\label{fig:bc_correlation}
\end{figure}

Another simple example is the generalized GHZ state family $\left\{\sqrt{1-\alpha}\ket{0^n}+\sqrt{\alpha}e^{i\theta}\ket{1^n}\right\}$. By straightforward calculation, $\langle Z_kZ_l\rangle=1$ and $\langle Z_k\rangle=1-2\alpha$ for all positions $\{k,l\}\subset[n]$ and all values of $\theta$. When $1-2\alpha$ is not a finite-length binary decimal, $(1,1-2\alpha)$ falls outside the feasible region (also see Fig.~\ref{fig:bc_correlation} for illustration) for arbitrary $K$, implying that the corresponding family exhibits LRM.

\section{Concluding remarks}

Driven by the envisaged importance of studying complex quantum many-body systems through a computational lens, we developed the concept of LRM and extensively explored its properties and manifestations leveraging connections with diverse areas including quantum error-correcting codes and condensed matter.  We showcased the fundamental importance of LRM across quantum computing, complexity theory, and many-body physics, unveiling various substantive avenues for enriching and advancing these fields and their interactions via LRM theory.

There are various compelling directions for future research. 
From the perspective of LRM state and phase theory itself, potentially important future work includes e.g.~developing quantitative characterizations or signatures (a natural possibility being magic variants of topological entanglement entropies), setting explicit bounds on the circuit depth required to erase LRM in various scenarios, exploring the mixed-state cases, and further looking into computational features of phases of matter.
Also, the question of the general possibility of learning SRM states with polynomial computational complexity~\cite{Huang2024Learning,Landau2024Learning} is important but remains not fully understood.
Moreover, the question of how LRM relates to nonstabilizerness distribution holds both theoretical and practical interest and merits further study, which we briefly discuss in Appendix~\ref{app:LRM_and_nonlocalMagic}.
Other directions worth pursuing include: relation with measurement-assisted complexity of many-body systems~\cite{PRXQuantum.4.020339,gu2024magicinducedcomputationalseparationentanglement,du2025spacetimequantumcircuitcomplexity}; influence on nonstabilizerness noise thresholds~\cite{wei2024noise}; further applications to QEC and fault tolerance via e.g., logical gates and computational complexity; a deeper understanding of the $n^{-1}$ critical error scaling, potentially in connection with approximate QEC scenarios where similar phenomena emerge in different fundamental ways~\cite{faist20,liu2022approximate,Yi_2024}.

Looking ahead, our LRM paradigm is expected to provide new lens for further understanding the interplay between quantum complexity and physics, which has already sparked notable interest from other perspectives~\cite{gu2024magicinducedcomputationalseparationentanglement,schuster-low}.
Further, we anticipate such unified perspectives of computation, information and physics will continue to foster novel insights and even unforeseen paradigm shifts.

\begin{acknowledgments}
We thank Anurag Anshu, Soonwon Choi, Tyler Ellison, Andi Gu, Tobias Haug, Yu-Jie Liu, Zimu Li, Zhengwei Liu, Zhenhuan Liu, Natalie Parham, Yuguo Shao, Jinmin Yi and Zishuo Zhao for valuable discussion and feedback.
F.W.\ is supported by BMSTC and ACZSP (Grant No.~Z221100002722017).
Z.-W.L.\ is supported in part by NSFC under Grant No.~12475023, Dushi Program, and a startup funding from YMSC. 
We note an independent work~\cite{korbany2025longrangenonstabilizernessphasesmatter} released contemporaneously with our study, which considers relevant notions of LRM but adopts different perspectives and focuses on 1D matrix product states.
\end{acknowledgments}


%

\appendix
\onecolumngrid

\section{Strong LRM from limitations of fault-tolerant logical gates on quantum codes}\label{app:LRM_TSC}

We define two edges on a lattice (more generally, a graph) to be adjacent if they share a common vertex. We define the distance between two edges to be the length of the shortest path on the lattice connecting the two edges.
For a set of edges $E$, its \emph{diameter} $\operatorname{Diam}(E)$ is defined as the smallest integer $\delta$ such that every pair of edges in $E$ is within distance $\delta$ of each other.
For an operator $O$ acting on qubits associated with the edges, we also denote $\operatorname{Diam}(O):=\operatorname{Diam}(\supp(O))$, where $\supp(O)$ is the set of qubits on which $O$ acts nontrivially.
Note that the ``diameter" of an operator on a lattice is well defined even when the lattice cannot be embedded into any finite-dimensional manifold.
Important cases include expander graphs and  complete graphs.

\begin{fact}\label{fact:pure_locally_same_then_same_Evalue}
Suppose $n$-qubit pure states $\ket{\psi_1}$ and $\ket{\psi_2}$ satisfy $\Tr_{\overline{A}}\ketbra{\psi_1}{\psi_1}=\Tr_{\overline{A}}\ketbra{\psi_2}{\psi_2}$ for some region $A\subset[n]$, where $\overline{A}:=[n]-A$ is the complement of $A$, and observable $O$ satisfies $\supp(O)\subset A$, then $\bra{\psi_1}O\ket{\psi_1}=\bra{\psi_2}O\ket{\psi_2}$.
\end{fact}
\begin{proof}
Since $\Tr_{\overline{\supp(O)}}\ketbra{\psi_1}{\psi_1}=\Tr_{\overline{\supp(O)}}\ketbra{\psi_2}{\psi_2}$, we know
\begin{align}
\bra{\psi_1}O\ket{\psi_1}=&\Tr(\ketbra{\psi_1}{\psi_1}O)=\Tr\left(\Tr_{\overline{\supp(O)}}(\ketbra{\psi_1}{\psi_1})O|_{\supp(O)}\right)\\
=&\Tr\left(\Tr_{\overline{\supp(O)}}(\ketbra{\psi_2}{\psi_2})O|_{\supp(O)}\right)=\Tr(\ketbra{\psi_2}{\psi_2}O)=\bra{\psi_2}O\ket{\psi_2},
\end{align}
where $O|_{\supp(O)}$ denotes the restriction of $O$ to its support.
\end{proof}

\begin{lem}\label{lemma:small_energy_ground_dim}
Suppose $H$ is a Hamiltonian acting on a finite-dimensional Hilbert space $\mathcal{H}$, with eigenvalues $E_0 < E_1 < E_2 < \cdots$. Let $D_0$ be the degeneracy of the ground-state space with energy $E_0$. 
For an orthonormal state set $\{\ket{\phi_1},\cdots,\ket{\phi_t}\}$, if
\begin{equation}
\bra{\phi_i} H \ket{\phi_i}<E_0+\frac{E_1-E_0}{D_0+1}
\end{equation}
for all $i=1,\cdots,t$, then $t \le D_0$.
\end{lem}
\begin{proof}
Denote by $\Pi_0,\Pi_1,\Pi_2,\cdots$  projectors onto the eigenspaces of $H$ corresponding to $E_0,E_1,E_2,\cdots$. For $i=1,\cdots,t$ we have
\begin{align}
\bra{\phi_i}H\ket{\phi_i}=&E_0\bra{\phi_i}\Pi_0\ket{\phi_i}+\sum_{i\ge1} E_i \bra{\phi_i}\Pi_i\ket{\phi_i}\\
\ge& E_0\bra{\phi_i}\Pi_0\ket{\phi_i}+ E_1 \sum_{i\ge1}\bra{\phi_i}\Pi_i\ket{\phi_i}\\
=&E_0\bra{\phi_i}\Pi_0\ket{\phi_i}+E_1\bra{\phi_i}(\mbb{I}-\Pi_0)\ket{\phi_i}\\
=&E_1-(E_1-E_0)\bra{\phi_i}\Pi_0\ket{\phi_i}.
\end{align}
Combining with $\bra{\phi_i}H\ket{\phi_i}< E_0+\frac{E_1-E_0}{D_0+1}$ we obtain $\bra{\phi_i}\Pi_0\ket{\phi_i}>\frac{D_0}{D_0+1}$. Denote $\Pi'=\sum_{i=1}^t\ketbra{\phi_i}{\phi_i}$, we have
\begin{equation}
\Tr(\Pi'\Pi_0)=\sum_{i=1}^t\bra{\phi_i}\Pi_0\ket{\phi_i}>\frac{tD_0}{D_0+1}.
\end{equation}
On the other hand, since
\begin{equation}
\Tr(\Pi'\Pi_0)\le\norm{\Pi'}_\infty\norm{\Pi_0}_1=D_0,
\end{equation}
we know $t<D_0+1$, that is, $t\le D_0$.
\end{proof}

\subsection{Proof of Theorem~\ref{thm:LRM_from_TSC}}\label{app:proof_of_LRM_from_TSC}

We restate Theorem~\ref{thm:LRM_from_TSC} here:

\begin{thm}\label{thm:LRM_from_TSC_app}
For any $D$-dimensional ($D\ge2$) TSC family and a $k$-qubit state $\ket{\phi}$, if
\begin{equation}
\ket{\phi}\notin\big\{V^\dagger\ket{0^k}\mid V\in \mc{C}_k^{(D)}\big\},
\end{equation}
where $k$ is the number of logical qubits     ($\mc{C}_k^{(D)}:=\{U\in\operatorname{U}(2^k)\mid UPU^\dagger\in\mc{C}_k^{(D-1)},\forall P\in\mc{P}_k\}$, with the first level $\mc{C}_k^{(1)}:=\mc{P}_k$ given by the multi-qubit Pauli group, is the $D$-th level of the Clifford hierarchy on $k$ qubits), then the corresponding logical state $\ket{\overline{\phi}}$ exhibits strong LRM.
\end{thm}

\begin{proof}[Proof sketch for the $D=2$ case]
Suppose for the sake of contradiction that a shallow circuit (family) $U$ can map a logical nonstabilizer state $\ket{\overline{\phi}}$ to a stabilizer state $\ket{S}$ that possibly lies outside the code space. In this scenario, the action of $U$ deforms the original TSC to a new subspace that contains the stabilizer state $\ket{S}$.

As a critical step, we deduce that any such deformed topological code containing a stabilizer state must itself be a TSC.
Since the code is obtained from the original one by a shallow circuit, it also has growing distance (equivalently, all of its code states are locally indistinguishable), and can be described as the ground space of a local Hamiltonian. Therefore, the deformed code space is exactly the set of states that are locally indistinguishable from $\ket{S}$. We then prove that this space coincides with a stabilizer code: it is precisely the common $+1$ eigenspace of the stabilizer group generated by all local Pauli strings that stabilize $\ket{S}$.

With this in mind, $U$ acts as a shallow circuit connecting the original TSC and a deformed TSC such that the logical magic state $\ket{\overline{\phi}}$ is mapped to a state $\ket{S}$, which constitutes a logical stabilizer state in the deformed code. The impossibility of this mapping is a consequence of fundamental restrictions on shallow circuits acting between 2D TSCs, as formulated in the context of fault-tolerant logical gates by the Bravyi–König theorem~\cite{Bravyi2013Classification}. Specifically, for 2D TSCs, any shallow circuit can only induce a logical Clifford operation between the codes, making it impossible to transform $\ket{\overline{\phi}}$ (nonstabilizer) into $\ket{S}$ (stabilizer). We therefore conclude that $\ket{\overline{\phi}}$ must exhibit LRM.

We can further make all the steps above error-robust to establish the strong LRM of $\ket{\overline{\phi}}$.
\end{proof}

\begin{proof}[Proof of Theorem~\ref{thm:LRM_from_TSC_app}]
Suppose, on the contrary, that the logical state family $\{\ket{\overline{\phi}_n}\}$, where $n$ denotes the number of qubits, does not have strong LRM. 
Then there exists a sequence $T\subset\mbb{Z}_{>0}$, a family of shallow circuits $\{U_n\}_{n\in T}$, a family of stabilizer states $\{\ket{S_n}\}_{n\in T}$, such that for $n\in T$ we have
\begin{equation}
\big\|U_n\ketbra{\overline{\phi}_n}{\overline{\phi}_n}U_n^\dagger-\ketbra{S_n}{S_n}\big\|_1=o(1/n).
\end{equation}

Denote projectors onto the code space of this family of TSC as $\{\Pi^{(n)}\}$.
Denote by $\{Q_1^{(n)},\cdots,Q_{n-k}^{(n)}\}$ a set of local Pauli stabilizer generators associated with this  TSC family.
Let $d_n$ be the code distance of $\{\Pi^{(n)}\}$, which satisfies $d_n\rightarrow\infty$ as $n\rightarrow\infty$.
Let $\eta$ be a sufficiently large constant that upper bound  the light-cone sizes of $\{U_n\}$.

Let $C$ be a constant independent of $n$ that upper bounds the diameters of all deformed generators $U_nQ_i^{(n)}U_n^\dagger$.
Define $\mc{B}_n$ as the set of states that are $C$-locally indistinguishable from $U_n\ket{\overline{\phi}_n}$:
\begin{equation}
\mc{B}_n:=\Big\{\ket{\psi}\Big|\Tr_{\overline{J}}\ketbra{\psi}{\psi}=\Tr_{\overline{J}}(U_n\ketbra{\overline{\phi}_n}{\overline{\phi}_n}U_n^\dagger),\forall J\subset[n]\text{ with }\operatorname{Diam}(J)\le C\Big\}.
\end{equation}
For an arbitrary $\ket{\psi}$ satisfying $\bra{\psi}U_n\Pi^{(n)}U_n^\dagger\ket{\psi}=1$, $U_n^\dagger\ket{\psi}$ is a logical state of $\Pi_n$ and thus indistinguishable from $\ket{\overline{\phi}_n}$ when restricted to regions of size $\le d_n-1$. Since the light-cone sizes of $U_n$ are upper bounded by $\eta$, we know $\ket{\psi}$ is indistinguishable from $U_n\ket{\overline{\phi}_n}$ when restricted to regions of size $\le \lfloor(d_n-1)/\eta\rfloor$. $d_n\rightarrow\infty$ implies $\ket{\psi}\in\mc{B}_n$ when $n\in T$ is sufficiently large.
On the other hand, since all $\ket{\psi}\in\mc{B}_n$ is $C$-locally indistinguishable from $U_n\ket{\overline{\phi}_n}$ and the diameters of $U_nQ_i^{(n)}U_n^\dagger$'s are upper bounded by $C$, by Fact~\ref{fact:pure_locally_same_then_same_Evalue} we know 
\begin{equation}
\bra{\psi}U_nQ_i^{(n)}U_n^\dagger\ket{\psi}=\bra{\overline{\phi}_n}U_n^\dagger U_nQ_i^{(n)}U_n^\dagger U_n\ket{\overline{\phi}_n}=\bra{\overline{\phi}_n}Q_i^{(n)}\ket{\overline{\phi}_n}=1
\end{equation}
for $i=1,\cdots,n-k$, implying $\bra{\psi}U_n\Pi^{(n)}U_n^\dagger\ket{\psi}=1$.
Therefore, $\mc{B}_n$ is exactly the set of pure states in the deformed code $U_n\Pi^{(n)}U_n^\dagger$:
\begin{equation}
\mc{B}_n=\big\{\ket{\psi}\mid\bra{\psi}U_n\Pi^{(n)}U_n^\dagger\ket{\psi}=1\big\}.
\end{equation}

Define $\mc{A}_n$ as the set of states that are $C$-locally indistinguishable from $\ket{S_n}$:
\begin{equation}
\mc{A}_n:=\Big\{\ket{\psi}\Big|\Tr_{\overline{J}}\ketbra{\psi}{\psi}=\Tr_{\overline{J}}\ketbra{S_n}{S_n},\forall J\subset[n]\text{ with }\operatorname{Diam}(J)\le C\Big\}.
\end{equation}
Denote by $G_n\subset\mc{P}_n$  the stabilizer group of $\ket{S_n}$.
Let $r_n$ be the maximum number of independent Pauli strings in $\{g\in G_n\mid\operatorname{Diam}(g)\le C\}$. Pick such an independent set $\{P_1^{(n)},\cdots,P_{r_n}^{(n)}\}\subset\{g\in G_n\mid\operatorname{Diam}(g)\le C\}$. We must have
\begin{equation}\label{eq:P_1_to_P_r_coverall}
\{g\in G_n\mid\operatorname{Diam}(g)\le C\}\subset\langle P_1,\cdots,P_r\rangle,
\end{equation}
because otherwise we would be able to find $r+1$ independent Pauli strings in $\{g\in G_n\mid\operatorname{Diam}(g)\le C\}$.
Let
\begin{equation}
\mc{A}_n':=\big\{\ket{\psi}\mid P_i^{(n)}\ket{\psi}=\ket{\psi},i=1,\cdots,r_n\big\}=\big\{\ket{\psi}\mid\bra{\psi}\widetilde{\Pi}^{(n)}\ket{\psi}=1\big\},
\end{equation}
where $\widetilde{\Pi}^{(n)}:=\prod_{i=1}^{r_n}\frac{\mbb{I}+P_i^{(n)}}{2}$.

Now we show that $\mc{A}_n\subset\mc{A}_n'$, and $\mc{A}_n$ contains a set of $2^{n-r_n}$ orthonormal pure states (actually $\mc{A}_n=\mc{A}_n'$, as we will show later).
For $\ket{a}\in\mc{A}_n$, since $\operatorname{Diam}(P_i^{(n)})\le C$, by Fact~\ref{fact:pure_locally_same_then_same_Evalue} we know $\bra{a}P_i^{(n)}\ket{a}=\bra{S_n}P_i^{(n)}\ket{S_n}=1$ for $i=1,\cdots,r_n$. This implies $\ket{a}\in\mc{A}_n'$ and thus $\mc{A}_n\subset\mc{A}_n'$.

We take $\{P_{r_n+1}^{(n)},\cdots,P_{n}^{(n)}\}\subset G_n$ such that $\{P_1^{(n)},\cdots,P_n^{(n)}\}$ are independent. We have $\langle P_1^{(n)},\cdots,P_n^{(n)}\rangle=G_n$.
Now we omit the upper script $(n)$ for $P_i^{(n)}$ for brevity.
For $j=r_n+1,\cdots,n$, since $P_j\notin\langle P_1,\cdots,P_{r_n}\rangle$, by \eqref{eq:P_1_to_P_r_coverall} we know $\operatorname{Diam}(P_j)> C$. For $\mathbf{b}=(b_{r_n+1},\cdots,b_n)\in\{0,1\}^{n-r_n}$, consider the states
\begin{align}
\ketbra{\psi_{\mathbf{b}}^{(n)}}{\psi_{\mathbf{b}}^{(n)}}&:=\prod_{i=1}^{r_n}\frac{\mbb{I}+P_i}{2}\prod_{j=r_n+1}^n\frac{\mbb{I}+(-1)^{b_j}P_j}{2}\\
&=\frac{1}{2^n}\sum_{\mathbf{x}\in\{0,1\}^n}\prod_{i=1}^{r_n}P_i^{x_i}\prod_{j=r_n+1}^n((-1)^{b_j}P_j)^{x_j},
\end{align}
where $\mathbf{x}=(x_1,\cdots,x_n)$. These states are mutually orthogonal. In particular, we have $\ket{\psi_{\mathbf{0}}^{(n)}}=\ket{S_n}$.
Consider an arbitrary subset $J\subset[n]$ with $\operatorname{Diam}(J)\le C$.
For $\mathbf{x}\in\{0,1\}^n$ such that $(x_{r_n+1},\cdots,x_n)\neq0^{n-r_n}$, we must have 
\begin{equation}
\operatorname{Diam}\left(\prod_{i=1}^{r_n}P_i^{x_i}\prod_{j=r_n+1}^n((-1)^{b_j}P_j)^{x_j}\right)=\operatorname{Diam}\left(\prod_{i=1}^{r_n}P_i^{x_i}\prod_{j=r_n+1}^nP_j^{x_j}\right)> C,
\end{equation}
otherwise $\prod_{i=1}^{r_n}P_i^{x_i}\prod_{j=r_n+1}^nP_j^{x_j}\in\{g\in G_n\mid\operatorname{Diam}(g)\le C\}\subset\langle P_1,\cdots,P_{r_n}\rangle$, which would contradict the independence of $\{P_1,\cdots,P_n\}$. Therefore, for arbitrary $\mathbf{b}\in\{0,1\}^{n-r_n}$, when $(x_{r_n+1},\cdots,x_n)\neq0^{n-r_n}$ we have
\begin{equation}
\begin{aligned}
\Tr_{\overline{J}}\left(\prod_{i=1}^{r_n}P_i^{x_i}\prod_{j=r_n+1}^n((-1)^{b_j}P_j)^{x_j}\right)=0.
\end{aligned}
\end{equation}
Denoting $\mathbf{y}=(y_1,\cdots,y_{r_n})$, for all $\mathbf{b}\in\{0,1\}^{n-r_n}$ we have
\begin{align}
\Tr_{\overline{J}}\ketbra{\psi^{(n)}_{\mathbf{b}}}{\psi^{(n)}_{\mathbf{b}}}=&\frac{1}{2^n}\sum_{\mathbf{x}\in\{0,1\}^n}\Tr_{\overline{J}}\left(\prod_{i=1}^{r_n}P_i^{x_i}\prod_{j=r_n+1}^n((-1)^{b_j}P_j)^{x_j}\right)\\
=&\frac{1}{2^n}\sum_{\mathbf{y}\in\{0,1\}^{r_n}}\Tr_{\overline{J}}\left(\prod_{i=1}^{r_n}P_i^{y_i}\prod_{j=r_n+1}^n((-1)^{b_j}P_j)^{0}\right)\\
=&\frac{1}{2^n}\sum_{\mathbf{y}\in\{0,1\}^{r_n}}\Tr_{\overline{J}}\left(\prod_{i=1}^{r_n}P_i^{y_i}\prod_{j=r_n+1}^nP_j^{0}\right)\\
=&\frac{1}{2^n}\sum_{\mathbf{x}\in\{0,1\}^n}\Tr_{\overline{J}}\left(\prod_{i=1}^{r_n}P_i^{x_i}\prod_{j=r_n+1}^nP_j^{x_j}\right)=\Tr_{\overline{J}}\ketbra{S}{S}.
\end{align}
Therefore, we find a set of $2^{n-r_n}$ orthonormal pure states $\{\ket{\psi^{(n)}_{\mathbf{b}}}\mid\mathbf{b}\in\{0,1\}^{n-r_n}\}\subset\mc{A}_n$. By $\mc{A}_n\subset\mc{A}_n'$ we know $\dim(\operatorname{span}(\mc{A}_n))=\dim(\operatorname{span}(\mc{A}_n'))=2^{n-r_n}$.

For the original TSC $\Pi^{(n)}$, define Hamiltonian $H^{(n)}=\sum_{i=1}^{n-k}\frac{\mbb{I}-Q_i^{(n)}}{2}$, which has eigenvalues $\{0,1,\cdots,n-k\}$ and ground space projector $\Pi^{(n)}$.
For the subspace $\widetilde{\Pi}^{(n)}$ corresponding to $\mc{A}_n'$, we define $\widetilde{H}^{(n)}=\sum_{j=1}^{r_n}\frac{\mbb{I}-P^{(n)}_j}{2}$, which has eigenvalues $\{0,1,\cdots,r_n\}$ and ground space projector $\widetilde{\Pi}^{(n)}$.
Intuitively, $U_n\ket{\overline{\phi}_n}\approx\ket{S_n}$ implies $U_n\Pi^{(n)}U_n^\dagger\approx\widetilde{\Pi}^{(n)}$.

We first show that $r_n=n-k$ when $n\in T$ is sufficiently large. For all $\mathbf{b}\in\{0,1\}^{n-r_n}$, since $\ket{\psi^{(n)}_{\mathbf{b}}}\in\mc{A}_n$ and $\operatorname{Diam}(U_nQ_i^{(n)}U_n^\dagger)\le C$, we know
\begin{equation}
\bra{\psi^{(n)}_{\mathbf{b}}}U_nH^{(n)}U_n^\dagger\ket{\psi^{(n)}_{\mathbf{b}}}=\sum_{i=1}^{n-k}\frac{1-\bra{\psi^{(n)}_{\mathbf{b}}}U_nQ_i^{(n)}U_n^\dagger\ket{\psi^{(n)}_{\mathbf{b}}}}{2}=\sum_{i=1}^{n-k}\frac{1-\bra{S_n}U_nQ_i^{(n)}U_n^\dagger\ket{S_n}}{2}=\bra{S_n}U_nH^{(n)}U_n^\dagger\ket{S_n},
\end{equation}
therefore
\begin{align}
\abs{\bra{\psi_{\mathbf{b}}^{(n)}}U_nH^{(n)}U_n^\dagger\ket{\psi^{(n)}_{\mathbf{b}}}-\bra{\overline{\phi}_n} H^{(n)}\ket{\overline{\phi}_n}}=&\abs{\bra{S_n}U_nH^{(n)}U_n^\dagger\ket{S_n}-\bra{\overline{\phi}_n} H^{(n)}\ket{\overline{\phi}_n}}\\
=&\abs{\Tr\big[H^{(n)}\big(U_n^\dagger\ketbra{S_n}{S_n}U_n-\ketbra{\overline{\phi}_n}{\overline{\phi}_n}\big)\big]}\\
\le&\sum_{i=1}^{n-k}\abs{\Tr\Big[\frac{\mbb{I}-Q_i^{(n)}}{2}\big(U_n^\dagger\ketbra{S_n}{S_n}U_n-\ketbra{\overline{\phi}_n}{\overline{\phi}_n}\big)\Big]}\\
\le&\sum_{i=1}^{n-k}\big\|\frac{\mbb{I}-Q_i^{(n)}}{2}\big\|_\infty\big\|\ketbra{S_n}{S_n}-U_n\ketbra{\overline{\phi}_n}{\overline{\phi}_n}U_n^\dagger\big\|_1\\
\le&(n-k)o(1/n)=o(1).
\end{align}
Therefore, for arbitrary $\mathbf{b}\in\{0,1\}^{n-r_n}$ we have
\begin{equation}
\bra{\psi^{(n)}_\mathbf{b}}U_nH^{(n)}U_n^\dagger\ket{\psi^{(n)}_\mathbf{b}}\le\bra{\overline{\phi}_n}H^{(n)}\ket{\overline{\phi}_n}+o(1)=o(1).
\end{equation}
Since $k$ is a constant independent of $n$, by Lemma~\ref{lemma:small_energy_ground_dim} we know that $2^{n-r_n}\le2^{k}$ when $n\in T$ is sufficiently large.
For arbitrary $\ket{b}\in\mc{B}_n$,
\begin{align}
\abs{\bra{b}\widetilde{H}^{(n)}\ket{b}-\bra{S_n}\widetilde{H}^{(n)}\ket{S_n}}=&\abs{\bra{\overline{\phi}_n}U_n^\dagger\widetilde{H}^{(n)}U_n\ket{\overline{\phi}_n}-\bra{S_n}\widetilde{H}^{(n)}\ket{S_n}}\\
=&\abs{\Tr\big[\widetilde{H}^{(n)}\big(U_n\ketbra{\overline{\phi}_n}{\overline{\phi}_n}U_n^\dagger-\ketbra{S_n}{S_n}\big)\big]}\\
\le&\sum_{i=1}^{r_n}\abs{\Tr\Big[\frac{\mbb{I}-P_i^{(n)}}{2}\big(U_n\ketbra{\overline{\phi}_n}{\overline{\phi}_n}U_n^\dagger-\ketbra{S_n}{S_n}\big)\Big]}\\
\le&\sum_{i=1}^{r_n}\big\|U_n\ketbra{\overline{\phi}_n}{\overline{\phi}_n}U_n^\dagger-\ketbra{S_n}{S_n}\big\|_1\\
\le&r_no(1/n)=o(1).
\end{align}
Therefore, for arbitrary $\ket{b}\in\mc{B}_n$ we have 
\begin{equation}\label{eq:b_energy_vanishingsmall}
\bra{b}\widetilde{H}^{(n)}\ket{b}\le\bra{S_n}\widetilde{H}^{(n)}\ket{S_n}+o(1)=o(1).
\end{equation}
We already know the ground space dimension $2^{n-r_n}$ of $\widetilde{H}^{(n)}$ is upper bounded by $2^k$. When $n\in T$ is sufficiently large, $\bra{b}\widetilde{H}^{(n)}\ket{b}\le\frac{1}{2^k+1}\le0+\frac{1-0}{2^{n-r_n}+1}$. By Lemma~\ref{lemma:small_energy_ground_dim} we know $2^k\le2^{n-r_n}$. We conclude that $r_n=n-k$ when $n\in T$ is sufficiently large.

Now we show $U_n\Pi^{(n)}U_n^{\dagger}\approx\widetilde{\Pi}^{(n)}$. 
Since $\widetilde{H}^{(n)}\ge\mbb{I}-\widetilde{\Pi}^{(n)}$, for arbitrary $\ket{b}\in\mc{B}_n$,
\begin{equation}
\bra{b}\widetilde{H}^{(n)}\ket{b}
\ge\bra{b}(\mbb{I}-\widetilde{\Pi}^{(n)})\ket{b}=1-\bra{b}\widetilde{\Pi}^{(n)}\ket{b}.
\end{equation}
By \eqref{eq:b_energy_vanishingsmall} we know
\begin{equation}
\bra{b}\widetilde{\Pi}^{(n)}\ket{b}\ge1-\bra{b}\widetilde{H}^{(n)}\ket{b}=1-o(1).
\end{equation}
We take $\{\ket{b_1},\cdots,\ket{b_{2^k}}\}$ be a set of orthogonal basis of $\operatorname{span}(\mc{B}_n)$, satisfying $\sum_{i=1}^{2^k}\ketbra{b_i}{b_i}=U_n\Pi^{(n)}U_n^\dagger$, then
\begin{equation}
\Tr\big(U_n\Pi^{(n)}U_n^\dagger\widetilde{\Pi}^{(n)}\big)=\sum_{i=1}^{2^k}\bra{b_i}\widetilde{\Pi}^{(n)}\ket{b_i}\ge2^k(1-o(1)).
\end{equation}
Therefore we have
\begin{align}
\big\|U_n\Pi^{(n)}U_n^{\dagger}-\widetilde{\Pi}^{(n)}\big\|_2^2=&\Tr\left(\big(U_n\Pi^{(n)}U_n^{\dagger}-\widetilde{\Pi}^{(n)}\big)^2\right)\\
=&\Tr\big(U_n\Pi^{(n)}U_n^{\dagger}\big)+\Tr\big(\widetilde{\Pi}^{(n)}\big)-2\Tr\big(U_n\Pi^{(n)}U_n^\dagger\widetilde{\Pi}^{(n)}\big)\\
\le&2^k+2^k-2\cdot2^k(1-o(1))=2^{k+1}o(1)=o(1),
\end{align}
implying $\big\|U_n\Pi^{(n)}U_n^{\dagger}-\widetilde{\Pi}^{(n)}\big\|_2=o(1)$.

Denote $\widetilde{d}_n$ be the distance of the stabilizer code $\widetilde{\Pi}^{(n)}$. Recall that $d_n$ is the distance of $\Pi_n$, and $\eta$ upper bounds the light-cone sizes of $\{U_n\}$.
Therefore, the deformed code $U_n\Pi^{(n)}U_n^\dagger$ has distance $\ge\lfloor\frac{d_n-1}{\eta}\rfloor+1$.
Now we show that when $n\in T$ is sufficiently large, $\widetilde{\Pi}^{(n)}$ is a code with distance
\begin{equation}
\widetilde{d}_n\ge\lfloor\frac{d_n-1}{\eta}\rfloor+1,
\end{equation}
which tends to infinity as $n\rightarrow\infty$. We remark that this implies all states in $\mc{A}_n'$ are $C$-locally indistinguishable, thus we have $\mc{A}_n=\mc{A}_n'$.

Suppose on the contrary that there exists a subsequence $T'\subset T$, such that $\widetilde{d}_n\le\lfloor\frac{d_n-1}{\eta}\rfloor$ for all $n\in T'$. 
Denote $K_n$ be the stabilizer group corresponding to $\widetilde{\Pi}^{(n)}$.
Then there exists a family of Pauli strings $\{E_n\}_{n\in T'}$, where $E_n\in\mc{P}_n$ has phase $+1$ and satisfies $\wt(E_n)=\widetilde{d}_n\le\lfloor\frac{d_n-1}{\eta}\rfloor$, such that $E_n$ is a logical error for $\widetilde{\Pi}^{(n)}$. This implies that $\mbb{I}_2^{\otimes n},i\mbb{I}_2^{\otimes n},-\mbb{I}_2^{\otimes n},-i\mbb{I}_2^{\otimes n}\notin\{E_nh\mid h\in K_n\}$.
We have
\begin{equation}
\widetilde{\Pi}^{(n)}E_n\widetilde{\Pi}^{(n)}=E_n\widetilde{\Pi}^{(n)} \text{, and } \Tr(E_n\widetilde{\Pi}^{(n)})=\frac{1}{2^{n-k}}\sum_{h\in K_n}\Tr(E_nh)=0.
\end{equation}
Since $U_n\Pi^{(n)}U_n^\dagger$ has distance $\ge\lfloor\frac{d_n-1}{\eta}\rfloor+1$, we know there exist a constant $c_n\in\mbb{R}$ with $\abs{c_n}\le1$ such that
\begin{equation}
(U_n\Pi^{(n)}U_n^\dagger)E_n(U_n\Pi^{(n)}U_n^\dagger)=c_nU_n\Pi^{(n)}U_n^\dagger.
\end{equation}
Therefore, we have
\begin{align}
&\bignorm{E_n\widetilde{\Pi}^{(n)}-c_nU_n\Pi^{(n)}U_n^\dagger}_2\\
=&\bignorm{\widetilde{\Pi}^{(n)}E_n\widetilde{\Pi}^{(n)}-(U_n\Pi^{(n)}U_n^\dagger)E_n(U_n\Pi^{(n)}U_n^\dagger)}_2\\
\le&\bignorm{\widetilde{\Pi}^{(n)}E_n\widetilde{\Pi}^{(n)}-\widetilde{\Pi}^{(n)}E_n(U_n\Pi^{(n)}U_n^\dagger)}_2+\bignorm{\widetilde{\Pi}^{(n)}E_n(U_n\Pi^{(n)}U_n^\dagger)-(U_n\Pi^{(n)}U_n^\dagger)E_n(U_n\Pi^{(n)}U_n^\dagger)}_2\\
\le&\bignorm{\widetilde{\Pi}^{(n)}}_\infty\bignorm{E_n}_\infty\bignorm{\widetilde{\Pi}^{(n)}-U_n\Pi^{(n)}U_n^\dagger}_2+\bignorm{\widetilde{\Pi}^{(n)}-U_n\Pi^{(n)}U_n^\dagger}_2\bignorm{E_n}_\infty\bignorm{U_n\Pi^{(n)}U_n^\dagger}_\infty\\
=&2\bignorm{\widetilde{\Pi}^{(n)}-U_n\Pi^{(n)}U_n^\dagger}_2=o(1).
\end{align}
This implies
\begin{align}
&\bignorm{c_nE_n\widetilde{\Pi}^{(n)}-\widetilde{\Pi}^{(n)}}_2=\bignorm{c_n\widetilde{\Pi}^{(n)}-E_n\widetilde{\Pi}^{(n)}}_2\\
\le&\bignorm{c_n\widetilde{\Pi}^{(n)}-c_nU_n\Pi^{(n)}U_n^\dagger}_2+\bignorm{c_nU_n\Pi^{(n)}U_n^\dagger-E_n\widetilde{\Pi}_n}_2=o(1).
\end{align}
However, this is impossible since $\Tr(c_nE_n\widetilde{\Pi}^{(n)})=0$ and $\Tr(\widetilde{\Pi}^{(n)})=2^k$. Therefore, we conclude that $\widetilde{d}_n\ge\lfloor\frac{d_n-1}{\eta}\rfloor+1$ when $n\in T$ is sufficiently large.

When $n\in T$ is sufficiently large, the stabilizer code family $\widetilde{\Pi}^{(n)}$ has the following properties:
\begin{enumerate}
\item It is defined on the same manifold as the original TSC $\Pi^{(n)}$.
\item It has local Pauli generators $P^{(n)}_1,\cdots,P^{(n)}_{n-k}$, with diameters $\le C$.
\item Its distance can be arbitrarily large when $n\rightarrow\infty$.
\end{enumerate}
Therefore, $\{\widetilde{\Pi}^{(n)}\}$ is a TSC family.

Denote $\Pi^{(n)}=J_1^{(n)}J_1^{(n)\dagger}$ and $\widetilde{\Pi}^{(n)}=J_2^{(n)}J_2^{(n)\dagger}$, where $J_1^{(n)}$ and $J_2^{(n)}$ are Clifford encoding isometries.
We have $\ket{\overline{\phi}_n}=J_1^{(n)}\ket{\phi}$.
The constant depth local circuit $U_n$ maps the original TSC $\Pi^{(n)}$ to $U_n\Pi^{(n)}U_n^\dagger$, which is near the newly defined TSC $\widetilde{\Pi}^{(n)}$: $\bignorm{U_n\Pi^{(n)}U_n^\dagger-\widetilde{\Pi}^{(n)}}_2=o(1)$.
By Theorem~\ref{thm:robust_BK}, the robust form of Bravyi--König theorem, there exists a sequence of $k$-qubit unitaries $\{V^{(n)}\}\subset\mc{C}_k^{(D)}$ such that the induced linear transform $J_2^{(n)\dagger} U_nJ_1^{(n)}$ (not a unitary in general) satisfies
\begin{equation}
\big\|J_2^{(n)\dagger} U_nJ_1^{(n)}-V^{(n)}\big\|_2=o(1).
\end{equation}

Therefore, we have
\begin{align}
&\big\|J_2^{(n)\dagger} U_nJ_1^{(n)}\ketbra{\phi}{\phi}\big(J_2^{(n)\dagger} U_nJ_1^{(n)}\big)^\dagger-V^{(n)}\ketbra{\phi}{\phi}V^{(n)\dagger}\big\|_1\\
\le&2\bignorm{J_2^{(n)\dagger} U_nJ_1^{(n)}-V^{(n)}}_1\le2\cdot2^{k/2}\bignorm{J_2^{(n)\dagger} U_nJ_1^{(n)}-V^{(n)}}_2=o(1).
\end{align}
Also,
\begin{equation}
\bignorm{J_2^{(n)\dagger} U_nJ_1^{(n)}\ketbra{\phi}{\phi}\big(J_2^{(n)\dagger} U_nJ_1^{(n)}\big)^\dagger-J_2^{(n)\dagger}\ketbra{S_n}{S_n}J_2^{(n)}}_1=\big\|U_n\ketbra{\overline{\phi}_n}{\overline{\phi}_n}U_n^\dagger-\ketbra{S_n}{S_n}\big\|_1=o(1/n).
\end{equation}
Therefore,
\begin{equation}\label{eq:stab_and_nonstab_near}
\big\|J_2^{(n)\dagger}\ketbra{S_n}{S_n}J_2^{(n)}-V^{(n)}\ketbra{\phi}{\phi}V^{(n)\dagger}\big\|_1\le o(1)+o(1/n)=o(1).
\end{equation}
$J_2^{(n)\dagger}\ket{S_n}$ is a $k$-qubit stabilizer state.
Note that $\ket{\phi}\notin\mc{C}_k^{(D)\dagger}\ket{0^k}:=\big\{W^\dagger\ket{0^k}\mid W\in \mc{C}_k^{(D)}\big\}$.
From Ref.~\cite{Zeng2008SemiClifford}, we know $\mc{C}_k^{(D)\dagger}\mc{C}^{(2)}_k=\mc{C}_k^{(D)\dagger}$, thus $\ket{\phi}\notin\mc{C}_k^{(D)\dagger}\ket{0^k}=\mc{C}_k^{(D)\dagger}\mc{C}^{(2)}_k\ket{0^k}$, implying $V\ket{\phi}\notin\mc{C}^{(2)}_k\ket{0^k}$. Therefore, $V\ket{\phi}$ is a $k$-qubit nonstabilizer state, implying that \eqref{eq:stab_and_nonstab_near} is impossible.

Therefore, we conclude that $\{\ket{\overline{\phi}_n}\}$ has strong LRM.
\end{proof}

\subsection{Robust extension of the Bravyi--König theorem}
\label{app:robust_bk}

Given a matrix $A$, we denote $\norm{A}:=\norm{A}_2=\sqrt{\Tr(A^\dagger A)}$.
For convenience, we define
\begin{equation}
\mc{C}_k^{(1)}:=\{e^{i\theta}P\mid \theta\in\mbb{R},P\in\mc{P}_k\}    
\end{equation}
in this section, which does not change the definition of $\mc{C}_k^{(j)}$ for $j\ge2$.

\begin{lem}\label{lem:nearly_commute_sequence}
For square matrices $A,B_1,\cdots,B_m$ of the same size, if for $j=1,\cdots,m$ we have $\bignorm{[A,B_j]}\le\epsilon_j$ and $\norm{B_j}_\infty\le1$, then
\begin{equation}
\bignorm{[A,B_1\cdots B_m]}\le\epsilon_1+\cdots+\epsilon_m.
\end{equation}
\end{lem}
\begin{proof}
We have
\begin{align}
\bignorm{[A,B_1\cdots B_m]}=&\bignorm{AB_1\cdots B_m-B_1\cdots B_mA}\\
\le&\sum_{i=1}^m\bignorm{B_1\cdots B_{i-1} (AB_i)B_{i+1}\cdots B_m-B_1\cdots B_{i-1} (B_iA)B_{i+1}\cdots B_m}\\
=&\sum_{i=1}^m\bignorm{B_1\cdots B_{i-1} (AB_i-B_iA)B_{i+1}\cdots B_m}\\
\le&\sum_{i=1}^m\norm{B_1}_\infty\cdots\norm{B_{i-1}}_\infty \norm{AB_i-B_iA}\norm{B_{i+1}}_\infty\cdots \norm{B_{m}}_\infty\\
\le&\epsilon_1+\cdots+\epsilon_m.
\end{align}
\end{proof}

\begin{lem}\label{lemma:approx-(anti)commute_with_Pauli_then_Pauli}
Suppose a unitary $M$ acts on $k$ qubits.
If for some  sufficiently small $\epsilon\ge0$,
\begin{equation}
\max_{P\in\mc{G}_k}~\min\left\{\bignorm{[M,P]}^2,~\bignorm{\{M,P\}}^2\right\}\le\epsilon,
\end{equation}
where $\mc{G}_k:=\{X_1,\cdots,X_k,Z_1,\cdots,Z_k\}$, then there exists a $Q\in\mc{C}_k^{(1)}$ such that
\begin{equation}
\norm{M-Q}^2\le2^{-1}k\epsilon+2^{-k-2}k^2\epsilon^2=O(\epsilon).
\end{equation}
\end{lem}
\begin{proof}
Decompose $M$ as $M=\sum_{R\in\mc{P}_k^+}c_RR$, where $\mc{P}_k^+$ denotes the set of $4^k$ $k$-qubit Pauli strings with phase $+1$, and $c_R=\frac{1}{2^k}\Tr(RM)\in\mbb{C}$.
Since $M$ is unitary, we have $\Tr(M^\dagger M)=\Tr(\mbb{I}_{2^k})=2^k$, implying $\sum_{R\in\mc{P}_k^+}\abs{c_R}^2=1$. For arbitrary $P\in\mc{P}_k$,
\begin{align}
[M,P]&=\sum_{R\in\mc{P}_k^+}c_R[R,P]=\sum_{R\in\mc{P}_k^+,RP=-PR}2c_RRP,\\
\{M,P\}&=\sum_{R\in\mc{P}_k^+}c_R\{R,P\}=\sum_{R\in\mc{P}_k^+,RP=PR}2c_RRP.
\end{align}
Therefore we have
\begin{align}
\bignorm{[M,P]}^2&=\Tr([M,P]^\dagger[M,P])=2^k\sum_{R\in\mc{P}_k^+,RP=-PR}4\abs{c_R}^2,\\
\bignorm{\{M,P\}}^2&=\Tr(\{M,P\}^\dagger\{M,P\})=2^k\sum_{R\in\mc{P}_k^+,RP=PR}4\abs{c_R}^2.
\end{align}
Consider $Q'\in\mc{P}_k^+$ defined as following: For $i=1,\cdots,k$,
\begin{enumerate}
\item If $\bignorm{[M,X_i]}^2\le\epsilon$ and $\bignorm{[M,Z_i]}^2\le\epsilon$, let $Q'|_i=\mbb{I}$,
\item If $\bignorm{[M,X_i]}^2\le\epsilon$ and $\bignorm{\{M,Z_i\}}^2\le\epsilon$, let $Q'|_i=X$,
\item If $\bignorm{\{M,X_i\}}^2\le\epsilon$ and $\bignorm{[M,Z_i]}^2\le\epsilon$, let $Q'|_i=Z$,
\item If $\bignorm{\{M,X_i\}}^2\le\epsilon$ and $\bignorm{\{M,Z_i\}}^2\le\epsilon$, let $Q'|_i=Y$,
\end{enumerate}
where $X_i$ is the $k$-qubit Pauli string which is $X$ on the $i$-th qubit and $\mbb{I}$ elsewhere, and $Q'|_i$ is the 
$i$-th Pauli operator of $Q'$.
For case 2, $\bignorm{[M,X_i]}^2\le\epsilon$ and $\bignorm{\{M,Z_i\}}^2\le\epsilon$, we have
\begin{align}
\bignorm{[M,X_i]}^2&=2^k\sum_{R\in\mc{P}_k^+,RX_i=-X_iR}4\abs{c_R}^2=2^{k+2}\sum_{R\in\mc{P}_k^+,R|_i=Y,Z}\abs{c_R}^2\le\epsilon,\\
\bignorm{\{M,Z_i\}}^2&=2^k\sum_{R\in\mc{P}_k^+,RZ_i=Z_iR}4\abs{c_R}^2=2^{k+2}\sum_{R\in\mc{P}_k^+,R|_i=I,Z}\abs{c_R}^2\le\epsilon,
\end{align}
thus $\sum_{R\in\mc{P}_k^+,R|_i\neq X}\abs{c_R}^2=\sum_{R\in\mc{P}_k^+,R|_i=I,Y,Z}\abs{c_R}^2\le2\cdot2^{-k-2}\epsilon=2^{-k-1}\epsilon$. Other cases are similar, and we obtain
\begin{equation}
\sum_{R\in\mc{P}_k^+,R|_i\neq Q'|_i}\abs{c_R}^2\le2^{-k-1}\epsilon.
\end{equation}
Notice that $\{R\in\mc{P}_k^+\mid R\neq Q'\}=\bigcup_{i=1}^k\big\{R\in\mc{P}_k^+\big|~ R|_i\neq Q'|_i\big\}$, thus
\begin{equation}
\sum_{R\in\mc{P}_k^+,R\neq Q'}\abs{c_R}^2\le\sum_{i=1}^k\sum_{R\in\mc{P}_k^+,R|_i\neq Q'|_i}\abs{c_R}^2\le2^{-k-1}k\epsilon.
\end{equation}
Therefore, $\abs{c_{Q'}}^2\ge1-2^{-k-1}k\epsilon$.
We write $c_{Q'}=e^{i\theta}\abs{c_{Q'}}$, and take $Q=e^{i\theta}Q'\in\mc{C}_k^{(1)}$.
For $\epsilon\ge0$ sufficiently small we $0\le1-2^{-k-1}k\epsilon\le1$, so $\abs{c_{Q'}}\ge\sqrt{1-2^{-k-1}k\epsilon}\ge1-2^{-k-1}k\epsilon$, implying $\abs{c_{Q'}-e^{i\theta}}^2=(1-\abs{c_{Q'}})^2\le(2^{-k-1}k\epsilon)^2$. Since $M-Q=\sum_{R\in\mc{P}_k^+}(c_R-e^{i\theta}\delta_{R=Q'})R$, we conclude that
\begin{align}
\norm{M-Q}^2=&2^k\sum_{R\in\mc{P}_k^+}\abs{c_R-e^{i\theta}\delta_{R=Q'}}^2\\
=&2^k\sum_{R\in\mc{P}_k^+,R\neq Q'}\abs{c_R}^2+2^k\abs{c_{Q'}-e^{i\theta}}^2\\
\le&2^{-1}k\epsilon+2^{-k-2}k^2\epsilon^2.
\end{align}
\end{proof}

\begin{lem}\label{lemma:climb_approx_CliffordHierachy}
Suppose $M$ is a unitary matrix acting on $k$ qubits, and $j\ge1$. If for some  sufficiently small $\delta\ge0$, and for every $P\in\mc{P}_k$, there exists $A_P\in\mc{C}^{(j)}_k$ such that $\norm{MPM^\dagger-A_P}\le\delta$, then there exists $B\in\mc{C}^{(j+1)}_{k}$, such that
\begin{equation}
\norm{M-B}=O(\delta).
\end{equation}
\end{lem}
\begin{proof}
For arbitrary $P,Q\in\mc{P}_k$ we have
\begin{align}
\norm{A_PA_Q-MPQM^\dagger}&\le\norm{A_PA_Q-A_PMQM^\dagger}+\norm{A_PMQM^\dagger-MPM^\dagger MQM^\dagger}\\
&\le\bignorm{A_P}_\infty\bignorm{A_Q-MQM^\dagger}+\bignorm{A_P-MPM^\dagger}\bignorm{MQM^\dagger}_\infty\le2\delta.
\end{align}
If $[P,Q]=0$, then
\begin{equation}
\norm{A_PA_Q-A_QA_P}\le\norm{A_PA_Q-MPQM^\dagger}+\norm{MPQM^\dagger-MQPM^\dagger}+\norm{MQPM^\dagger-A_QA_P}\le4\delta.
\end{equation}
If $\{P,Q\}=0$, then
\begin{equation}
\bignorm{A_{P}A_{Q}+A_QA_P}\le\norm{A_PA_Q-MPQM^\dagger}+\norm{MPQM^\dagger+MQPM^\dagger}+\norm{A_QA_P-MQPM^\dagger}\le4\delta.
\end{equation}

For all $j\ge1$, $\mc{C}^{(j)}_k$ is finite up to global phases. Therefore, the sets
\begin{equation}
\big\{\norm{C_1C_2-C_2C_1}\mid C_1,C_2\in\mc{C}^{(j)}_k\big\}\text{ and }\big\{\norm{C_1C_2+C_2C_1}\mid C_1,C_2\in\mc{C}^{(j)}_k\big\}
\end{equation}
are finite subsets of $\mbb{R}$.
Therefore, when $\delta$ is sufficiently small,
\begin{align}
[P,Q]&=0\Rightarrow\norm{A_PA_Q-A_QA_P}=0\iff[A_P,A_Q]=0,\\
\{P,Q\}&=0\Rightarrow\norm{A_PA_Q+A_QA_P}=0\iff\{A_P,A_Q\}=0.
\end{align}

Denote $\mc{G}_k:=\{X_1,\cdots,X_k,Z_1,\cdots,Z_k\}$, which generates $\mc{P}_k$ up to phases.
For all $P\in\mc{G}_k$, $MPM^\dagger$ has eignevalues half $+1$ half $-1$.
Since $\mc{C}^{(j)}_k$ is a set of unitaries which is finite up to global phases, when $\delta$ is sufficiently small, $\norm{MPM^\dagger-A_P}\le\delta$ ensures that $A_P\in\mc{C}^{(j)}_k$ must equal to some matrix $A_P'\in\mc{C}^{(j)}_k$ which has eigenvalues half $+1$ half $-1$, up to a small global phase, otherwise $\norm{MPM^\dagger-A_P}$ cannot be arbitrarily small.
We write $A_P'=e^{-i\theta_P}A_P$, with small $\abs{\theta_P}$. $A_P'$ is unitary and has real eigenvalues, thus it is Hermitian.

Since the eigenvalues of $A_P$ are half $e^{i\theta_P}$ half $-e^{i\theta_P}$, by the Hoffman--Wielandt inequality~\cite{hoffman2003variation} we know $\norm{MPM^\dagger-A_P}^2\ge 2^k\abs{1-e^{i\theta_P}}^2$, thus we have $\abs{1-e^{i\theta_P}}\le2^{-k/2}\delta$.

Now we obtain $\{A_P'\mid P\in\mc{G}_k\}\subset\mc{C}^{(j)}_k$, a set of ``Pauli like" matrices, satisfying
\begin{enumerate}
\item For all $P,Q\in\mc{G}_k$, $[P,Q]=0\iff[A_P',A_Q']=0$ \text{, and } $\{P,Q\}=0\iff\{A_P',A_Q'\}=0$;
\item For all $P\in\mc{G}_k$, $A_P'$ is a Hermitian matrix satisfying $(A_P')^2=\mbb{I}$.
\end{enumerate}

By these relations, we now show that there exists a unitary $B$ such that $BPB^\dagger=A_P'$, for all $P\in\mc{G}_k$.

$\mc{G}_k$ is a minimal generating set for the matrix algebra $M_n(\mbb{C})$ by multiplication and linear combination, with relations: \textbf{I.} $P^2=\mbb{I},\forall P\in\mc{G}_k$; \textbf{II.} $[Z_i,Z_j]=0,[X_i,X_j]=0,\forall i,j\in[k]$; \textbf{III.} $[Z_i,X_j]=0,\forall i\neq j$; \textbf{IV.} $\{Z_i,X_i\}=0,\forall i\in[k]$.
Since $\{A_P'\mid P\in\mc{G}_k\}\subset\mc{C}^{(j)}_k\subset M_n(\mbb{C})$ is another set of generators with the same relations, we know it also generates $M_n(\mbb{C})$, and the map $\varphi:\mc{G}_k\rightarrow\{A_P'\mid P\in\mc{G}_k\}$ defined by $\varphi(P)=A_P'$ induces an automorphism $\varphi:M_n(\mbb{C})\rightarrow M_n(\mbb{C})$ of $M_{2^k}(\mathbb{C})$.

By the Skolem--Noether theorem (see e.g.~Ref.~\cite{lorenz2007algebra}), we know that every automorphism of the matrix algebra $M_n(\mbb{C})$ is inner (i.e., has the form $\varphi(x)=u^{-1}xu$). Therefore there exists an invertible matrix $\widetilde{B}\in M_n(\mbb{C})$ such that $\varphi(\cdot)=\widetilde{B}\cdot\widetilde{B}^{-1}$.

For all $P\in\mc{G}_k$, since $(A_P')^\dagger=A_P'$, we have
\begin{equation}
(\widetilde{B}^\dagger \widetilde{B})P(\widetilde{B}^\dagger \widetilde{B})^{-1}=\widetilde{B}^\dagger(\widetilde{B}P\widetilde{B}^{-1})(\widetilde{B}^{\dagger})^{-1}=\widetilde{B}^{\dagger}A_P'(\widetilde{B}^{\dagger})^{-1}=(\widetilde{B}^{-1}A_P'\widetilde{B})^{\dagger}=P^\dagger=P.
\end{equation}
Therefore, for all $P\in\mc{G}_k$ we have $(\widetilde{B}^\dagger \widetilde{B})P=P(\widetilde{B}^\dagger \widetilde{B})$, implying $\widetilde{B}^\dagger \widetilde{B}=b\mbb{I}$ for some $b>0$.
Define $B=b^{-1/2}\widetilde{B}$, we have $B^\dagger B=\mbb{I}$, and $\varphi(\cdot)=\widetilde{B}\cdot\widetilde{B}^{-1}=B\cdot B^{-1}=B\cdot B^{\dagger}$.
Therefore $BPB^\dagger=A_P'$ for all $P\in\mc{G}_k$.

Now we show that $B\in\mc{C}^{(j+1)}_k$.
For arbitrary $R\in\mc{P}_k$, write $R=e^{i\varphi}P_1\cdots P_s$ where $P_1,\cdots,P_s\in\mc{G}_k$ and $\theta\in\mbb{R}$.
We have
\begin{equation}
BRB^\dagger=e^{i\varphi}(BP_1B^\dagger)\cdots(BP_sB^\dagger)=e^{i\varphi}A_{P_1}'\cdots A_{P_s}'.
\end{equation}
There exists $A_R\in\mc{C}^{(j)}_k$ close to $MRM^\dagger$, thus
\begin{align}
\norm{BRB^\dagger-A_R}\le&\norm{BRB^\dagger-MRM^\dagger}+\norm{MRM^\dagger-A_R}\\
=&\norm{A_{P_1}'\cdots A_{P_s}'-A_{P_1}\cdots A_{P_s}}+\norm{MRM^\dagger-A_R}\\
\le&\sum_{i=1}^s\abs{1-e^{i\theta_{P_i}}}\norm{\mbb{I}}+\delta\le k2^{-k/2}\delta2^{k/2}+\delta=(k+1)\delta.
\end{align}
Since $BRB^\dagger=e^{i\varphi}A_{P_1}'\cdots A_{P_s}'\in\big(\mc{C}^{(j)}_k\big)^s$ and $A_R\in\mc{C}^{(j)}_k$, where $\big(\mc{C}^{(j)}_k\big)^s=\mc{C}^{(j)}_k\cdots\mc{C}^{(j)}_k$ contains $\mc{C}^{(j)}_k$ and is a finite set up to global phases, for sufficiently small $\delta$, $BRB^\dagger$ must equal to $A_R$ up to global phases. By $A_R\in\mc{C}^{(j)}_k$ we know $BRB^\dagger\in\mc{C}^{(j)}_k$.
Therefore we have $B\mc{P}_kB^\dagger\subset\mc{C}^{(j)}_k$, that is, $B\in\mc{C}^{(j+1)}_k$.

Now we upper bound the distance between $M$ and $B$. For any $P\in\mc{G}_k$,
\begin{align}
\bignorm{(M^\dagger B)P-P(M^\dagger B)}=\bignorm{BPB^\dagger-MPM^\dagger}=\bignorm{A_P'-MPM^\dagger}\\
\le\bignorm{A_P'-A_P}+\bignorm{A_P-MPM^\dagger}\le\delta+\delta=2\delta.
\end{align}
By Lemma~\ref{lemma:approx-(anti)commute_with_Pauli_then_Pauli}, we conclude that
\begin{equation}
\bignorm{B-M}=\bignorm{M^\dagger B-\mbb{I}}=O(\delta).
\end{equation}
\end{proof}

\begin{thm}[Robust Bravyi--König]\label{thm:robust_BK}
Suppose two families of $D$-dimensional ($D\ge2$) TSCs are defined on the same lattice $\Lambda^{(n)}$, both have $k$ logical qubits, with code projectors $\Pi_1^{(n)}=J_1^{(n)}J_1^{(n)\dagger}$ and $\Pi_2^{(n)}=J_2^{(n)}J_2^{(n)\dagger}$ respectively, where $n$ is the number of physical qubits and $J_1^{(n)},J_2^{(n)}$ are Clifford encoding isometries.
If a family of shallow circuit unitaries $\{U^{(n)}\}$ satisfies $\big\|U^{(n)}\Pi_1^{(n)}U^{(n)\dagger}-\Pi_2^{(n)}\big\|_2=o(1)$, then there exist $\{Q^{(n)}\}\subset\mc{C}_{k}^{(D)}$ such that the induced linear transform $J_2^{(n)\dagger} U^{(n)}J_1^{(n)}$ satisfies
$\big\|J_2^{(n)\dagger} U^{(n)}J_1^{(n)}-Q^{(n)}\big\|_2=o(1)$.
\end{thm}

\begin{proof}
We omit superscripts $(n)$ for brevity and implicitly consider such families in the proof.

As in Ref.~\cite{Bravyi2013Classification} (arXiv version, page 5), we partition the lattice $\Lambda$ into $D+1$ regions $\Lambda=\bigcup_{j=1}^{D+1}\Lambda_j$. Here $\Lambda_j$ corresponds to $(j-1)$-dimensional simplices in a triangulation of $\Lambda$. Each region $\Lambda_j$ is a disjoint union of chunks of size $O(1)$ separated by distance $\Omega(1)$.
Each region $\Lambda_j$ is correctable. 

 Also take $n$-qubit Pauli strings $P_1,P_2,\cdots,P_D$ such that $[P_1,\Pi_1]=0$ and $[P_i,\Pi_2]=0$ for $i=2,\cdots,D$. We can ensure that $\supp(P_j)$ does not overlap with $\Lambda_j$ together with its $\rho$-neighborhood for some constant $\rho$. We can take $P_1,P_2,\cdots,P_D$ such that $J_1^\dagger P_1J_1$ and $J_2^\dagger P_2J_2,\cdots,J_2^\dagger P_DJ_2$ be arbitrary elements in $\mc{P}_k$.

Define $K_1=UP_1U^\dagger$, and $K_j=P_j^\dagger K_{j-1}P_jK_{j-1}^\dagger$ for $j=2,\cdots,D$. Same as the proof in Ref.~\cite{Bravyi2013Classification}, we know $K_j$ acts trivially on $\bigcup_{i=1}^j\Lambda_i$. Therefore, $K_D$ is supported on the correctable region $\Lambda_{D+1}$.

Since $\bignorm{\Pi_1U^\dagger-U^\dagger\Pi_2}=\bignorm{U\Pi_1-\Pi_2U}=\bignorm{U\Pi_1U^\dagger-\Pi_2}=o(1)$, we have 
\begin{align}
\bignorm{[\Pi_2,K_1]}=&\norm{\Pi_2UP_1U^\dagger-UP_1U^\dagger\Pi_2}\\
\le&\norm{\Pi_2UP_1U^\dagger-U\Pi_1P_1U^\dagger}+\norm{UP_1\Pi_1U^\dagger-UP_1U^\dagger\Pi_2}\\
\le&\bignorm{\Pi_2U-U\Pi_1}\bignorm{P_1U^\dagger}_\infty+\bignorm{UP_1}_\infty\bignorm{\Pi_1U^\dagger-U^\dagger\Pi_2}=o(1).
\end{align}
For $j=2,\cdots,D$, there are $2^{j-1}$ $K_1$'s in $K_j$. By Lemma~\ref{lem:nearly_commute_sequence},
\begin{equation}
\bignorm{[K_j,\Pi_2]}\le2^{j-1}\bignorm{[K_1,\Pi_2]}=o(1).
\end{equation}

Since $K_D$ is a correctable error, there exists $c\in\mbb{C}$ with $\abs{c}\le1$ such that $\Pi_2 K_D\Pi_2=c\Pi_2$.
Now we show that $c\approx\pm1$ and $K_D\Pi_2\approx\pm\Pi_2$. We have
\begin{align}
\norm{K_D\Pi_2-c\Pi_2}=&\norm{K_D\Pi_2\Pi_2-\Pi_2 K_D\Pi_2}\le\norm{K_D\Pi_2-\Pi_2 K_D}\norm{\Pi_2}_\infty=o(1),\\
\norm{\Pi_2K_D-c\Pi_2}=&\norm{\Pi_2\Pi_2K_D-\Pi_2 K_D\Pi_2}\le\norm{\Pi_2}_\infty\norm{\Pi_2 K_D-K_D\Pi_2}=o(1).
\end{align}
As $P_D^2=e^{i\theta}\mbb{I}$, we have
\begin{align}
&\norm{\Pi_2-c^2\Pi_2}=\norm{(\Pi_2K_{D-1}P_DK_{D-1}^\dagger)(K_{D-1}P_DK_{D-1}^\dagger\Pi_2)-(cP_D\Pi_2)(cP_D\Pi_2)}\\
\le&\bignorm{\Pi_2K_{D-1}P_DK_{D-1}^\dagger-cP_D\Pi_2}\bignorm{K_{D-1}P_DK_{D-1}^\dagger\Pi_2}_\infty+\bignorm{cP_D\Pi_2}_\infty\bignorm{K_{D-1}P_DK_{D-1}^\dagger\Pi_2-cP_D\Pi_2}\\
\le&\bignorm{\Pi_2K_{D-1}P_DK_{D-1}^\dagger-cP_D\Pi_2}+\abs{c}\bignorm{K_{D-1}P_DK_{D-1}^\dagger\Pi_2-cP_D\Pi_2}\\
=&\norm{\Pi_2K_D-c\Pi_2}+\abs{c}\norm{K_D\Pi_2-c\Pi_2}=o(1).
\end{align}
Therefore, $\abs{1-c^2}\le\norm{\Pi_2}^{-1}o(1)=o(1)$, implying that either $\abs{c-1}=o(1)$ or $\abs{c+1}=o(1)$, and
\begin{equation}
\text{either }\norm{K_D\Pi_2-\Pi_2}=o(1)\text{ or }\norm{K_D\Pi_2+\Pi_2}=o(1).
\end{equation}
We write this as $\norm{K_D\Pi_2\pm\Pi_2}=o(1)$. By $\Pi_2=J_2J_2^\dagger$ we know that this is equivalent to
\begin{equation}
\bignorm{J_2^\dagger K_DJ_2\pm \mbb{I}}=o(1).
\end{equation}

Denote $\hat{U}=J_2^\dagger UJ_1$, $\hat{P}_1=J_1^\dagger P_1J_1$, and $\hat{P}_i=J_2^\dagger P_iJ_2$ for $i=2,\cdots,D$.
Note that $\hat{P}_1,\cdots,\hat{P}_D\in\mc{P}_k$, while $\hat{U}$ is a linear transform on $k$ qubits that is not a unitary in general.
Also define $k$-qubit linear transformations $\hat{K}_1=\hat{U}\hat{P}_1\hat{U}^\dagger$, and $\hat{K}_j=\hat{P}_j^\dagger\hat{K}_{j-1}\hat{P}_j\hat{K}_{j-1}^\dagger$ for $j=2,\cdots,D$. They are not unitaries in general.

Now we show that $\bignorm{\hat{K}_D-J_2^\dagger K_DJ_2}=o(1)$. We have
\begin{align}
\bignorm{\hat{K}_1-J_2^\dagger K_1J_2}=&\bignorm{J_2^\dagger UJ_1J_1^\dagger P_1J_1J_1^\dagger U^\dagger J_2-J_2^\dagger UP_1U^\dagger J_2}\\
=&\bignorm{J_2^\dagger (U\Pi_1)P_1(\Pi_1U^\dagger)J_2-J_2^\dagger(\Pi_2 U)P_1(U^\dagger \Pi_2)J_2}\\
\le&2\bignorm{U\Pi_1-\Pi_2U}=o(1)
\end{align}
Suppose for $i\ge1$ we have $\bignorm{\hat{K}_i-J_2^\dagger K_iJ_2}=o(1)$, then
\begin{align}
\bignorm{\hat{K}_{i+1}-J_2^\dagger K_{i+1}J_2}
=&\bignorm{J_2^\dagger P_{i+1}^\dagger J_2\hat{K}_iJ_2^\dagger P_{i+1} J_2\hat{K}_i^\dagger-J_2^\dagger P_{i+1}^\dagger K_iP_{i+1}K_i^\dagger J_2}\\
\le&\bignorm{J_2^\dagger P_{i+1}^\dagger J_2\hat{K}_iJ_2^\dagger P_{i+1} J_2\hat{K}_i^\dagger-J_2^\dagger P_{i+1}^\dagger J_2(J_2^\dagger K_iJ_2)J_2^\dagger P_{i+1} J_2(J_2^\dagger K_i^\dagger J_2)}\\
&+\bignorm{J_2^\dagger P_{i+1}^\dagger J_2(J_2^\dagger K_iJ_2)J_2^\dagger P_{i+1} J_2(J_2^\dagger K_i^\dagger J_2)-J_2^\dagger P_{i+1}^\dagger K_iP_{i+1}K_i^\dagger J_2}\\
\le&2\bignorm{\hat{K}_i-J_2^\dagger K_iJ_2}+\bignorm{J_2^\dagger P_{i+1}^\dagger K_iP_{i+1}(\Pi_2 K_i^\dagger) J_2-J_2^\dagger P_{i+1}^\dagger K_iP_{i+1}(K_i^\dagger\Pi_2) J_2}\\
\le&2\bignorm{\hat{K}_i-J_2^\dagger K_iJ_2}+\bignorm{\Pi_2 K_i^\dagger-K_i^\dagger\Pi_2}=o(1).
\end{align}
Therefore, we obtain $\bignorm{\hat{K}_D-J_2^\dagger K_DJ_2}=o(1)$ by induction, implying
\begin{equation}
\bignorm{\hat{K}_D\pm\mbb{I}}=o(1).
\end{equation}

We have
\begin{equation}
\bignorm{\hat{U}\hat{U}^\dagger-\mbb{I}}=\bignorm{J_2^\dagger (U\Pi_1) U^\dagger J_2-J_2^\dagger (\Pi_2U) U^\dagger J_2}\le\bignorm{U\Pi_1-\Pi_2U}=o(1),
\end{equation}
thus there exists a unitary matrix $\hat{U}'$ such that $\|\hat{U}-\hat{U}'\|=o(1)$. For $i=1,\cdots,D$, we define $\hat{K}_i'$ by replacing $\hat{U}$ (or $\hat{U}^\dagger$) in $\hat{K}_i$ by $\hat{U}'$ (or $\hat{U}'^\dagger$), thus $\hat{K}_i',i=1,\cdots,D$ are unitaries.
We have $\|\hat{K}_i-\hat{K}_i'\|=o(1)$ for all $i=1,\cdots,D$.

By $\bignorm{\hat{K}_D\pm\mbb{I}}=o(1)$ implies $\bignorm{\hat{K}_D'\pm\mbb{I}}=o(1)$, thus
\begin{equation}
\bignorm{\hat{P}_D^\dagger \hat{K}_{D-1}'\pm\hat{K}_{D-1}'\hat{P}_D^\dagger}=\bignorm{\hat{K}_D'\pm\mbb{I}}=o(1).
\end{equation}
Since $\hat{P}_D$ can be an arbitrary element in $\mc{P}_k$, by Lemma~\ref{lemma:approx-(anti)commute_with_Pauli_then_Pauli} we know there exists $Q_{1}\in\mc{C}_k^{(1)}$ such that
\begin{equation}
\bignorm{\hat{K}_{D-1}'-Q_{1}}=o(1).
\end{equation}
Now for $1\le j\le D-2$, suppose we have shown that there exists $Q_{j}\in\mc{C}_k^{(j)}$ such that $\bignorm{\hat{K}_{D-j}'-Q_j}=o(1)$, then
\begin{equation}
\bignorm{\hat{K}_{D-j-1}'\hat{P}_{D-j}\hat{K}_{D-j-1}'^\dagger-\hat{P}_{D-j}Q_j}=\bignorm{\hat{K}_{D-j}'-Q_j}=o(1).
\end{equation}
Since $\hat{P}_{D-j}$ can be an arbitrary element in $\mc{P}_k$, and $\hat{P}_{D-j}Q_j\in\mc{C}_k^{(1)}\mc{C}_k^{(j)}=\mc{C}_k^{(j)}$, by Lemma~\ref{lemma:climb_approx_CliffordHierachy} we know there exists $Q_{j+1}\in\mc{C}_k^{(j+1)}$ such that
\begin{equation}
\bignorm{\hat{K}_{D-j-1}'-Q_{j+1}}=o(1).
\end{equation}
In particular, there exists $Q_{D-1}\in\mc{C}_k^{(D-1)}$ such that 
\begin{equation}
\bignorm{\hat{U}'\hat{P}_1\hat{U}'^\dagger-Q_{D-1}}=\bignorm{\hat{K}_{1}'-Q_{D-1}}=o(1).
\end{equation}
Since $\hat{P}_{1}$ can be arbitrary element in $\mc{P}_k$, by Lemma~\ref{lemma:climb_approx_CliffordHierachy} we know that there exists $Q_{D}\in\mc{C}_k^{(D)}$ such that $\bignorm{\hat{U}'-Q_D}=o(1)$. Together with $\|\hat{U}-\hat{U}'\|=o(1)$ we obtain
\begin{equation}
\bignorm{J_2^\dagger UJ_1-Q_D}=\bignorm{\hat{U}-Q_D}=o(1).
\end{equation}
\end{proof}

\subsection{Example -- LRM family emerging from the toric code}
\label{app:t0}

In the following, we discuss in detail a concrete basic example of LRM family from the toric code, which may help consolidate the general proof intuition.

\begin{exmp}\label{thm:toric_code}
The logical nonstabilizer state $\ket{\overline{T0}}$ of the toric code has LRM.
\end{exmp}

As a fundamental feature, all code states in the toric code on a $L\times L$ lattice ($n=2L^2$ physical qubits, $2$ logical qubits) have long-range entanglement, requiring local circuits of depth $\Theta(L)$ to prepare from a product state~\cite{Topological2022Dennis,LiebRobinson2006Bravyi}.
Notice that while $\ket{\overline{T0}}$ has circuit complexity $\Theta(L)$, we cannot immediately rule out the possibility that it is mapped to a stabilizer state, thus losing all magic, by a shallow circuit.
The ``topological nonstabilizerness" of $\ket{\overline{T0}}$ can be visualized by its Pauli spectrum.
Recall that for any stabilizer state $\ket{\zeta}$ and Pauli string $P$ with $+1$ phase, we have $\bra{\zeta}P\ket{\zeta}=\pm1$ when $\pm P$ belongs to the stabilizer group of $\ket{\zeta}$, and $\bra{\zeta}P\ket{\zeta}=0$ otherwise.
Due to the global \emph{non-Pauli} stabilizer generator $\frac{1}{\sqrt{2}}(\overline{X}_1+\overline{Z}_1)$ for $\ket{\overline{T0}}$ (see Fig.~\ref{fig:logical_T}), we find $\bra{\overline{T0}}P\ket{\overline{T0}}=\pm\frac{1}{\sqrt{2}}$ for exponentially many $P$ with support wrapping around the torus. This suggests that shallow circuits on the torus cannot organize its Pauli spectrum to be $-1,0,1$ to match that of a stabilizer state.

\begin{figure}[t]
\centering
\includegraphics[width=0.45\textwidth]{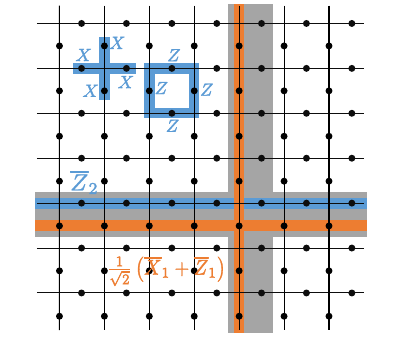}
\caption{{LRM from logical nonstabilizerness.} A set of $n$ stabilizer generators of the logical state $\ket{\overline{T0}}$ in toric code is shown, one of them is non-Pauli.
We define $\ket{T}:=\operatorname{cos}(\pi/8)\ket{0}+\operatorname{sin}(\pi/8)\ket{1}$, the $+1$ eigenstate of $\frac{1}{\sqrt{2}}(X+Z)$.
The Pauli generators (blue) consist of $n-2$ local ones and a global $\overline{Z}_2$, which is a tensor product of Pauli $Z$'s around the blue handle. The orange one is a global non-Pauli generator $\frac{1}{\sqrt{2}}(\overline{X}_1+\overline{Z}_1)$, which introduces global $\frac{1}{\sqrt{2}}$'s into the Pauli spectrum of $\ket{\overline{T0}}$, producing ``topological nonstabilizerness". Tracing out the entire region outside an arbitrary gray ``cross" will not decrease the amount of nonstabilizerness contained in $\ket{\overline{T0}}$.}
\label{fig:logical_T}
\end{figure}

\begin{proof}[Proof of the LRM-ness of $\ket{\overline{T0}}$ in toric code]
We prove by contradiction. 
Suppose on the contrary that $\ket{\overline{T0}}$ does not have LRM. Then, for an infinite family of lattice size $L$, there exists a shallow circuit $U$ such that $U\ket{\overline{T0}}=\ket{S}$, where $\ket{S}$ is a stabilizer state.
Since $U$ is shallow, there exists a constant $C$ (independent of $L$) that upper bounds the diameters of the $n-2$ deformed local generators $\{UA_vU^\dagger\}$, $\{UB_pU^\dagger\}$.
There are $3$ ``brother" states of $\ket{S}=U\ket{\overline{T0}}$, namely $U\overline{Y}_1\ket{\overline{T0}}$, $U\overline{Y}_2\ket{\overline{T0}}$, and $U\overline{Y}_1\overline{Y}_2\ket{\overline{T0}}$. These brother states and $\ket{S}$ are mutually orthogonal, and their reduced density matrices are identical when restricted to systems of diameter $\le C$.
At most $3$ such brother states of $\ket{S}$ exist since these states are all invariant under the action of $n-2$ deformed generators $\{UA_vU^\dagger\}$, $\{UB_pU^\dagger\}$.

Now we show that the deformed code space $U\Pi U^\dagger$ is also a TSC, where $\Pi$ denotes the projector onto the original toric code.
Note that the deformed generators $\{UA_vU^\dagger\}$, $\{UB_pU^\dagger\}$ may not be Pauli strings. We find the Pauli stabilizers of $U\Pi U^\dagger$ from the stabilizer group $G$ of $\ket{S}$.
We claim that there must exist $n-2$ independent Pauli strings in $G$ with diameters $\le C$, where ``independent" means they cannot multiply to identity.
Suppose, for example, that there exists and exists at most $n-3$ such local Pauli strings, then we can construct $7$ brother states for $\ket{S}$ as follows. 
We first find $3$ Pauli strings $Q_1,Q_2,Q_3$ in $G$, satisfying that $Q_1,Q_2,Q_3$ together with the $n-3$ local Pauli strings generate $G$. 
The support of $Q_1,Q_2,Q_3$ must have diameter $>C$ by assumption.
We can then examine the states stabilized by these $n-3$ local Pauli strings and $\pm Q_1,\pm Q_2,\pm Q_3$, get $7$ brother states for $\ket{S}$ and leading to a contradiction.
Therefore, we proved the claim that there are $n-2$ independent $C$-local Pauli strings that stabilize $\ket{S}$. 
Since $U$ is shallow, the deformed code states in $U\Pi U^\dagger$ are locally indistinguishable from $\ket{S}$ on any region of constant size, meaning that these $n-2$ local Pauli strings also fix other code states.
Additionally, the shallow circuit $U$ will degrade the distance of toric code, $L$, by at most a constant factor, so the deformed code family $U\Pi U^\dagger$ also has macroscopic distance, thereby proving that $U\Pi U^\dagger$ is a TSC.

Now we know that $U$ is a shallow circuit connecting the toric code and another TSC such that the logical magic state $\ket{\overline{T0}}$ is mapped to a stabilizer state $\ket{S}$, thus is a logical stabilizer state in the deformed TSC.
By the Bravyi–König theorem~\cite{Bravyi2013Classification}, we know $U$ is a logical Clifford gate, impossible to transform $\ket{\overline{T0}}$ into $\ket{S}$. Therefore, we conclude that $\ket{\overline{T0}}$ in toric code has LRM.
\end{proof}

We comment that LRM can be generated by only one non-Clifford gate: $\ket{\overline{T0}}$ can be obtained by first preparing the 2-qubit state $\ket{T0}$ with a non-Clifford gate $T$, and then perform the Clifford encoding circuit for toric code. 
We also notice an intriguing property of topological nonstabilizerness of TSCs: tracing out qubits from the union of several contractable regions on the torus (see Fig.~\ref{fig:logical_T}) does not reduce the amount of nonstabilizerness in $\ket{\overline{T0}}$, because the erasure of such regions is correctable by Clifford recovery operations.

\section{Application to quantum codes -- gate and symmetry testing}\label{app:logical_gate_testing}

As a concrete application of the connection we establish between magic complexity and fault-tolerant logical gates on quantum codes, we devise a simple logical gate/symmetry testing scheme, which is expected to be widely useful in QEC and many-body physics contexts. 
Specifically, we derive a concrete criterion for ruling out any transversal implementation of certain target logical gates (which also corresponds to onsite symmetries in physical contexts) that only involves simple calculations of Pauli expectation values of code states. Here we report further details for this result and provide two prominent use cases.

\subsection{General criterion}

Consider an $[\![n,k]\!]$ stabilizer code. If a ($k$-qubit) non-Clifford gate $U$ admits a transversal implementation, then for any state $\ket{\psi}\in\{U\ket{S}:\ket{S}\text{ a $k$-qubit stabilizer state}\}$, the corresponding logical state $\ket{\overline{\psi}}$ exhibits $\mathrm{SRM}_0$ (its magic can be erased by a layer of transversal unitaries).
By Lemma~\ref{lemma:pauli_support_invariance} below, for any region $R\subset[n]$, $\sum_{P:\supp(P)=R}|\bra{\overline{\psi}}P\ket{\overline{\psi}}|^2=\sum_{P:\supp(P)=R}|\bra{\overline{S}}P\ket{\overline{S}}|^2$, hence it is an integer. Consequently:

\begin{thm}[Transversal gate/symmetry testing criterion (formal)]\label{thm:transversal_testing_app}
Let $J:(\mathbb{C}^2)^{\otimes k}\to (\mathbb{C}^2)^{\otimes n}$ be a Clifford encoding isometry which defines an $[\![n,k]\!]$ stabilizer code.
Let $U$ be a $k$-qubit unitary. Suppose there exist a $k$-qubit stabilizer state $\ket{S}$ and a subset $R\subset[n]$ such that for $\ket{\psi}=U\ket{S}$ we have
\begin{equation}\label{eq:transversal_gate_testing_app}
\sum_{P\in\mc{P}_n^+:\,\supp(P)=R} 
\Tr\big(P\ketbra{\overline{\psi}}{\overline{\psi}}\big)^2
\notin \mathbb{Z},
\end{equation}
where $\ket{\overline{\psi}}:=J\ket{\psi}$ and $\mc{P}_n^+=\{\mbb{I},X,Y,Z\}^{\otimes n}$, then $\ket{\overline{\psi}}$ has $\mathrm{LRM}_0$, and there do not exist single-qubit unitaries $V_1,\cdots,V_n$ such that
\begin{equation}
(V_1\otimes\cdots\otimes V_n)J = e^{i\theta}JU,
\end{equation}
where $e^{i\theta}$ is a global phase.
\end{thm}

\begin{proof}
Suppose for the sake of contradiction that $(V_1\otimes\cdots\otimes V_n)J = e^{i\theta}JU$ for single-qubit unitaries $V_i$.
Then $\ket{\overline{\psi}}$ has $\mathrm{SRM}_0$, because
\begin{equation}
\ketbra{\overline{\psi}}{\overline{\psi}}=J U \ketbra{S}{S} U^\dagger J^\dagger
= (V_1\otimes\cdots\otimes V_n)J\ketbra{S}{S}J^\dagger(V_1\otimes\cdots\otimes V_n)^\dagger.
\end{equation}
By Lemma~\ref{lemma:pauli_support_invariance} below, we have
\begin{align}
\sum_{P\in\mc{P}_n^+:\,\supp(P)=R} 
\Tr(P\ketbra{\overline{\psi}}{\overline{\psi}})^2=&\sum_{P\in\mc{P}_n^+:\,\supp(P)=R}\Tr\Big(P(V_1\otimes\cdots\otimes V_n)J\ketbra{S}{S}J^\dagger(V_1\otimes\cdots\otimes V_n)^\dagger\Big)^2\\
=&\sum_{P\in\mc{P}_n^+:\,\supp(P)=R}\Tr(PJ\ketbra{S}{S}J^\dagger)^2,
\end{align}
which is an integer since $J\ket{S}$ is a stabilizer state (so $\bra{S}J^\dagger PJ\ket{S}\in\{-1,0,+1\}$). This contradicts \eqref{eq:transversal_gate_testing_app}.
\end{proof}

\begin{lem}\label{lemma:pauli_support_invariance}
For any $n$-qubit state $\rho$, any $n$-qubit unitary $U=U_1\otimes\cdots\otimes U_n$ which is a tensor product of single-qubit unitaries, and any subset $R\subset[n]$, we have
\begin{equation}
\sum_{P\in\mc{P}_n^+:\supp(P)=R}\Tr(P\rho)^2=\sum_{P\in\mc{P}_n^+:\supp(P)=R}\Tr(PU\rho U^\dagger)^2.
\end{equation}
\end{lem}

\begin{proof}
Let $\mc H$ be the real Hilbert space of $n$-qubit Hermitian operators (which has dimension $4^n$) equipped with the Hilbert--Schmidt inner product $\langle A,B\rangle \coloneqq \Tr(AB)$.
Let
\begin{equation}
\mc V_R \coloneqq\mathrm{span}_{\mathbb R}\{P\in\mc P_n^+:\supp(P)=R\} \subset \mc H.
\end{equation}
Since $\Tr(PQ)=2^n\delta_{P,Q}$ for $P,Q\in\mc P_n^+$, the set
\begin{equation}
\Big\{\widetilde P \equiv 2^{-n/2}P: P\in\mc P_n^+,\ \supp(P)=R\Big\}
\end{equation}
is an orthonormal basis of $\mc V_R$.

Let $\mathrm{Proj}_R:\mc H\to \mc V_R$ be the orthogonal projector onto $\mc{V}_R$.
Then for any $\rho\in\mc H$,
\begin{equation}\label{eq:PiR_norm}
\|\mathrm{Proj}_R(\rho)\|_{\mathrm{HS}}^2
=\sum_{P\in\mc P_n^+:\supp(P)=R}\langle \widetilde P,\rho\rangle^2
=\sum_{P\in\mc P_n^+:\supp(P)=R}2^{-n}\Tr(P\rho)^2,
\end{equation}
where $\|\cdot\|_{\mathrm{HS}}$ is the Hilbert--Schmidt norm.

Now consider the conjugation map $\operatorname{Ad}_U:\mc H\to\mc H$, $\operatorname{Ad}_U(A)\coloneqq UAU^\dagger$.
It is an isometry for $\langle\cdot,\cdot\rangle$:
\begin{equation}
\langle \operatorname{Ad}_U(A),\operatorname{Ad}_U(B)\rangle=\Tr(UAU^\dagger UBU^\dagger)=\Tr(AB)=\langle A,B\rangle.
\end{equation}
Moreover, $\operatorname{Ad}_U$ preserves $\mc V_R$. Indeed, if $P\in\mc P_n^+$ has $\supp(P)=R$, then
\begin{equation}
UPU^\dagger=\bigotimes_{i\in R} (U_i P_i U_i^\dagger) \otimes \bigotimes_{i\notin R} \mbb{I},
\end{equation}
where for each $i\in R$ the single-qubit operator $U_i P_i U_i^\dagger$ is Hermitian and traceless and hence lies in
$\mathrm{span}_{\mathbb R}\{X,Y,Z\}$; in particular, it has no $\mbb{I}$ component. Therefore, every term in the Pauli expansion of $UPU^\dagger$ still has support exactly $R$, so $\operatorname{Ad}_U(\mc V_R)\subseteq \mc V_R$. Since $\operatorname{Ad}_U$ is invertible, we have
$\operatorname{Ad}_U(\mc V_R)=\mc V_R$.

Because $\operatorname{Ad}_U$ is an isometry and $\mc V_R$ is invariant, the projector commutes with $\operatorname{Ad}_U$:
\begin{equation}
\mathrm{Proj}_R\circ \operatorname{Ad}_U = \operatorname{Ad}_U\circ \mathrm{Proj}_R.
\end{equation}
Therefore,
\begin{equation}
\|\mathrm{Proj}_R(U\rho U^\dagger)\|_{\mathrm{HS}}^2
=\|\operatorname{Ad}_U(\mathrm{Proj}_R(\rho))\|_{\mathrm{HS}}^2
=\|\mathrm{Proj}_R(\rho)\|_{\mathrm{HS}}^2.
\end{equation}
Then the claimed identity follows from applying \eqref{eq:PiR_norm} to $\rho$ and $U\rho U^\dagger$.
\end{proof}

\subsection{Pauli spectrum of logical states}

{To apply the criterion in Theorem~\ref{thm:transversal_testing_app}, the following lemma which characterizes the Pauli spectrum of logical states is useful.} 

\begin{lem}
\label{lemma:logical_state_spectrum}
Let $G\subset \mc{P}_n$ be the stabilizer group of an $[\![n,k]\!]$ stabilizer code.
Let $V$ be its encoding Clifford unitary, so that for every $k$-qubit density operator $\rho$, $\overline{\rho}\coloneqq V\big(\rho\otimes \ketbra{0^{n-k}}{0^{n-k}}\big)V^\dagger$
is a state supported on the codespace.
For every $Q\in\mc P_k^+$, define an encoded logical Pauli
$\overline{Q}\coloneqq V\big(Q\otimes \mbb I_2^{\otimes (n-k)}\big)V^\dagger$.
Then
\begin{equation}\label{eq:encoded_state_expansion_full}
\overline{\rho}=\frac{1}{2^n}\sum_{Q\in\mc P_k^+}\sum_{g\in G}\Tr(Q\rho)\,\overline{Q}g.
\end{equation}
\end{lem}

Note that the set $\{\overline{Q}g: Q\in\mc P_k^+,g\in G\}$ consists of $4^k2^{n-k}$ distinct elements, and up to phases, consists of all elements in $\mc{N}(G)$ (which is a group of size $4^{k+1}2^{n-k}$).
Elements in $\{\overline{Q}g: Q\in\mc P_k^+,g\in G\}$ have global phases $+1$ or $-1$ since $[\overline{Q},g]=0$. Let $\mathrm{sgn}(\overline{Q}g)\in\{+1,-1\}$ be the global phase. The Pauli spectrum of $\overline{\rho}$ can be read off directly from the decomposition
\begin{equation}
\overline{\rho}=\frac{1}{2^n}\sum_{Q\in\mc P_k^+}\sum_{g\in G}\Big[\mathrm{sgn}(\overline{Q}g)\Tr(Q\rho)\Big]\Big[\mathrm{sgn}(\overline{Q}g)\overline{Q}g\Big].
\end{equation}

\begin{proof}[Proof of Lemma~\ref{lemma:logical_state_spectrum}]
We can write
\begin{equation}
\ketbra{0^{n-k}}{0^{n-k}}=\frac{1}{2^{n-k}}\sum_{\mathbf b\in\{0,1\}^{n-k}} Z^{\mathbf b}
\end{equation}
with $Z^{\mathbf b}\coloneqq Z^{b_1}\otimes\cdots\otimes Z^{b_{n-k}}$.
Expanding $\rho$ in the Hermitian Pauli basis yields
\begin{equation}\label{eq:pauli_expansion_rho}
\rho=\frac{1}{2^k}\sum_{Q\in\mc P_k^+}\Tr(Q\rho)\,Q.
\end{equation}
Therefore,
\begin{align}
\overline{\rho}
&=V\big(\rho\otimes \ketbra{0^{n-k}}{0^{n-k}}\big)V^\dagger\\
&=\frac{1}{2^k}\sum_{Q\in\mc P_k^+}\Tr(Q\rho)\cdot \frac{1}{2^{n-k}}
\sum_{\mathbf b\in\{0,1\}^{n-k}} V\big(Q\otimes Z^{\mathbf b}\big)V^\dagger\\
&=\frac{1}{2^n}\sum_{Q\in\mc P_k^+}\Tr(Q\rho)\sum_{\mathbf b}
V\big((Q\otimes \mbb I)(\mbb I\otimes Z^{\mathbf b})\big)V^\dagger\\
&=\frac{1}{2^n}\sum_{Q\in\mc P_k^+}\Tr(Q\rho)\sum_{\mathbf b}
\underbrace{V(Q\otimes \mbb I)V^\dagger}_{\overline{Q}}
\underbrace{V(\mbb I\otimes Z^{\mathbf b})V^\dagger}_{\in G}.
\end{align}

Now for $i=1,\cdots,n-k$, define 
\begin{equation}
g_i\coloneqq V Z_{k+i}V^\dagger,
\end{equation}
which are mutually commuting and independent Pauli operators that generate $G$. Moreover,
\begin{equation}
V(\mbb I\otimes Z^{\mathbf b})V^\dagger=g_1^{b_1}\cdots g_{n-k}^{b_{n-k}},
\end{equation}
so as $\mathbf b$ ranges over $\{0,1\}^{n-k}$, $V(\mbb I\otimes Z^{\mathbf b})V^\dagger$ is a bijection onto $G$.
Therefore, we can write
\begin{equation}
\overline{\rho}
=\frac{1}{2^n}\sum_{Q\in\mc P_k^+}\sum_{g\in G}\Tr(Q\rho)\,\overline Q\,g.
\end{equation}

Note that, equivalently, with  the codespace projector given by $\Pi=\frac{1}{2^{n-k}}\sum_{g\in G} g$, we can write
\begin{equation}
\overline{\rho}
=\Big(\frac{1}{2^k}\sum_{Q\in\mc P_k^+}\Tr(Q\rho)\,\overline{Q}\Big)\Pi.
\end{equation}
\end{proof}

\subsection{No transversal $T$ on the gross code}

\begin{figure}[t]
\centering
\includegraphics[width=0.7\textwidth]{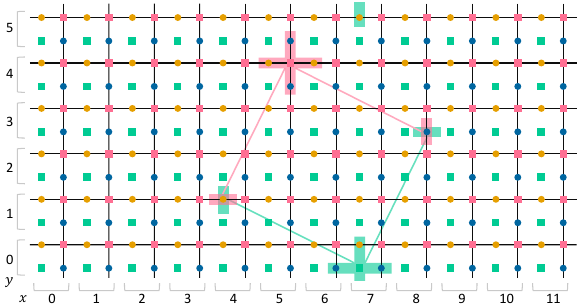}
\caption{The qubit layout and the stabilizer generators of the $[\![144,12,12]\!]$ gross code~\cite{Bravyi2024High,yoder2025tourgrossmodularquantum}. For illustration, the $X$-check $x^5y^4X(A, B)$ and the $Z$-check $x^7Z(B^\T,A^\T)$ are highlighted; they commute since their supports overlap on exactly two qubits, namely $(L,x^8y^3)$ and $(R,x^4y)$.}
\label{fig:gross_code}
\end{figure}

To exemplify the applicability of our method on a practically important finite code, we examine the $[\![144,12,12]\!]$ gross code, a representative example of the bivariate bicycle code family~\cite{Bravyi2024High,yoder2025tourgrossmodularquantum} which has received substantial interest recently for its desirable features for building fault-tolerant quantum architectures.
See Ref.~\cite{yoder2025tourgrossmodularquantum} for the detailed definitions and notations regarding this code that we use in the subsequent discussion.

This code is defined on a lattice of size $12\times6$ on the torus, with qubits placed on the edges. There are 72 qubits placed on the vertical lines (colored in blue and named as ``L qubits"), and 72 qubits placed on the horizontal lines (colored in orange and named as ``R qubits").
The 144 qubits are labeled by $L/R$ and monomials, e.g. $(L,x^8y^3)$ for the blue qubit with coordinate $(8,3)$, see Fig.~\ref{fig:gross_code}.

There are 72 $X$-type Pauli stabilizer generators on vertices (not independent, colored in pink, all weight 6) and 72 $Z$-type Pauli stabilizer generators on plaquettes (not independent, colored in green, all weight 6).
Let
\begin{equation}
A=1+y+x^3y^{-1},\quad B=1+x+x^{-1}y^{-3},
\end{equation}
the $X$-type and $Z$-type stabilizer generators are defined as:
\begin{align}
\gamma X(A,B)=&X(\gamma A,\gamma B)\text{ for all monomials }\gamma,\\
\gamma Z(B^{\T}, A^{\T})=&Z(\gamma B^{\T},\gamma A^{\T})\text{ for all monomials }\gamma.
\end{align}
Here, for example, $x^5y^4X(A,B)$ denotes the Pauli string that acts by $X$ on the $L$ qubits indexed by the monomials appearing in $x^5y^4A=x^5y^4+x^5y^5+x^8y^3$, also by $X$ on the $R$ qubits indexed by the monomials appearing in $x^5y^4B$, and by $\mbb{I}$ on all remaining qubits.
See Fig.~\ref{fig:gross_code} for $x^5y^4X(A,B)$ and $x^7Z(B^\T,A^\T)$.
Let $G$ be the stabilizer group generated by these stabilizers.

The logical Pauli group is generated by $\overline{X}_1,\cdots,\overline{X}_{12},\overline{Z}_1,\cdots,\overline{Z}_{12}$. Let
\begin{equation}
\overline{X}_1=X(p,q),\quad\overline{X}_7=X(r,s),\quad\overline{Z}_1=Z(\nu s^{\T},\nu r^{\T}),\quad\overline{Z}_7=Z(\mu q^{\T},\mu p^{\T}),
\end{equation}
with polynomials
\begin{align}
p&=x^4+x^5+x^6y+x^4y^2+x^5y^4+x^6y^5,\\
q&=x^3+x^4+x^3y+x^3y^2+x^4y^2+x^3y^5,\\
r&=1+x^8+xy+x^9y+x^3y^4+x^{11}y^4,\\
s&=x+x^9+x^4y^4+x^8y^4+y^5+x^8y^5,\\
\mu&=\nu=xy.
\end{align}
See Fig.~\ref{fig:gross_code_LRM0} for $\overline{X}_1$. For $i=1,\cdots,6$, define
\begin{align}
&\overline{X}_i=\alpha_i\overline{X}_1=X(\alpha_ip,\alpha_iq),
&&\overline{X}_{i+6}=\beta_i^{\T}\overline{X}_7=X(\beta_i^{\T}r,\beta_i^{\T}s),\\
&\overline{Z}_i=\beta_i\overline{Z}_1=Z(\beta_i\nu s^{\T},\beta_i\nu r^{\T}),
&&\overline{Z}_{i+6}=\alpha_i^{\T}\overline{Z}_7=Z(\alpha_i^{\T}\mu q^{\T},\alpha_i^{\T}\mu p^{\T}),
\end{align}
where
\begin{align}
\alpha&=(1,x^3y^5,x^{11}y^5,x^{10}y,x^5y^4,x^4y^2),\\
\beta&=(1,xy,x^4,x^5y^4,x^4y^3,x^3y^5).
\end{align}

We write $P_1\doteq P_2$ if Pauli strings $P_1$ and $P_2$ are equal up to a global phase. 

The full argument using Theorem~\ref{thm:transversal_testing_app} to rule out transversal $T$ on this gross code proceeds in the following three steps.

\textit{(i) No transversal implementation of the $T$ gate on the first logical qubit.}

Denote $\ket{T}=T\ket{+}=\frac{1}{\sqrt{2}}(\ket{0}+e^{i\pi/4}\ket{1})$, corresponding to the Bloch vector $(\frac{1}{\sqrt{2}},\frac{1}{\sqrt{2}},0)$.
Consider the 12-qubit state $\ket{\psi}\coloneqq(T\otimes\mbb{I}_2^{\otimes11})\ket{+^{12}}$. By Lemma~\ref{lemma:logical_state_spectrum}, the corresponding logical state can be written as
\begin{equation}\label{eq:Pauli_spectrum_no_T_on_first_gross}
\ketbra{\overline{\psi}}{\overline{\psi}}=\frac{1}{2^n}\sum_{Q\in\mc P_{12}^+}\sum_{g\in G}\bra{\psi}Q\ket{\psi}\,\overline{Q}g.
\end{equation}
By $\bra{\psi}X_1\ket{\psi}=\frac{1}{\sqrt{2}}$, we know
\begin{align}
\bra{\overline{\psi}}\overline{X}_1\ket{\overline{\psi}}=&\frac{1}{2^n}\sum_{Q\in\mc P_{12}^+}\sum_{g\in G}\bra{\psi}Q\ket{\psi}\Tr\big(\overline{X}_1\overline{Q}g\big)\\
=&\frac{1}{2^n}\bra{\psi}X_1\ket{\psi}\Tr\big(\overline{X}_1\overline{X}_1\mbb{I}\big)=\frac{1}{\sqrt{2}},
\end{align}
where the second equality holds since only $Q=X_1$ and $g=\mbb{I}$ can make $\Tr(\overline{X}_1\overline{Q}g)\neq0$.

By Eq.~\eqref{eq:Pauli_spectrum_no_T_on_first_gross}, we know
\begin{equation}\label{eq:not_commute_then0_no_T_on_first_gross}
P\notin\mc{N}(G)\Longrightarrow\bra{\overline{\psi}}P\ket{\overline{\psi}}=0.
\end{equation}
Now we show that
\begin{equation}\label{eq:commute_samesuppthensame_no_T_on_first_gross}
P\in\mc{N}(G)\text{ and }\supp(P)=\supp(\overline{X}_1)\Longrightarrow P\doteq\overline{X}_1.
\end{equation}

Any Pauli string can be written (up to phase) as $P\doteq X(a,b)Z(c,d)$. 
We first prove that
\begin{equation}
(c,d)=(0,0).
\end{equation}
Since $\supp(P)=\supp(X(a,b))\cup\supp(Z(c,d))$, we have $\supp(Z(c,d))\subset\supp(\overline{X}_1)$.
By $P\in\mc{N}(G)$, we know that $P$ commutes with all the $X$-checks; equivalently, $Z(c,d)$ commutes with all the $X$-checks.
As $X(a,b)$ commutes with all $X$-checks, this is
equivalent to $Z(c,d)$ commuting with every $X$-check, i.e.
\begin{equation}\label{eq:even_overlap_condition_no_T_on_first_gross}
\big|\supp(Z(c,d))\cap\supp(\gamma X(A,B))\big|\equiv 0 \pmod 2,
\quad\text{for all } \gamma.
\end{equation}

\begin{figure}[t]
\centering
\includegraphics[width=0.7\textwidth]{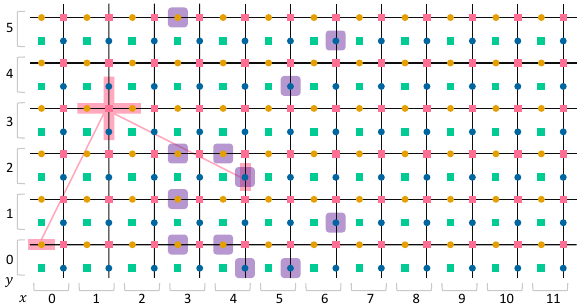}
\caption{The logical Pauli $\overline{X}_1$ and the $X$-check $xy^3X(A,B)$ overlap at only one qubit. The purple squares label the qubits that $\overline{X}_1$ acts on. The only Pauli string that commutes with the stabilizer group of the gross code and has the same support as $\overline{X}_1$ is $\overline{X}_1$ itself (up to a global phase).}
\label{fig:gross_code_LRM0}
\end{figure}

We have
\begin{equation}
\begin{aligned}
\supp(\overline{X}_1)=\{&(L,x^4),(L,x^4y^2),(L,x^5),(L,x^5y^4),(L,x^6y),(L,x^6y^5),\\
&(R,x^3),(R,x^3y),(R,x^3y^2),(R,x^3y^5),(R,x^4),(R,x^4y^2)\}.
\end{aligned}
\end{equation}
For each of the first ten qubits $q\in\supp(\overline X_1)-\{(R,x^4),(R,x^4y^2)\}$, there exists an $X$-check $\gamma_q X(A,B)$ such that
\begin{equation}\label{eq:single_overlap_checks_no_T_on_first_gross}
\supp(\gamma_q X(A,B))\cap \supp(\overline X_1)=\{q\}.
\end{equation}
Explicitly, one may take
\begin{equation}
\begin{array}{c|cccccccccc}
q & (L,x^4) & (L,x^4y^2) & (L,x^5) & (L,x^5y^4) & (L,x^6y) & (L,x^6y^5)
  & (R,x^3) & (R,x^3y) & (R,x^3y^2) & (R,x^3y^5)\\ \hline
\gamma_q & xy & xy^3 & x^5 & x^5y^4 & x^6y & x^6y^5 & x^2 & x^3y & x^2y^2 & x^3y^5
\end{array}.
\end{equation}
See Fig.~\ref{fig:gross_code_LRM0} for the $q=(L,x^4y^2)$ case.
Now if $Z(c,d)$ acted on any such $q$, then $\supp(Z(c,d))\cap\supp(\gamma_q X(A,B))$ would have odd cardinality (because the intersection would contain exactly $q$), contradicting \eqref{eq:even_overlap_condition_no_T_on_first_gross}.
Hence $\supp(Z(c,d))\subset \{(R,x^4),(R,x^4y^2)\}$.

The same ``single-overlap'' argument eliminates these two remaining positions.
There exist $X$-checks $\gamma_q X(A,B)$ such that
\begin{equation}
\supp(\gamma_q X(A,B))\cap\{(R,x^4),(R,x^4y^2)\}=\{q\},
\end{equation}
For example, we can take
\begin{equation}
\begin{array}{c|cc}
q & (R,x^4) & (R,x^4y^2)\\ \hline
\gamma_q & x^4 & x^4y^2
\end{array}.
\end{equation}
Therefore $Z(c,d)$ cannot act on either $(R,x^4)$ or $(R,x^4y^2)$ while still satisfying
\eqref{eq:even_overlap_condition_no_T_on_first_gross}. We conclude that $(c,d)=(0,0)$.

Now we know $P\doteq X(a,b)$, so by $\supp(P)=\supp(\overline{X}_1)$ we know $(a,b)=(p,q)$, that is, $P\doteq\overline{X}_1$.

Therefore, by \eqref{eq:not_commute_then0_no_T_on_first_gross} and \eqref{eq:commute_samesuppthensame_no_T_on_first_gross} the only Pauli string $P\in\mc P_n^+$ with $\supp(P)=\supp(\overline X_1)$ that can contribute nontrivially to $\bra{\overline\psi}P\ket{\overline\psi}$ is $P\doteq \overline X_1$.
Hence
\begin{equation}
\sum_{P\in\mc{P}_n^+:\,\supp(P)=\supp(\overline{X}_1)} \big|\bra{\overline{\psi}} P\ket{\overline{\psi}}\big|^2
=\big|\bra{\overline{\psi}} \overline{X}_1\ket{\overline{\psi}}\big|^2=\frac{1}{2} \notin \mathbb{Z}.
\end{equation}
Due to Theorem~\ref{thm:transversal_testing_app}, we conclude that there is no transversal implementation of $T\otimes\mbb{I}_2^{\otimes 11}$ for the gross code.

\textit{(ii) No transversal implementation of the $T$ gate on the 7-th logical qubit.}

The proof approach is the same as (i). We need to show
\begin{equation}
P\in\mc{N}(G)\text{ and }\supp(P)=\supp(\overline{X}_7)\Longrightarrow P\doteq\overline{X}_7.
\end{equation}
Suppose $P\doteq X(a,b)Z(c,d)$. We first show that $(c,d)=(0,0)$. Notice that
\begin{equation}
\begin{aligned}
\supp(\overline{X}_7)=\{&(L,1),(L,x^8),(L,xy),(L,x^9y),(L,x^3y^4),(L,x^{11}y^4)\\
&(R,x),(R,x^9),(R,x^4y^4),(R,x^8y^4),(R,y^5),(R,x^8y^5)\}.
\end{aligned}
\end{equation}
By considering $X$-checks $\gamma X(A,B)$ with $\gamma=x^2y^3,x^4y^4,x^7y^5,x^8y^4,x^{10}y^3,x^{11}y^5$, we know $\supp(Z(c,d))$ has no $R$-qubits.
Then consider $X$-checks $\gamma X(A,B)$ with $\gamma=1,x^8,xy,x^9y,x^3y^4,x^{11}y^4$, we know $(c,d)=0$. Hence $P\doteq\overline{X}_7$.
The remainder of the proof proceeds exactly as in part (i).

\textit{(iii) Finally, by translation invariance, none of the remaining $10$ logical qubits admits a transversal implementation of a logical $T$ gate.}

\subsection{No transversal $T$ on general-dimensional toric codes}

Then we consider toric codes in arbitrary dimensions and show that they do not support transversal $T$, providing a prominent example of the application of our method to asymptotic code families.

Similar to the proof of the gross code case, for each logical qubit $i$ there exists a choice of logical $\overline{X}_i$ such that the only Pauli $P\in\mc{P}_n^+$ commuting with $\overline{X}_i$ and satsifying $\supp(P)=\supp(\overline{X}_i)$ is $\overline{X}_i$ itself. Hence, by
\begin{equation}
\sum_{P\in\mc{P}_n^+:\,\supp(P)=\supp(\overline{X}_i)} \big|\bra{\overline{\psi}} P\ket{\overline{\psi}}\big|^2=1/2
\end{equation}
(where $\ket{\psi}=T_i\ket{+^k}$) and Theorem~\ref{thm:transversal_testing_app}, we conclude that transversal implementation of a $T$ gate on any logical qubit is impossible.

\section{Topological order and LRM phases}\label{app:TopologicalOrder}

\subsection{Qudit stabilizer formalism}\label{app:qudit_stab_formalism}

There are various ways to define the stabilizer formalism for qudits of general dimension. Here, we adopt the definition that matches the construction of TSCs in Ref.~\cite{Ellison2022Pauli}, which uses topological stabilizer codes on qudits of different, possibly non-prime, dimensions to realize topological orders.

For any integer $q\ge2$, define the generalized Pauli operators $X^{(q)},Z^{(q)}\in L(\mbb{C}^q)$ by
\begin{equation}
\begin{aligned}
& X^{(q)}\ket{j}=\ket{(j+1)~\operatorname{mod}~q}, \\
& Z^{(q)}\ket{j}=(e^{2\pi i/q})^j\ket{j},
\end{aligned}
\end{equation}
where $i=\sqrt{-1}$.
We have $Z^{(q)}X^{(q)}=e^{2\pi i/q}X^{(q)}Z^{(q)}$.
Define the \emph{Pauli group}~\cite{Hostens2005Stabilizer} on $\mbb{C}^q$ to be 
\begin{equation}
\mc{P}^{(q)}:=\langle e^{i\pi/q}\mbb{I},X^{(q)},Z^{(q)}\rangle.
\end{equation} 
When $q=2$, we get ordinary qubit Pauli operators $X,Z$, and the ordinary qubit Pauli group $\langle i\mbb{I},X,Z\rangle$.

For an $N$-particle system $\mc{H}=\otimes_{i=1}^N\mbb{C}^{q_i}$, where $q_i\in\mbb{Z}_{\ge2}$ for all $i=1,\cdots,N$, define the  \emph{multi-particle Pauli group} to be 
\begin{equation}
\mc{P}^{\mc{H}}:=\otimes_{i=1}^N\mc{P}^{(q_i)},
\end{equation}
which aligns with the tensor structure of $\mc{H}$. 
An element in $\mc{P}^{\mc{H}}$ is called a \emph{Pauli string} on $\mc{H}$.
Notice that $q_1,\cdots,q_N$ can be distinct and can be composite numbers.

We call a subgroup $G\subset\mc{P}^{\mc{H}}$ a \emph{stabilizer group} on $\mc{H}$, if:
\begin{enumerate}
\item $G$ is Abelian.
\item For any $e^{i\theta}\mbb{I}$ with $e^{i\theta}\neq1$, we have $e^{i\theta}\mbb{I}\notin G$.
\end{enumerate}

For every stabilizer group $G$ on $\mc{H}$, we define the corresponding \emph{stabilizer code} to be the subspace $\mc{C}_G$ of $\mc{H}$ that is the common $+1$ eigenspace of elements in $G$, that is,
\begin{equation}
\mc{C}_G:=\operatorname{span}_{\mbb{C}}\{\ket{\psi}\mid g\ket{\psi}=\ket{\psi}\text{ for all }g\in G\}.
\end{equation}
Then \emph{stabilizer states} are pure states corresponding to $1$-dimensional stabilizer codes. We say a pure state on $\mc{H}$ is a \emph{nonstabilizer (magic) state}, if it is not a stabilizer state.

\begin{lem}\label{lemma:qudit_stabilizerprojector}
For an $N$-particle system $\mc{H}=\otimes_{i=1}^N\mbb{C}^{q_i}$, and a stabilizer group $G$ on $\mc{H}$, the projector onto $\mc{C}_G$ is
\begin{equation}
\Pi_G:=\frac{1}{\abs{G}}\sum_{g\in G}g,
\end{equation}
and we have
\begin{equation}
\operatorname{dim}\mc{C}_G=\frac{\prod_{i=1}^Nq_i}{\abs{G}}.
\end{equation}
\end{lem}
\begin{proof}
Since all $g\in G$ are unitary and $G$ is a group, we know $g^\dagger=g^{-1}\in G$. Hence
\begin{equation}
\Pi_G^{\dagger}=\frac{1}{|G|}\sum_{g\in G}g^\dagger=\frac{1}{|G|}\sum_{g\in G}g^{-1}=\frac{1}{|G|}\sum_{g\in G}g=\Pi_G,
\end{equation}
so $\Pi_G$ is Hermitian. Since for all $g\in G$ we have $gG=G$, we know $g\Pi_G=\Pi_G$ and
\begin{equation}
\Pi_G\Pi_G=\frac{1}{|G|}\sum_{g\in G}g\Pi_G=\frac{1}{|G|}\sum_{g\in G}\Pi_G=\Pi_G.
\end{equation}
Therefore, all the eigenvalues of $\Pi_G$ are either $0$ or $1$.

($\operatorname{Im}\Pi_G\subset\mc{C}_G$): For any $\ket{\psi}$ satisfying $\Pi_G\ket{\psi}=\ket{\psi}$, and any $g\in G$, we have
\begin{equation}
g\ket{\psi}=g\Pi_G\ket{\psi}=\Pi_G\ket{\psi}=\ket{\psi}.
\end{equation}
($\mc{C}_G\subset\operatorname{Im}\Pi_G$): For any $\ket{\psi}\in\mc{C}_G$ we have
\begin{equation}
\Pi_G\ket{\psi}=\frac{1}{|G|}\sum_{g\in G}g\ket{\psi}=\frac{1}{|G|}\sum_{g\in G}\ket{\psi}=\ket{\psi}.
\end{equation}
Consequently, $\mc{C}_G=\operatorname{Im}\Pi_G$.

Since $\mbb{I}\in G$, and every element in $G$ except $\mbb{I}$ has trace 0, we have
\begin{equation}
\operatorname{dim}\mc{C}_G=\Tr(\Pi_G)=\frac{\Tr(\mbb{I})}{\abs{G}}=\frac{\prod_{i=1}^Nq_i}{\abs{G}}.
\end{equation}
\end{proof} 

\begin{lem}\label{lemma:qudit_PauliGroup_max_order}
Consider $\mc{H}:=\otimes_{i=1}^N\mbb{C}^{q_i}$.
For all $Q\in\mc{P}^{\mc{H}}$ we have $Q^{2L}=\mbb{I}$, where $L=\operatorname{lcm}(q_1,\cdots,q_N)$ is the least common multiple of $q_1,\cdots,q_N$.
\end{lem}
For example, in qubit case we have $(iX)^4=\mbb{I}$.
\begin{proof}
First notice that
\begin{equation}
\langle e^{i\pi/q_1}\mbb{I},\cdots,e^{i\pi/q_N}\mbb{I}\rangle=\langle e^{i\pi/L}\mbb{I}\rangle.
\end{equation}
Since $Z^{(q)}X^{(q)}=e^{2\pi i/q}X^{(q)}Z^{(q)}$, we know every element of $\mc P^{\mc H}$ has a unique decomposition
\begin{equation}
Q=\zeta\bigotimes_{i=1}^{N}(X^{(q_i)})^{a_i}(Z^{(q_i)})^{b_i},\;\text{ where }\;
\zeta=e^{i\pi k/L},\;
a_i,b_i\in\{0,\cdots,q_i-1\},\;
k\in\mbb{Z}.
\end{equation}
We have $\zeta^{2L}=e^{i2\pi k}=1$. Also for each tensor factor, we know
\begin{equation}
\big[(X^{(q_i)})^{a_i}(Z^{(q_i)})^{b_i}\big]^{2L}=(e^{2\pi i/q_i})^{-L(2L-1)a_ib_i}(X^{(q_i)})^{2La_i}(Z^{(q_i)})^{2Lb_i}=(e^{2\pi i\frac{L}{q_i}})^{-(2L-1)a_ib_i}\mbb{I}=\mbb{I}.
\end{equation}
Therefore $Q^{2L}=\mbb{I}$.
\end{proof}

\begin{lem}\label{lemma:qudit_stabilizercode_has_stabilizerstate}
Consider $\mc{H}:=\otimes_{i=1}^N\mbb{C}^{q_i}$. For every stabilizer group $G$ on $\mc{H}$, there exists a stabilizer group $G'$ on $\mc{H}$ such that $G\subset G'$ and $\abs{G'}=\operatorname{dim}(\mc{H})$.
\end{lem}

\begin{proof}
Using the symplectic representation of the Pauli group, we can show that if a stabilizer group $G$ satisfies $\abs{G}<\operatorname{dim}(\mc{H})$, then there exists a Pauli operator that commutes with all elements of $G$ but is not itself contained in $G$, up to a global phase. Adding this element to $G$, we can generate a strictly larger stabilizer group. By iterating this process, we eventually obtain a stabilizer group $G'\supset G$ such that $\abs{G'}=\operatorname{dim}(\mc{H})$.
\end{proof}

\subsection{TSC and local Hamiltonian realization of topological orders}\label{app:TSC_and_localHamiltonian}

\begin{defn}
Let $\mathcal M$ be a $D$-dimensional manifold with $D\ge 2$.
Let $m\ge1$ be an integer, and $q_1,\cdots,q_m\in\mbb{Z}_{\ge2}$.
An infinite family of subspaces $\{\mc{C}^{(n)}\}_n$ is said to be a \emph{topological stabilizer code (TSC)} on $\mc{M}$ with \emph{local configuration} $(q_1,\cdots,q_m)$ (which captures the local subsystem dimensions) given the following conditions:
\begin{enumerate}
\item For each $n$ there is a regular lattice $\Lambda^{(n)}$ embedded in $\mathcal M$ whose edge set is $E^{(n)}$ with $\lvert E^{(n)}\rvert = n$.
\item Each edge $e\in E^{(n)}$ carries the local Hilbert space $\mc{H}_e:=\otimes_{j=1}^{m}\mathbb C^{q_j}$.
\item $\mc{C}^{(n)}$ is a subspace of $\mc{H}^{(n)}:=\otimes_{e\in E^{(n)}} \mc{H}_e$. All the $\mc{C}^{(n)}$'s have the same dimension.
\item $\mc{C}^{(n)}$ is a stabilizer code, with stabilizer group $G^{(n)}\subset\mc{P}^{\mc{H}^{(n)}}$.
\item There exists a subset $S^{(n)}\subset G^{(n)}$, satisfying $\langle S^{(n)}\rangle=G^{(n)}$, and $\operatorname{Diam}(s)\le\ell$ for all $s\in S^{(n)}$, where $\ell$ is a constant independent of $n$.
\item $\{\mc{C}^{(n)}\}_n$ has macroscopic distance. That is, the distance of $\mc{C}^{(n)}$ tends to $\infty$ when $n\rightarrow\infty$.
\end{enumerate}
\end{defn}
Here we define the code distance by treating the collection of $m$ subsystems on a single edge as a unit.
Specifically, the code has distance $d$ iff for any set of $d-1$ edges, the reduced density matrices of all code states restricted to these edges are identical.

We explicitly write the local tensor structure rather than treating $\otimes_{i=1}^m\mbb{C}^{q_i}$ as $\mbb{C}^{\prod_{i=1}^mq_i}$, because the stabilizer formalism depends on the tensor structure of the Hilbert space.

\begin{defn}
Let $\mathcal M$ be a $D$-dimensional manifold with $D\ge 2$.
Let $m\ge1$ be an integer, and $q_1,\cdots,q_m\in\mbb{Z}_{\ge2}$.
A topological order on $\mathcal M$ is said to be \emph{realized} by a family of local Hamiltonians
$\{H^{(n)}\}_{n}$
with \emph{local configuration} $(q_1,\cdots,q_m)$ given the following conditions:
\begin{enumerate}
\item For each $n$ there is a regular lattice $\Lambda^{(n)}$ embedded in $\mathcal M$ whose edge set is $E^{(n)}$ with $\lvert E^{(n)}\rvert = n$.
\item Each edge $e\in E^{(n)}$ carries the local Hilbert space $\mc{H}_e:=\otimes_{j=1}^{m}\mathbb C^{q_j}$.
\item The Hamiltonian $H^{(n)} = \sum_{i\in I^{(n)}} H_i^{(n)}$ acts on $\mc{H}^{(n)}:=\otimes_{e\in E^{(n)}} \mc{H}_e$.
Every local term satisfies $\operatorname{Diam}\big(H_i^{(n)}\big)\le\ell$, where $\ell$ is a constant independent of $n$. The size of the index set satisfies $|I^{(n)}|=\Theta(n)$. $\|H_i^{(n)}\|_\infty\le1$. $\{H^{(n)}\}$ is gapped.
\item The ground space family of $\{ H^{(n)}\}_n$ realizes this topological order.
\end{enumerate}
If, in addition, the ground space family of $\{H^{(n)}\}_n$ is a TSC, we say that this topological order is \emph{realized by a TSC} with local configuration $(q_1,\cdots,q_m)$.
\end{defn}

\begin{exmp}
The code family defined in \cite[Eq.~(9)]{Ellison2022Pauli} is a TSC on tori with local configuration $q_1=4$ (here $m=1$).
It realizes the double semion phase on tori.
\end{exmp}

\begin{exmp}
Let $\mc{M}$ be an arbitrary 2D manifold.
The code family defined in \cite[Eq.~(131)]{Ellison2022Pauli} is a TSC on $\mc{M}$ with local configuration $(N_1^2,\cdots,N_M^2)$.
It realizes the Abelian twisted quantum double on $\mc{M}$ characterized by the group $\oplus_{i=1}^M\mbb{Z}_{N_i}$.
\end{exmp}

\subsection{Equivalent characterization of LRM phases}\label{app:SRM/LRMphase_TSC}

\begin{thm}\label{thm:SRM/LRMphase_TSC_app}
Given a local configuration, a topological order is an LRM phase $\iff$ it cannot be realized by any TSC with this local configuration.
\end{thm}

\begin{proof}
We prove the equivalent expression: Given a local configuration, a topological order is not an LRM phase $\iff$ it can be realized by a TSC with this local configuration.

($\Leftarrow$):

Fix a local configuration.
By Lemma~\ref{lemma:qudit_stabilizercode_has_stabilizerstate} we know every stabilizer code contains a stabilizer state.
Therefore, the given TSC family (with this local configuration) that realizes this topological order has a family of ground stabilizer states.
By definition, this topological order is not an LRM phase (with this local configuration).

($\Rightarrow$):

Fix a local configuration $(q_1,\cdots,q_m)$.
By definition, we know there exists a local Hamiltonian realization $\{\tilde{H}^{(n)}=\sum_{i\in I^{(n)}}\tilde{H}_i^{(n)}\}$ of this topological order with this local configuration, such that $\{\tilde{H}^{(n)}\}$ has a family of ground states $\{\ket{\phi^{(n)}}\}$ exhibiting SRM. There exists a family of shallow circuit $\{U^{(n)}\}$ and a family of stabilizer states $\{\ket{S^{(n)}}\}$, such that $U^{(n)}\ket{\phi^{(n)}}=\ket{S^{(n)}}$.

Denote $H_i^{(n)}=U^{(n)}\tilde{H}_i^{(n)}U^{(n)\dagger}$, and let $H^{(n)}=U^{(n)}\tilde{H}^{(n)}U^{(n)\dagger}$. $\{H^{(n)}\}$ is a family of local Hamiltonian, with $\ket{S^{(n)}}$ as its ground states. Since $\{U^{(n)}\}$ is shallow, $\{H^{(n)}\}$ also realizes this topological order.

Denote $C$ as a constant independent of $n$ that upper bounds the diameters of local terms $\{H_i^{(n)}\}$.
Let $\mc{A}_n$ be the set of states that are $C$-locally indistinguishable from $\ket{S^{(n)}}$:
\begin{equation}
\begin{aligned}
\mc{A}_n:=\Big\{\ket{\psi}\Big|\Tr_{\overline{J}}\ketbra{\psi}{\psi}=\Tr_{\overline{J}}\ketbra{S^{(n)}}{S^{(n)}},\forall J\subset[n]\text{ with }\operatorname{Diam}(J)\le C\Big\}.
\end{aligned}
\end{equation}

We show that $\mc{A}_n$ is the set of pure ground states of $H^{(n)}$. 
On one hand, for an arbitrary ground state $\ket{\phi}$ of $H^{(n)}$, by the TQO condition we know the reduced density matrices of $\ket{\phi}$ and $\ket{S^{(n)}}$ are the same locally, i.e., $\ket{\phi}\in\mc{A}_n$. On the other hand, for any $\ket{a}\in\mc{A}_n$, by Fact~\ref{fact:pure_locally_same_then_same_Evalue} we know $\bra{a}H_i^{(n)}\ket{a}=\bra{S^{(n)}}H_i^{(n)}\ket{S^{(n)}}$.
Therefore, $\bra{a}H^{(n)}\ket{a}=\bra{S^{(n)}}H^{(n)}\ket{S^{(n)}}=\lambda_{\min}(H^{(n)})$. Hence we have
\begin{equation}\label{eq:qudit_LRM_from_TO_A}
\mc{A}_n=\{\text{pure ground states of } H^{(n)}\}.
\end{equation}

Denote by $G_n$ the stabilizer group of $\ket{S^{(n)}}$. Consider
\begin{equation}
K_n:=\langle\{g\in G_n\mid\operatorname{Diam}(g)\le C\}\rangle,
\end{equation}
which is the group generated by elements in $G_n$ with diameters $\le C$.
Define
\begin{equation}\label{eq:qudit_LRM_from_TO_A'}
\mc{A}'_n:=\{\ket{\psi}\mid P\ket{\psi}=\ket{\psi},\text{for all }P\in K_n\}.
\end{equation}
Now we show that $\mc{A}_n'=\mc{A}_n$.

For all $\ket{a}\in\mc{A}_n$ and all $g\in G_n$ with $\operatorname{Diam}(g)\le C$, we have $\bra{a}g\ket{a}=\bra{S^{(n)}}g\ket{S^{(n)}}=1$, that is, $g\ket{a}=\ket{a}$. Therefore $P\ket{a}=\ket{a}$ for all $P\in K_n=\langle\{g\in G_n\mid\operatorname{Diam}(g)\le C\}\rangle$, implying $\mc{A}_n\subset\mc{A}_n'$.

Consider the quotient group $G_n/K_n$, which is a finite Abelian group.
By the  {fundamental theorem of finite Abelian groups} (see e.g.~\cite{gallian2021contemporary}), we know that there exists an isomorphism
\begin{equation}
\varphi:G_n/K_n\rightarrow\mbb{Z}_{t_1}\oplus\cdots\oplus\mbb{Z}_{t_k},
\end{equation}
where $t_1,\cdots,t_k\in\mbb{Z}_{\ge2}$. (Actually, we have $t_1=p_1^{n_1},\cdots,t_k=p_k^{n_k}$, where $p_1,\cdots,p_k$ are (not necessarily distinct) prime numbers, and $n_1,\cdots,n_k\in\mbb{Z}_{\ge1}$~\cite{gallian2021contemporary}.)
Consider $e_1,\cdots,e_k\in\mbb{Z}_{t_1}\oplus\cdots\oplus\mbb{Z}_{t_k}$, where $e_i$ takes value $1$ in the $i$'th summand and $0$ in the others, e.g., $e_1=(1,0,\cdots,0)$.
Note that $e_1,\cdots,e_k$ \emph{independently} generate $\mbb{Z}_{t_1}\oplus\cdots\oplus\mbb{Z}_{t_k}$, that is,
\begin{equation}
\langle e_1,\cdots,e_k\rangle=\mbb{Z}_{t_1}\oplus\cdots\oplus\mbb{Z}_{t_k}
\end{equation}
and
\begin{equation}
a_1e_1+\cdots+a_ke_k=(0,\cdots,0)\;\;\Rightarrow\;\;a_1,\cdots,a_k\equiv0\;\text{ mod }\; t_1,\cdots,t_k.
\end{equation}
Since $\varphi$ is an isomorphism, $\varphi^{-1}(e_1),\cdots,\varphi^{-1}(e_k)$ independently generate $G_n/K_n$.
Denote $\pi:G_n\rightarrow G_n/K_n$ be the natural homomorphism.
For $i=1,\cdots,k$, take Pauli strings
\begin{equation}
Q_i\in G_n,\text{ such that } \pi(Q_i)=\varphi^{-1}(e_i).
\end{equation}
Since 
$\operatorname{ord}(\pi(Q_i))|\operatorname{ord}(Q_i)$ (where $\operatorname{ord}(\cdot)$ denotes the order of a group element), and
\begin{equation}
\operatorname{ord}(\pi(Q_i))=\operatorname{ord}(\varphi^{-1}(e_i))=\operatorname{ord}(e_i)=t_i,
\end{equation}
we know $t_i|\operatorname{ord}(Q_i)$.

Denote $L:=\operatorname{lcm}(q_1,\cdots,q_m)$ be the least
common multiple of numbers in the local configuration.
For all $i=1,\cdots,k$, by Lemma~\ref{lemma:qudit_PauliGroup_max_order}, we know $\operatorname{ord}(Q_i)|2L$, and $\langle e^{i\frac{2\pi}{2L}}\mbb{I}\rangle$ is contained in the Pauli group. Together with $t_i|\operatorname{ord}(Q_i)$ we know $t_i|2L$, implying $\langle e^{i\frac{2\pi}{t_i}}\mbb{I}\rangle$ is contained in the Pauli group.
For all $\mathbf{b}\in\mbb{Z}_{t_1}\oplus\cdots\oplus\mbb{Z}_{t_k}$, define Abelian subgroups of the Pauli group
\begin{equation}
G_n^{\mathbf{b}}:=\big\langle K_n,e^{-i\frac{2\pi}{t_1}b_1}Q_1,\cdots,e^{-i\frac{2\pi}{t_k}b_k}Q_k\big\rangle,
\end{equation}
and let
\begin{equation}\label{eq:qudit_codestate=group_average}
\psi_{\mathbf{b}}:=\frac{1}{\abs{G_n^{\mathbf{b}}}}\sum_{P\in G_n^{\mathbf{b}}}P.
\end{equation}

Now we prove that:
\begin{enumerate}
\item For all $\mathbf{b}\in\mbb{Z}_{t_1}\oplus\cdots\oplus\mbb{Z}_{t_k}$, $\psi_{\mathbf{b}}$ is a rank-1 projector. In particular, $\psi_{\mathbf{0}}=\ketbra{S^{(n)}}{S^{(n)}}$.
\item For distinct $\mathbf{b}_1,\mathbf{b}_2\in\mbb{Z}_{t_1}\oplus\cdots\oplus\mbb{Z}_{t_k}$, we have $\bracket{\psi_{\mathbf{b}_1}}{\psi_{\mathbf{b}_2}}=0$.
\item For all $\mathbf{b}\in\mbb{Z}_{t_1}\oplus\cdots\oplus\mbb{Z}_{t_k}$, $\ket{\psi_{\mathbf{b}}}\in\mc{A}_n$. 
\end{enumerate}

We show that the groups $G_n^{\mathbf{b}}$ do not contain $e^{i\theta}\mbb{I}$ if $e^{i\theta}\neq1$. Suppose $e^{i\theta}\mbb{I}\in G_n^{\mathbf{b}}$, then there exists $P\in K_n$ and $x_1,\cdots,x_k\in\mbb{Z}$, such that $P\prod_{j=1}^k\big(e^{-i\frac{2\pi}{t_j}b_j}Q_j\big)^{x_j}=e^{i\theta}\mbb{I}$, equivalently
\begin{equation}
P\prod_{j=1}^kQ_j^{x_j}=e^{i\theta}\prod_{j=1}^ke^{i\frac{2\pi}{t_j}b_jx_j}\mbb{I}
\end{equation}
Since $P\prod_{j=1}^kQ_j^{x_j}\in G_n$ and $G_n$ is a stabilizer group, we know
\begin{equation}
e^{i\theta}\prod_{j=1}^ke^{i\frac{2\pi}{t_j}b_jx_j}=1\quad\text{and}\quad P\prod_{j=1}^kQ_j^{x_j}=\mbb{I}.
\end{equation}
Consider the map $\varphi\circ\pi$, where $\pi:G_n\rightarrow G_n/K_n$ is the natural homomorphism and $\varphi:G_n/K_n\rightarrow\mbb{Z}_{t_1}\oplus\cdots\oplus\mbb{Z}_{t_k}$ is the isomorphism, we have $\varphi\circ\pi(P)=\varphi\circ\pi(\mbb{I})=(0,\cdots,0)$ and $\varphi\circ\pi(Q_i)=e_i$. Taking $\varphi\circ\pi$ over two sides of $P\prod_{j=1}^kQ_j^{x_j}=\mbb{I}$ we get
\begin{equation}
x_1e_1+\cdots+x_ke_k=(0,\cdots,0),
\end{equation}
implying $t_j|x_j$ for $j=1,\cdots,k$. Hence
\begin{equation}
e^{i\theta}=\prod_{j=1}^ke^{-i2\pi b_j\frac{x_j}{t_j}}=1.
\end{equation}
Therefore, Abelian subgroups $G_n^{\mathbf{b}}$ of the Pauli group do not contain $e^{i\theta}\mbb{I}$ with $e^{i\theta}\neq1$, meaning that they are all stabilizer groups.

Now we calculate $\abs{G_n^{\mathbf{b}}}$.
Every element in $G_n^{\mathbf{b}}$ can be written as $P'\big(e^{-i\frac{2\pi}{t_1}b_1}Q_1\big)^{x_1'}\cdots\big(e^{-i\frac{2\pi}{t_k}b_k}Q_k\big)^{x_k'}$, where $P'\in K_n$ and $(x_1',\cdots,x_k')\in\mbb{Z}^k$. For $j=1,\cdots,k$ write $x_j'=x_j+c_jt_j$ with $x_j\in\{0,1,\cdots,t_j-1\}$ and $c_j\in\mbb{Z}$. Then
\begin{align}
P'\big(e^{-i\frac{2\pi}{t_1}b_1}Q_1\big)^{x_1'}\cdots\big(e^{-i\frac{2\pi}{t_k}b_k}Q_k\big)^{x_k'}=&P'\big(e^{-i\frac{2\pi}{t_1}b_1}Q_1\big)^{x_1+c_1t_1}\cdots\big(e^{-i\frac{2\pi}{t_k}b_k}Q_k\big)^{x_k+c_kt_k}\\
=&\big(P'Q_1^{t_1c_1}\cdots Q_k^{t_kc_k}\big)\big(e^{-i\frac{2\pi}{t_1}b_1}Q_1\big)^{x_1}\cdots\big(e^{-i\frac{2\pi}{t_k}b_k}Q_k\big)^{x_k}\\
=&P\big(e^{-i\frac{2\pi}{t_1}b_1}Q_1\big)^{x_1}\cdots\big(e^{-i\frac{2\pi}{t_k}b_k}Q_k\big)^{x_k},
\end{align}
where by $Q_j^{t_j}\in K_n$ we know $P:=P'Q_1^{t_1c_1}\cdots Q_k^{t_kc_k}\in K_n$.
Hence we have
\begin{equation}\label{eq:elements_in_G_n^b}
G_n^{\mathbf{b}}=\Big\{P\big(e^{-i\frac{2\pi}{t_1}b_1}Q_1\big)^{x_1}\cdots\big(e^{-i\frac{2\pi}{t_k}b_k}Q_k\big)^{x_k}\;\big|\; P\in K_n,(x_1,\cdots,x_k)\in\mbb{Z}_{t_1}\oplus\cdots\oplus\mbb{Z}_{t_k}\Big\}
\end{equation}
We claim that the elements in this set are distinct. Suppose 
\begin{equation}
P_1\big(e^{-i\frac{2\pi}{t_1}b_1}Q_1\big)^{x_1}\cdots\big(e^{-i\frac{2\pi}{t_k}b_k}Q_k\big)^{x_k}=P_2\big(e^{-i\frac{2\pi}{t_1}b_1}Q_1\big)^{y_1}\cdots\big(e^{-i\frac{2\pi}{t_k}b_k}Q_k\big)^{y_k}
\end{equation}
for some $P_1,P_2\in K_n$ and $(x_1,\cdots,x_k),(y_1,\cdots,y_k)\in\mbb{Z}_{t_1}\oplus\cdots\oplus\mbb{Z}_{t_k}$, then
\begin{equation}
P_1P_2^{-1}Q_1^{x_1-y_1}\cdots Q_k^{x_k-y_k}=e^{i\frac{2\pi}{t_1}b_1(x_1-y_1)}\cdots e^{i\frac{2\pi}{t_k}b_k(x_k-y_k)}\mbb{I}.
\end{equation}
Taking $\varphi\circ\pi$ on two sides, we know $(x_1,\cdots,x_k)=(y_1,\cdots,y_k)$ and thus $P_1=P_2$. Hence for any $\mathbf{b}\in\mbb{Z}_{t_1}\oplus\cdots\oplus\mbb{Z}_{t_k}$ we have
\begin{equation}
\abs{G_n^{\mathbf{b}}}=\abs{K_n}t_1\cdots t_k.
\end{equation}
Also note that
\begin{equation}
\Tr(\mbb{I})=\abs{G_n}=\abs{K_n}\cdot\abs{G_n/K_n}=\abs{K_n}t_1\cdots t_k.
\end{equation}
By Lemma~\ref{lemma:qudit_stabilizerprojector} and Eq.~\eqref{eq:qudit_codestate=group_average}, we know $\{\psi_{\mathbf{b}}\mid \mathbf{b}\in\mbb{Z}_{t_1}\oplus\cdots\oplus\mbb{Z}_{t_k}\}$ are projectors, all with 
\begin{equation}
\Tr(\psi_{\mathbf{b}})=\Tr(\mbb{I})/\abs{G_n^{\mathbf{b}}}=1.
\end{equation}
Therefore, $\psi_{\mathbf{b}}$'s are stabilizer states.
By $G_n^{\mathbf{0}}\subset G_n$ and $\abs{G_n^{\mathbf{0}}}=\abs{G_n}$ we know $G_n^{\mathbf{0}}=G_n$, implying $\psi_{\mathbf{0}}=\ketbra{S^{(n)}}{S^{(n)}}$.

For distinct $\mathbf{b}_1=(b_{1,i},\cdots,b_{1,k}),\mathbf{b}_2=(b_{2,i},\cdots,b_{2,k})\in\mbb{Z}_{t_1}\oplus\cdots\oplus\mbb{Z}_{t_k}$, suppose $b_{1,j}\neq b_{2,j}$ at some $j$. Then by
\begin{equation}
e^{-i\frac{2\pi}{t_{j}}b_{1,j}}Q_j\ket{\psi_{\mathbf{b}_1}}=\ket{\psi_{\mathbf{b}_1}}\quad\text{and}\quad e^{-i\frac{2\pi}{t_{j}}b_{2,j}}Q_j\ket{\psi_{\mathbf{b}_2}}=\ket{\psi_{\mathbf{b}_2}},
\end{equation}
we know $\ket{\psi_{\mathbf{b}_1}}$ and $\ket{\psi_{\mathbf{b}_2}}$ lies in different eigenspaces of $Q_j$, thus $\bracket{\psi_{\mathbf{b}_1}}{\psi_{\mathbf{b}_2}}=0$.

For arbitrary $\ket{\psi_{\mathbf{b}}}$ and an arbitrary subset $J\subset[n]$ with $\operatorname{Diam}(J)\le C$, we have
\begin{align}
\Tr_{\overline{J}}(\psi_{\mathbf{b}})=&\frac{1}{\abs{G_n^{\mathbf{b}}}}\sum_{R\in G_n^{\mathbf{b}}}\Tr_{\overline{J}}(R)\\
=&\frac{1}{\abs{G_n^{\mathbf{b}}}}\sum_{P\in K_n}\sum_{x_1\in\mbb{Z}_{t_1}}\cdots\sum_{x_k\in\mbb{Z}_{t_k}}\Tr_{\overline{J}}\left[P\big(e^{-i\frac{2\pi}{t_1}b_1}Q_1\big)^{x_1}\cdots\big(e^{-i\frac{2\pi}{t_k}b_k}Q_k\big)^{x_k}\right].
\end{align}
When $(x_1,\cdots,x_k)\neq(0,\cdots,0)$, we have $PQ_1^{x_1}\cdots Q_k^{x_k}\notin K_n$, implying 
\begin{equation}
\operatorname{Diam}\left(P\big(e^{-i\frac{2\pi}{t_1}b_1}Q_1\big)^{x_1}\cdots\big(e^{-i\frac{2\pi}{t_k}b_k}Q_k\big)^{x_k}\right)=\operatorname{Diam}\left(PQ_1^{x_1}\cdots Q_k^{x_k}\right)>C.
\end{equation}
Hence
\begin{equation}
\Tr_{\overline{J}}(\psi_{\mathbf{b}})=\frac{1}{\abs{G_n^{\mathbf{b}}}}\sum_{P\in K_n}\Tr_{\overline{J}}\left(P\right)=\frac{1}{t_1\cdots t_k}\frac{1}{\abs{K_n}}\sum_{P\in K_n}\Tr_{\overline{J}}\left(P\right),
\end{equation}
which is the same for all $\mathbf{b}$. Therefore, all $\ket{\psi_{\mathbf{b}}}$'s are $C$-locally indistinguishable with $\ket{\psi_{\mathbf{0}}}=\ket{S^{(n)}}$, implying
\begin{equation}
\{\ket{\psi_{\mathbf{b}}}\mid \mathbf{b}\in\mbb{Z}_{t_1}\oplus\cdots\oplus\mbb{Z}_{t_k}\}\subset\mc{A}_n.
\end{equation}

Therefore, $\dim (\operatorname{span}(\mc{A}_n))\ge t_1\cdots t_k$. Note that $\dim (\operatorname{span}(\mc{A}'_n))=\Tr(\mbb{I})/\abs{K_n}=t_1\cdots t_k$, and $\operatorname{span}(\mc{A}_n)\subset\operatorname{span}(\mc{A}'_n)$.
We deduce that $\operatorname{span}(\mc{A}_n)=\operatorname{span}(\mc{A}'_n)$. By Eq.~\eqref{eq:qudit_LRM_from_TO_A} and Eq.~\eqref{eq:qudit_LRM_from_TO_A'} we know $\mc{A}_n$ and $\mc{A}'_n$ are the sets of unit vectors in $\operatorname{span}(\mc{A}_n)$ and $\operatorname{span}(\mc{A}'_n)$ respectively, which implies
\begin{equation}
\mc{A}_n=\mc{A}'_n.
\end{equation}

To conclude, the ground space family of $\{H^{(n)}\}$ is a TSC, since
\begin{enumerate}
\item All ground spaces of $\{H^{(n)}\}$ have the same dimension.
\item The ground space of $H^{(n)}$ is a stabilizer code, with stabilizer group $K_n$.
\item The diameters of all elements in the subset $\{g\in G_n\mid\operatorname{Diam}(g)\le C\}\subset K_n$ are upper bounded by $C$, and this subset generates $K_n$.
\item It exhibits a macroscopic distance.
\end{enumerate}
Therefore, this topological order can be realized by a TSC (with local configuration $(q_1,\cdots,q_m)$).
\end{proof}

\subsection{General and explicit conditions for strong LRM phases}\label{app:strongLRMphase}

\begin{thm}\label{thm:sufficientcondition_strongLRM_app}
Given a local configuration. Suppose $\mc{T}$ is a topological order. If for any family of local Hamiltonians (with this local configuration) with ground space projectors $\{\Pi^{(n)}_H\}$ that realizes this topological order, and any TSC (with this local configuration) with code space projectors $\{\Pi^{(n)}_C\}$, we have
\begin{equation}
\bignorm{\Pi^{(n)}_H-\Pi^{(n)}_C}_1=\Omega(1),
\end{equation} 
then $\mc{T}$ is a strong LRM phase.
\end{thm}

\begin{proof}[Proof sketch]
Suppose a family of ground states $\{\ket{\phi_n}\}$ of a topological order described by Hamiltonians $\{H^{(n)}\}$ satisfies $\big\|U_n\ketbra{\phi_n}{\phi_n}U_n^\dagger-\ketbra{S_n}{S_n}\big\|_1=o(1/n)$ for shallow circuits $\{U_n\}$ and stabilizer states $\{\ket{S_n}\}$.
Consider the ground space of $U_n H^{(n)} U_n^\dagger$ and the stabilizer code defined by the local Pauli strings stabilizing $\ket{S_n}$. 
One can show that these two subspaces all have macroscopic distance.
The $o(1/n)$ trace distance between $U_n\ket{\phi_n}$ and $\ket{S_n}$ implies that all states in the respective subspaces have vanishingly small excitation energy under the Hamiltonian of the other subspace.
Due to the spectral gap of $\{H^{(n)}\}$, this implies that the two subspaces must have the same dimension, and hence the distance between their projectors is $o(1)$.
\end{proof}

\begin{proof}[Proof of Theorem~\ref{thm:sufficientcondition_strongLRM_app}]
Fix a local configuration $(q_1,\cdots,q_m)\in\mbb{Z}_{\ge2}^m$.
Suppose on the contrary that such a topological order is not a strong LRM phase.
As a result, there exists a family of local Hamiltonians $\{H^{(n)}=\sum_{i\in I^{(n)}}H_i^{(n)}\}_{n\in T}$ with this configuration that realizes this topological order, and there exists a family of ground states $\{\ket{\phi_n}\}_{n\in T}$, a family of shallow circuits $\{U_n\}_{n\in T}$, and a family of stabilizer states $\{\ket{S_n}\}_{n\in T}$, such that for $n\in T$,
\begin{equation}
\big\|U_n\ketbra{\phi_n}{\phi_n}U_n^\dagger-\ketbra{S_n}{S_n}\big\|_1=o(1/n).
\end{equation}
Denote $C$ as a constant independent of $n$ that upper bounds the diameters of all deformed local terms $UH_i^{(n)}U^\dagger$.
Define $\mc{B}_n$ as the set of states that are $C$-locally indistinguishable from $U_n\ket{\phi_n}$,
\begin{equation}
\mc{B}_n:=\Big\{\ket{\psi}\Big|\Tr_{\overline{J}}\ketbra{\psi}{\psi}=\Tr_{\overline{J}}(U_n\ketbra{\phi_n}{\phi_n}U_n^\dagger),\forall J\subset[n]\text{ with }\operatorname{Diam}(J)\le C\Big\}.
\end{equation}
For an arbitrary ground state $\ket{\phi}$ of $U_nH^{(n)}U_n^\dagger$, $U_n^\dagger\ket{\phi}$ is a ground state of $H^{(n)}$, and thus locally indistinguishable from $\ket{\phi_n}$. Since $U_n$ is shallow, we know $\ket{\psi}$ is locally indistinguishable from $U_n\ket{\phi_n}$, implying $\ket{\psi}\in\mc{B}_n$ when $n\in T$ is sufficiently large.
On the other hand, since all $\ket{\psi}\in\mc{B}_n$ is $C$-locally indistinguishable from $U_n\ket{\phi_n}$ and the diameters of $U_nH_i^{(n)}U_n^\dagger$'s are upper bounded by $C$, by Fact~\ref{fact:pure_locally_same_then_same_Evalue} we know $\bra{\psi}U_nH_i^{(n)}U_n^\dagger\ket{\psi}=\bra{\phi_n}H_i^{(n)}\ket{\phi_n}$. Therefore $\bra{\psi}U_nH^{(n)}U_n^\dagger\ket{\psi}=\bra{\phi_n}H^{(n)}\ket{\phi_n}=\lambda_{\min}(H^{(n)})$. Hence we have
\begin{equation}
\mc{B}_n=\big\{\text{pure ground states of }U_nH^{(n)}U_n^\dagger\big\}.
\end{equation}
Define $\mc{A}_n$ as the set of states that are $C$-locally indistinguishable from $\ket{S_n}$:
\begin{equation}
\mc{A}_n:=\Big\{\ket{\psi}\Big|\Tr_{\overline{J}}\ketbra{\psi}{\psi}=\Tr_{\overline{J}}\ketbra{S_n}{S_n},\forall J\subset[n]\text{ with }\operatorname{Diam}(J)\le C\Big\}.
\end{equation}
Denote by $G_n$ the stabilizer group of $\ket{S^{(n)}}$. Consider $K_n:=\langle\{g\in G_n\mid\operatorname{Diam}(g)\le C\}\rangle$, and define
\begin{equation}
\mc{A}'_n:=\{\ket{\psi}\mid P\ket{\psi}=\ket{\psi},\text{for all }P\in K_n\}.
\end{equation}

Following the same procedure as in the proof of Theorem~\ref{thm:SRM/LRMphase_TSC_app}, we know $\mc{A}_n\subset\mc{A}_n'$, and there exists a set of $t_1\cdots t_k=\Tr(\mbb{I})/\abs{K_n}$ orthonormal states
\begin{equation}
\{\ket{\psi_{\mathbf{b}}}\mid \mathbf{b}\in\mbb{Z}_{t_1}\oplus\cdots\oplus\mbb{Z}_{t_k}\}\subset\mc{A}_n.
\end{equation}
Together with $\operatorname{dim}(\operatorname{span}(\mc{A}_n'))=\Tr(\mbb{I})/\abs{K_n}$ we know
\begin{equation}
\operatorname{dim}(\operatorname{span}(\mc{A}_n))=\operatorname{dim}(\operatorname{span}(\mc{A}_n'))=\Tr(\mbb{I})/\abs{K_n}.
\end{equation}

Denote by $\Pi^{(n)}$ the ground space projector of $H^{(n)}$, and denote by $\widetilde{\Pi}^{(n)}$ the projector onto the stabilizer code $\operatorname{span}(\mc{A}_n')$.
Intuitively, $U_n\ket{\phi_n}\approx\ket{S_n}$ implies $U_n\Pi^{(n)}U_n^\dagger\approx\widetilde{\Pi}^{(n)}$.

By Lemma~\ref{lemma:qudit_stabilizerprojector}, the projector onto the $+1$ eigenspace of a Pauli string $P$ is given by
\begin{equation}
\operatorname{Proj}(P,1):=\frac{1}{\operatorname{ord}(P)}\sum_{r=0}^{\operatorname{ord}(P)-1}P^{r}.
\end{equation}
Define Hamiltonian 
\begin{equation}
\widetilde{H}^{(n)}:=\sum_{g\in G_n,\,\operatorname{Diam}(g)\le C}\big(\mbb{I}-\operatorname{Proj}(g,1)\big).
\end{equation}
Since all the summands commute, we know $\widetilde{H}^{(n)}$ has integer eigenvalues. Its ground state energy is 0, and its ground space projector is $\widetilde{\Pi}^{(n)}$.
The number of summands satisfies $\#\{g\in G_n\mid\operatorname{Diam}(g)\le C\}=O(n)$.

Let $D_{\mathrm{gs}}$ be the ground space degeneracy of this topological order. We have $\operatorname{dim}(\operatorname{span}(\mc{B}_n))=D_{\mathrm{gs}}$ for all $n$.
To show $U_n\Pi^{(n)}U_n^\dagger\approx\widetilde{\Pi}^{(n)}$, we first show that
\begin{equation}\label{eq:dimension_same_sufficientconditionfor_strongLRM}
\Tr(\mbb{I})/\abs{K_n}=\operatorname{dim}(\operatorname{span}(\mc{B}_n))=D_{\mathrm{gs}},
\end{equation}
when $n\in T$ is sufficiently large.
For all $\mathbf{b}\in\mbb{Z}_{t_1}\oplus\cdots\oplus\mbb{Z}_{t_k}$, since $\ket{\psi_{\mathbf{b}}}\in\mc{A}_n$ and $\operatorname{Diam}(U_nH_i^{(n)}U_n^\dagger)\le C$, we have
\begin{equation}
\bra{\psi_{\mathbf{b}}}U_nH^{(n)}U_n^\dagger\ket{\psi_{\mathbf{b}}}=\sum_{i\in I^{(n)}}\bra{\psi_{\mathbf{b}}}U_nH^{(n)}_iU_n^\dagger\ket{\psi_{\mathbf{b}}}=\sum_{i\in I^{(n)}}\bra{S_n}U_nH^{(n)}_iU_n^\dagger\ket{S_n}=\bra{S_n}U_nH^{(n)}U_n^\dagger\ket{S_n},
\end{equation}
so
\begin{align}
\abs{\bra{\psi_{\mathbf{b}}}U_nH^{(n)}U_n^\dagger\ket{\psi_{\mathbf{b}}}-\bra{\phi_n} H^{(n)}\ket{\phi_n}}=&\abs{\bra{S_n}U_nH^{(n)}U_n^\dagger\ket{S_n}-\bra{\phi_n}H^{(n)}\ket{\phi_n}}\\
=&\abs{\Tr\big[H^{(n)}\big(U_n^\dagger\ketbra{S_n}{S_n}U_n-\ketbra{\phi_n}{\phi_n}\big)\big]}\\
\le&\sum_{i\in I^{(n)}}\abs{\Tr\big[H_i^{(n)}\big(U_n^\dagger\ketbra{S_n}{S_n}U_n-\ketbra{\phi_n}{\phi_n}\big)\big]}\\
\le&\sum_{i\in I^{(n)}}\big\|H_i^{(n)}\big\|_\infty\big\|\ketbra{S_n}{S_n}-U_n\ketbra{\phi_n}{\phi_n}U_n^\dagger\big\|_1\\
\le&|I^{(n)}|o(1/n)=o(1).
\end{align}
Therefore, for arbitrary $\mathbf{b}\in\mbb{Z}_{t_1}\oplus\cdots\oplus\mbb{Z}_{t_k}$ we have
\begin{equation}
\bra{\psi_\mathbf{b}}U_nH^{(n)}U_n^\dagger\ket{\psi_\mathbf{b}}\le\bra{\phi_n}H^{(n)}\ket{\phi_n}+o(1)=\lambda_{\min}(U_nH^{(n)}U_n^\dagger)+o(1).
\end{equation}
Since $\{U_nH^{(n)}U_n^\dagger\}$ is topological ordered and thus gapped, by Lemma~\ref{lemma:small_energy_ground_dim} we know $t_1\cdots t_k\le\dim(\operatorname{span}(\mc{B}_n))$ when $n\in T$ is sufficiently large.
For arbitrary $\ket{b}\in\mc{B}_n$, since $\operatorname{Diam}(\operatorname{Proj}(g,1))\le C$ for all $g\in G_n$ with $\operatorname{Diam}(g)\le C$, we know
\begin{align}
\bra{b}\widetilde{H}^{(n)}\ket{b}=&\sum_{g\in G_n,\,\operatorname{Diam}(g)\le C}\big(1-\bra{b}\operatorname{Proj}(g,1)\ket{b}\big)\\
=&\sum_{g\in G_n,\,\operatorname{Diam}(g)\le C}\Big(1-\bra{\phi_n}U_n^\dagger\operatorname{Proj}(g,1)U_n\ket{\phi_n}\Big)=\bra{\phi_n}U_n^\dagger\widetilde{H}^{(n)}U_n\ket{\phi_n},
\end{align}
so
\begin{align}
\abs{\bra{b}\widetilde{H}^{(n)}\ket{b}-\bra{S_n}\widetilde{H}^{(n)}\ket{S_n}}=&\abs{\bra{\phi_n}U_n^\dagger\widetilde{H}^{(n)}U_n\ket{\phi_n}-\bra{S_n}\widetilde{H}^{(n)}\ket{S_n}}\\
=&\abs{\Tr\big[\widetilde{H}^{(n)}\big(U_n\ketbra{\phi_n}{\phi_n}U_n^\dagger-\ketbra{S_n}{S_n}\big)\big]}\\
\le&\sum_{g\in G_n,\,\operatorname{Diam}(g)\le C}\abs{\Tr\Big[\big(\mbb{I}-\operatorname{Proj}(g,1)\big)\big(U_n\ketbra{\phi_n}{\phi_n}U_n^\dagger-\ketbra{S_n}{S_n}\big)\Big]}\\
\le&\sum_{g\in G_n,\,\operatorname{Diam}(g)\le C}\big\|U_n\ketbra{\phi_n}{\phi_n}U_n^\dagger-\ketbra{S_n}{S_n}\big\|_1\\
=&O(n)o(1/n)=o(1).
\end{align}
Therefore, for arbitrary $\ket{b}\in\mc{B}_n$ we have 
\begin{equation}\label{eq:b_energy_vanishingsmall_strongLRMphase}
\bra{b}\widetilde{H}^{(n)}\ket{b}\le\bra{S_n}\widetilde{H}^{(n)}\ket{S_n}+o(1)=o(1).
\end{equation}
Since $\{\widetilde{H}^{(n)}\}$ is gapped, and we already know its ground space dimension $t_1\cdots t_k$ is upper bounded by $\operatorname{dim}(\operatorname{span}(\mc{B}_n))$ which is a constant independent of $n$, by Lemma~\ref{lemma:small_energy_ground_dim} we obtain $\dim(\operatorname{span}(\mc{B}_n))\le t_1\cdots t_k$ when $n\in T$ is sufficiently large.
So we have proved Eq.~\eqref{eq:dimension_same_sufficientconditionfor_strongLRM}.

Since $\widetilde{H}^{(n)}\ge\mbb{I}-\widetilde{\Pi}^{(n)}$, for arbitrary $\ket{b}\in\mc{B}_n$ we have
$\bra{b}\widetilde{H}^{(n)}\ket{b}\ge\bra{b}(\mbb{I}-\widetilde{\Pi}^{(n)})\ket{b}=1-\bra{b}\widetilde{\Pi}^{(n)}\ket{b}$.
By \eqref{eq:b_energy_vanishingsmall_strongLRMphase} we know
\begin{equation}
\bra{b}\widetilde{\Pi}^{(n)}\ket{b}\ge1-\bra{b}\widetilde{H}^{(n)}\ket{b}=1-o(1).
\end{equation}
We take $\{\ket{b_1},\cdots,\ket{b_{D_{\mathrm{gs}}}}\}$ be a set of orthogonal basis of $\operatorname{span}(\mc{B}_n)$, satisfying $\sum_{i=1}^{D_{\mathrm{gs}}}\ketbra{b_i}{b_i}=U_n\Pi^{(n)}U_n^\dagger$. Then
\begin{equation}
\Tr\big(U_n\Pi^{(n)}U_n^\dagger\widetilde{\Pi}^{(n)}\big)=\sum_{i=1}^{D_{\mathrm{gs}}}\bra{b_i}\widetilde{\Pi}^{(n)}\ket{b_i}\ge(1-o(1))D_{\mathrm{gs}}.
\end{equation}
Therefore we have
\begin{align}
\big\|U_n\Pi^{(n)}U_n^{\dagger}-\widetilde{\Pi}^{(n)}\big\|_2^2=&\Tr\left(\big(U_n\Pi^{(n)}U_n^{\dagger}-\widetilde{\Pi}^{(n)}\big)^2\right)\\
=&\Tr\big(U_n\Pi^{(n)}U_n^{\dagger}\big)+\Tr\big(\widetilde{\Pi}^{(n)}\big)-2\Tr\big(U_n\Pi^{(n)}U_n^\dagger\widetilde{\Pi}^{(n)}\big)\\
\le&D_{\mathrm{gs}}+D_{\mathrm{gs}}-2(1-o(1))D_{\mathrm{gs}}=o(1)D_{\mathrm{gs}}=o(1),
\end{align}
implying $\big\|U_n\Pi^{(n)}U_n^{\dagger}-\widetilde{\Pi}^{(n)}\big\|_2=o(1)$. Due to the inequality $\norm{A}_1\le\sqrt{\operatorname{rank}(A)}\norm{A}_2$, we know
\begin{equation}
\big\|U_n\Pi^{(n)}U_n^{\dagger}-\widetilde{\Pi}^{(n)}\big\|_1=o(1).
\end{equation}

Denote by $d_n$ ($\widetilde{d}_n$) the distance of the code $\Pi^{(n)}$ ($\widetilde{\Pi}^{(n)}$) (which means the reduced density matrices of all states in $\Pi^{(n)}$ ($\widetilde{\Pi}^{(n)}$) on arbitrary $d_n-1$ ($\widetilde{d}_n-1$) edges are the same). Let $\eta$ be a constant that upper bounds the light-cone sizes of $\{U_n\}$.
The deformed ground space $U_n\Pi^{(n)}U_n^\dagger$ has distance $\ge\lfloor\frac{d_n-1}{\eta}\rfloor+1$.
Now we show that when $n\in T$ is sufficiently large, $\widetilde{\Pi}^{(n)}$ is a code with distance
\begin{equation}
\widetilde{d}_n\ge\lfloor\frac{d_n-1}{\eta}\rfloor+1,
\end{equation}
which tends to infinity as $n\rightarrow\infty$. We note that this implies that all states in $\mc{A}_n'$ are $C$-locally indistinguishable, thus   $\mc{A}_n=\mc{A}_n'$.

Suppose on the contrary that there exists a subsequence $T'\subset T$, such that $\widetilde{d}_n\le\lfloor\frac{d_n-1}{\eta}\rfloor$ for all $n\in T'$.
Then there exists a family of Pauli strings $\{E_n\}_{n\in T'}$, where $E_n$ is a Pauli string that has nontrivial support on $\widetilde{d}_n$ edges, such that $E_n$ is a logical error for $\widetilde{\Pi}^{(n)}$ (That is, $E_n$ commutes with all elements in $K_n$, while $E_n\notin\{e^{i\theta}h\mid h\in K_n,\theta\in\mbb{R}\}$). This implies that $e^{i\theta}\mbb{I}\notin\{E_nh\mid h\in K_n\}$ for any $\theta\in\mbb{R}$.
We have
\begin{equation}
\widetilde{\Pi}^{(n)}E_n\widetilde{\Pi}^{(n)}=E_n\widetilde{\Pi}^{(n)} \text{, and } \Tr(E_n\widetilde{\Pi}^{(n)})=\frac{1}{\abs{K_n}}\sum_{h\in K_n}\Tr(E_nh)=0.
\end{equation}
Since $U_n\Pi^{(n)}U_n^\dagger$ has distance $\ge\lfloor\frac{d_n-1}{\eta}\rfloor+1$, we know there exist a constant $c_n\in\mbb{C}$ with $\abs{c_n}\le1$ such that
\begin{equation}
(U_n\Pi^{(n)}U_n^\dagger)E_n(U_n\Pi^{(n)}U_n^\dagger)=c_nU_n\Pi^{(n)}U_n^\dagger.
\end{equation}
Following the same calculation in the proof of Theorem~\ref{thm:LRM_from_TSC_app}, we obtain $\bignorm{E_n\widetilde{\Pi}^{(n)}-c_nU_n\Pi^{(n)}U_n^\dagger}_2=o(1)$, implying
\begin{equation}
\bignorm{c_nE_n\widetilde{\Pi}^{(n)}-\widetilde{\Pi}^{(n)}}_2=o(1).
\end{equation}
However, this is impossible since $\Tr(cE_n\widetilde{\Pi}^{(n)})=0$ and $\Tr(\widetilde{\Pi}^{(n)})=D_{\mathrm{gs}}>0$. Therefore, we conclude that $\widetilde{d}_n\ge\lfloor\frac{d_n-1}{\eta}\rfloor+1$ when $n\in T$ is sufficiently large.

When $n\in T$ is sufficiently large, the stabilizer code family $\{\widetilde{\Pi}^{(n)}\}_{n\in T}$ has the following properties:
\begin{enumerate}
\item The code space dimension $D_{\mathrm{gs}}$ is fixed.
\item It has stabilizer group $K_n$. The diameters of all elements in the subset $\{g\in G_n\mid\operatorname{Diam}(g)\le C\}\subset K_n$ are upper bounded by $C$, and this subset generates $K_n$.
\item It has macroscopic distance.
\end{enumerate}
Therefore, $\{\widetilde{\Pi}^{(n)}\}_{n\in T}$ is a TSC (with local configuration $(q_1,\cdots,q_m)$).

To conclude, this topological order has a local Hamiltonian realization (with local configuration $(q_1,\cdots,q_m)$) with ground space projectors $\{U_n\Pi^{(n)}U_n^{\dagger}\}$, and there exists a TSC (with local configuration $(q_1,\cdots,q_m)$) with code space projectors $\widetilde{\Pi}^{(n)}$, such that
$\big\|U_n\Pi^{(n)}U_n^{\dagger}-\widetilde{\Pi}^{(n)}\big\|_1=o(1)$.
This leads to a contradiction. In conclusion, this topological order is a strong LRM phase.

\end{proof}

\begin{cor}
Given local configuration $(q_1,\cdots,q_m)$, if the ground space degeneracy of a topological order contains a prime factor that does not divide $\prod_{i=1}^mq_i$, then it is a strong LRM phase.
\end{cor}

\begin{proof}
Assume, for the sake of contradiction, that there exists a family of local Hamiltonian $\{H^{(n)}\}_{n \in T}$ local configuration $(q_1,\cdots,q_m)$  that realizes this topological order, with ground states $\{\ket{\phi_n}\}_{n \in T}$, along with a family of shallow circuits $\{U_n\}_{n \in T}$ and a family of stabilizer states $\{\ket{S_n}\}_{n \in T}$, such that for all $n \in T$, $\left\| U_n \ketbra{\phi_n}{\phi_n} U_n^\dagger - \ketbra{S_n}{S_n} \right\|_1 = o(1/n)$.

Let $D_{\mathrm{gs}}$ denote the ground space degeneracy of this topological order. According to the proof of Theorem~\ref{thm:sufficientcondition_strongLRM_app}, there exists a topological stabilizer code (TSC) with local configuration $(q_1,\cdots,q_m)$ whose code space dimension $D_{\mathrm{code}}$ equals $D_{\mathrm{gs}}$.

However, the dimension of a TSC with local configuration $(q_1,\cdots,q_m)$ must divide $(\prod_{i=1}^m q_i)^{n_0}$, where $n_0$ is the number of edges corresponding to the smallest code in the family. Since $D_{\mathrm{gs}}$ contains a prime factor that does not divide $\prod_{i=1}^m q_i$, it follows that $D_{\mathrm{code}} \neq D_{\mathrm{gs}}$, which leads to a contradiction.
\end{proof}

\subsection{LRM phase examples}\label{app:number_theory_lemmas}

The following lemma is needed for the doubled Fibonacci example in the main text:
\begin{lem}
Let $\tau=\frac{1+\sqrt{5}}{2}$. Then for every integer $n\ge 0$,
\begin{equation}
\tau^{n}+\tau^{-n}=
\begin{cases}
L_n, & \text{if } n \text{ is even},\\[2mm]
\sqrt{5}F_n, & \text{if } n \text{ is odd},
\end{cases}
\end{equation}
where $F_n$ and $L_n$ are the Fibonacci and Lucas numbers defined by
$F_0=0, F_1=1, F_{n+1}=F_n+F_{n-1}$
and
$L_0=2, L_1=1, L_{n+1}=L_n+L_{n-1}$.

As a corollary, for all integers $n\ge0$, $(\tau^{n}+\tau^{-n})^2$ is an integer.
\end{lem}
\begin{proof}
Let $\sigma=\frac{1-\sqrt{5}}{2}$.
We have $\sigma=-\tau^{-1}$. 
By Binet's formulae~\cite{vajda2008fibonacci}:
\begin{equation}
F_n=\frac{\tau^n-\sigma^n}{\sqrt{5}},
\quad
L_n=\tau^n+\sigma^n,
\end{equation}
we obtain
\begin{equation}
\tau^n+\tau^{-n}
=\tau^n+(-\sigma)^n
=\tau^n+(-1)^n \sigma^n=
\begin{cases}
L_n, & \text{if } n \text{ is even},\\[2mm]
\sqrt{5}F_n, & \text{if } n \text{ is odd}.
\end{cases}
\end{equation}
Its square equals $L_{n}^2$ or $5F_{n}^2$, both of which are integers.
\end{proof}

The following lemma is needed for the $S_3$ quantum double example in the main text:
\begin{lem}
For $g\ge2$, there exists a prime factor of $\mathrm{GSD}(g)=2\cdot6^{2g-2}+4\cdot3^{2g-2}+2\cdot2^{2g-2}$ that does not divide 6.
\end{lem}
\begin{proof}
Let $m=2g-2$, so $\mathrm{GSD}(g)=2^{m+1}3^m+4\cdot3^{m}+2^{m+1}$. Therefore we have
\begin{equation}
\mathrm{GSD}(g)\equiv 0+0+2^{m+1}\equiv (-1)^{m+1}\not\equiv 0 \pmod{3},
\end{equation}
implying $3\nmid \mathrm{GSD}(g)$.
Notice that
\begin{equation}
Q:=4^{-1}\mathrm{GSD}(g)=2^{m-1}(3^{m}+1)+3^{m}
\end{equation}
is odd, because $2^{m-1}(3^{m}+1)$ is even and $3^{m}$ is odd.

For $m\ge 2$, we have $Q\ge 29>1$, so $Q$ has an odd prime factor $p$. 
By $3\nmid \mathrm{GSD}(g)$ we know $3\nmid Q$, so $p\neq 3$, implying $p\nmid 6$.
Therefore, $\mathrm{GSD}(g)$ has a prime factor $p$ that does not divide $6$.
\end{proof}

\section{Diagnosing LRM with correlation}\label{app:LRM_from_correlation}

\subsection{Two-qubit reduced density matrices of SRM states}\label{app:SRM_2qubits_cor_by_EPRs}
For a quantum circuit $U$ acting on $n$ qubits and a subset $A\subset[n]:=\{1,2,\cdots,n\}$, we define the backward light-cone $\mc{L}(A)\subset[n]$ as the set of qubits that can influence the output of $U$ at $A$.

\begin{thm}\label{thm:SRM_2qubits_cor_by_EPRs_app}
For arbitrary connectivity, given a family of SRM states specified by $\{\ket{\psi_n}=U_n\ket{S_n}\}$ where $\{U_n\}$ are shallow circuits and $\{\ket{S_n}\}$ are stabilizer states, the following holds:
For all subsystems $A,B\subset[n]$ of size $O(1)$ satisfying $\mc{L}(A)\cap\mc{L}(B)=\emptyset$, the bipartite reduced state $\Tr_{\overline{A\cup B}}\ketbra{\psi_n}{\psi_n}$ can be obtained by local operations on $O(1)$ EPR pairs $\ket{\Phi^+}=\frac{1}{\sqrt{2}}(\ket{00}+\ket{11})$ shared by the two parties, where $\overline{A\cup B}:=[n]-(A\cup B)$.   
\end{thm}
To prove this theorem, note that (see Fig.~\ref{fig:SRM_lightcone})
\begin{equation}
\Tr_{\overline{A\cup B}}\ketbra{\psi_n}{\psi_n}=\mc{E}_A\otimes\mc{E}_B\big(\Tr_{\overline{\mc{L}(A)\cup\mc{L}(B)}}\ketbra{S_n}{S_n}\big),
\end{equation}
where $\mc{E}_A(\sigma):=\Tr_{\mc{L}(A)-A}(U_A\sigma U_A^\dagger)$ is a completely-positive trace preserving (CPTP) map from $\mc{L}(A)$ to $A$, with $U_A$ being the circuit within the light-cone of $A$ (colored orange in Fig.~\ref{fig:SRM_lightcone}), and $\mc{E}_B$ is defined similarly.
$\Tr_{\overline{\mc{L}(A)\cup\mc{L}(B)}}\ketbra{S_n}{S_n}$ is a reduced state of a stabilizer state, and thus is a bipartite stabilizer mixed state~\cite{Aaronson2004Improved} on $M=\abs{\mc{L}(A)}+\abs{\mc{L}(B)}=O(1)$ qubits. 
It can only create rigid correlations between $\mc{L}(A)$ and $\mc{L}(B)$, similar to those of EPR pairs, thereby restricting the possible states that can be created by performing local operations on it.
By Lemma~\ref{lemma:structure_of_STABmixed} in Appendix~\ref{app:Bipartite_stabilizer_mixed_states}, $\Tr_{\overline{\mc{L}(A)\cup\mc{L}(B)}}\ketbra{S_n}{S_n}$ can be generated by local operations from $O(1)$ EPR pairs, and so can $\Tr_{\overline{A\cup B}}\ketbra{\psi_n}{\psi_n}$.

\subsection{Bipartite stabilizer mixed states}\label{app:Bipartite_stabilizer_mixed_states}

Given the stabilizer generator set $\{g_1, \cdots, g_r\}$ for an $n$-qubit stabilizer code, where $r \le n$ and $g_1, \cdots, g_r$ are independent, mutually commuting Pauli strings such that $\langle g_1, \cdots, g_r \rangle$ does not contain $-\mbb{I}$, the corresponding \emph{stabilizer mixed state}~\cite{Aaronson2004Improved} is defined as $\frac{1}{2^r}\prod_{i=1}^r(\mbb{I} + g_i)$, which is the uniform distribution over all states within the code space. Stabilizer mixed states are in one-to-one correspondence with stabilizer codes.

\begin{lem}\label{lemma:structure_of_STABmixed}
A bipartite stabilizer mixed state $\rho$ on $N$ qubits can be obtained by local operations on at most $N$ EPR pairs.
\end{lem}
\begin{proof}
By Theorem 1 in Ref.~\cite{Audenaert2005mixed}, every bipartite stabilizer mixed state $\rho_{AB}$ on $N=N_A+N_B$ qubits with $r$ independent generators is locally unitary equivalent to $\ketbra{\Phi^+}{\Phi^+}^{\otimes p}\otimes\rho'$, where $\ket{\Phi^+}=\frac{1}{\sqrt{2}}(\ket{00}+\ket{11})$ is EPR pair shared by $A$ and $B$, and
\begin{equation}
\rho'=\frac{1}{2^{N-2p}}\prod_{l=2p+1}^r(\mbb{I}+h_l)
\end{equation}
is a bipartite separable stabilizer mixed state acting on $(N_A-p)+(N_B-p)$ qubits, with $h_l$ being a tensor product of single-qubit $\mbb{I}$ and $X$, acting on $N-2p$ qubits.
The number of EPR pairs $p$ is upper bounded by $p\le\min(\lfloor r/2\rfloor,N_A,N_B)$.

The separable state $\rho'$ is diagonal in the $\{\ket{+},\ket{-}\}$ basis and is the equal mixture of $2^{N-r}$ such states. It can be created by local operations on $N-r$ EPR pairs by measuring the EPR pairs in the computational basis and preparing the corresponding bipartite pure state in the state ensemble of $\rho'$, depending on the measurement results.
Therefore, we see that every such stabilizer mixed state $\rho$ on $N$ qubits can be obtained by at most $p+N-r\le p+N-2p\le N$ EPR pairs.
\end{proof}

\subsection{Generate correlations by EPR pairs}\label{app:correlations_EPR}

The following Proposition holds for any single Pauli $P\in\{X,Y,Z\}$. We take $P=Z$ without loss of generality.
\begin{prop}\label{prop:correlations_EPR}
For integer $K\ge1$, $K$ EPR pairs can generate a two-qubit state $\rho$ by local CPTP maps, with $\langle Z\otimes Z\rangle_\rho=b$ and $\langle Z\otimes \mbb{I}\rangle_\rho=\langle \mbb{I}\otimes Z\rangle_\rho=c$, if and only if $(b,c)$ lies in the region given by
\begin{empheq}[left=\empheqlbrace]{align}
\label{eq:bc_region1}
&1+2c+b\ge0,\\
\label{eq:bc_region2}
&1-2c+b\ge0,\\
\label{eq:bc_region3}
&\frac{1+2c+b}{4}2^K\le l+\left(\frac{1+c}{2}2^K-l\right)^2,\text{ for }l=0,1,\cdots,2^K-1.
\end{empheq}
\end{prop}
See Fig.~\ref{fig:bc_correlation} for a demonstration of the region specified by constraints \eqref{eq:bc_region1}--\eqref{eq:bc_region3}.
\begin{proof}
($\Rightarrow$):

Suppose $\rho=\mc{E}_A\otimes\mc{E}_B\big(\ketbra{\Phi^+}{\Phi^+}^{\otimes K}\big)$ for some local CPTP maps $\mc{E}_A$ and $\mc{E}_B$.
Measuring $\rho$ in the computational basis, we obtain $\ket{00}$ with probability $(1+2c+b)/4$, $\ket{11}$ with probability $(1-2c+b)/4$, and $\ket{01}$ or $\ket{10}$ with probability $(1-b)/4$. Therefore, we get a classical correlation~\cite{Jain2013Efficient}
\begin{equation}\label{eq:correlation}
\frac{1}{4}
\begin{bmatrix}
1+2c+b & 1-b \\
1-b & 1-2c+b \\
\end{bmatrix}.
\end{equation}
For \eqref{eq:correlation} to be a valid correlation, each element should be non-negative.
By measuring local POVMs $\{\mc{E}_A^\dagger(\ketbra{0}{0}),\mc{E}_A^\dagger(\ketbra{1}{1})\}$ and $\{\mc{E}_B^\dagger(\ketbra{0}{0}),\mc{E}_B^\dagger(\ketbra{1}{1})\}$ on $\ketbra{\Phi^+}{\Phi^+}^{\otimes K}$, we can get the correlation \eqref{eq:correlation}.
All $2\times2$ classical correlations $M$ that can be obtained by $K$ EPR pairs, can be obtained by measuring local POVMs $\{E,\mbb{I}_{2^K}-E\}$ and $\{F,\mbb{I}_{2^K}-F\}$ acting on $K$ qubits on $\ketbra{\Phi^+}{\Phi^+}^{\otimes K}$, and can be written as 
\begin{align}
\label{eq:correlation2}
M=&\begin{bmatrix}
\Tr\left(E\otimes F\ketbra{\Phi^+}{\Phi^+}^{\otimes K}\right) & \Tr\left(E\otimes (\mbb{I}_{2^K}-F)\ketbra{\Phi^+}{\Phi^+}^{\otimes K}\right) \\
\Tr\left((\mbb{I}_{2^K}-E)\otimes F\ketbra{\Phi^+}{\Phi^+}^{\otimes K}\right) & \Tr\left((\mbb{I}_{2^K}-E)\otimes (\mbb{I}_{2^K}-F)\ketbra{\Phi^+}{\Phi^+}^{\otimes K}\right)
\end{bmatrix}\\
\label{eq:correlation2'}
=&2^{-K}\begin{bmatrix}
\Tr\left(E^{\T} F\right) & \Tr\left(E\right)-\Tr\left(E^{\T} F\right) \\
\Tr\left(F\right)-\Tr\left(E^{\T} F\right) & 2^K-\Tr\left(E\right)-\Tr\left(F\right)+\Tr\left(E^{\T} F\right)
\end{bmatrix}.
\end{align}
Letting \eqref{eq:correlation2'} $=$ \eqref{eq:correlation}, we get
\begin{empheq}[left=\empheqlbrace]{align}
\label{eq:EF_constraints1}
&\Tr(E)=\frac{1+c}{2}2^K,\\
\label{eq:EF_constraints2}
&\Tr(F)=\frac{1+c}{2}2^K,\\
\label{eq:EF_constraints3}
&\Tr(E^{\T}F)=\frac{1+2c+b}{4}2^K,\\
\label{eq:EF_constraints4}
&0\le E,F\le \mbb{I}_{2^K}.
\end{empheq}
Note that $\Tr(E^{\T}F)=\Tr((E^*)^\dagger F)\le\twonorm{E^*}\twonorm{F}=\twonorm{E}\twonorm{F}$,
where $E^*$ is the complex conjugate of $E$. For $l=0,1,\cdots,2^{K}-1$, when $l\le\frac{1+c}{2}2^K\le l+1$,
we have $\twonorm{E},\twonorm{F}\le\sqrt{l+(\frac{1+c}{2}2^K-l)^2}$. Therefore, we obtain
\begin{equation}
\Tr(E^{\T}F)\le l+\left(\frac{1+c}{2}2^K-l\right)^2,
\end{equation}
giving the constraint $\frac{1+2c+b}{4}2^K\le l+(\frac{1+c}{2}2^K-l)^2$.
These $2^K$ constraints also imply $1-b\ge0$.

($\Leftarrow$): 

Suppose $(b,c)$ lies in the region given by \eqref{eq:bc_region1}--\eqref{eq:bc_region3}. We claim that we can generate the correlation \eqref{eq:correlation} by measuring some POVMs $\{E,\mbb{I}_{2^K}-E\}$ and $\{F,\mbb{I}_{2^K}-F\}$ on $\ketbra{\Phi^+}{\Phi^+}^{\otimes K}$, that is, we can find $E$ and $F$ satisfying the constraints in \eqref{eq:EF_constraints1}--\eqref{eq:EF_constraints4}.

For $l=0,1,\cdots,2^{K-1}-1$, when $c$ satisfies $l\le\frac{1+c}{2}2^K\le l+1$, we have
\begin{equation}
0\le\frac{1+2c+b}{4}2^K\le l+\left(\frac{1+c}{2}2^K-l\right)^2.
\end{equation}
Since $\Tr(E)=\Tr(F)=\frac{1+c}{2}2^K\le 2^{K-1}$, we can take $E^*$ and $F$ such their support are disjoint and get $\Tr(E^{\T}F)=0$. By taking $E=F=\operatorname{diag}(1,\cdots,1,\frac{1+c}{2}2^K-l,0,\cdots,0)$ with $l$ many 1's, we get $\Tr(E^{\T}F)=l+(\frac{1+c}{2}2^K-l)^2$. Moreover, $\Tr(E^{\T}F)$ takes arbitrary values in the region $[0,l+(\frac{1+c}{2}2^K-l)^2]$ by appropriately choosing $0\le E,F\le\mbb{I}_{2^K}$.
Therefore, we get $E$ and $F$ satisfying the constraints in \eqref{eq:EF_constraints1}--\eqref{eq:EF_constraints4}.

For $l=2^{K-1},1,\cdots,2^{K}-1$, when $c$ satisfies $l\le\frac{1+c}{2}2^K\le l+1$, we have 
\begin{equation}
c2^K\le\frac{1+2c+b}{4}2^K\le l+\left(\frac{1+c}{2}2^K-l\right)^2.
\end{equation}
By taking $E=\operatorname{diag}(1,\cdots,1,c,\cdots,c)$ and $F=\operatorname{diag}(c,\cdots,c,1,\cdots,1)$ with $2^{K-1}$ 1's and $2^{K-1}$ $c$'s, we get $\Tr(E)=\Tr(F)=\frac{1+c}{2}2^K$ and $\Tr(E^{\T}F)=c2^K$. 
By taking $E=F=\operatorname{diag}(1,\cdots,1,\frac{1+c}{2}2^K-l,0,\cdots,0)$ with $l$ 1's, we get $\Tr(E^{\T}F)=l+(\frac{1+c}{2}2^K-l)^2$. Moreover, $\Tr(E^{\T}F)$ takes arbitrary values in the region $[c2^K,l+(\frac{1+c}{2}2^K-l)^2]$ by appropriately choosing $0\le E,F\le\mbb{I}_{2^K}$.
Therefore, we get $E$ and $F$ satisfying the constraints in \eqref{eq:EF_constraints1}--\eqref{eq:EF_constraints4}.

Given $(b,c)$, measuring the corresponding POVMs $\{E,\mbb{I}_{2^K}-E\}$ and $\{F,\mbb{I}_{2^K}-F\}$ on $\ketbra{\Phi^+}{\Phi^+}^{\otimes K}$ yields the correlation \eqref{eq:correlation}, then prepare the diagonal state
\begin{equation}
\rho:=\frac{1+2c+b}{4}\ketbra{00}{00}+\frac{1-b}{4}\ketbra{01}{01}+\frac{1-b}{4}\ketbra{10}{10}+\frac{1-2c+b}{4}\ketbra{11}{11}
\end{equation}
conditioned on the measurement results. We have $\langle Z\otimes Z\rangle_\rho=b$ and $\langle Z\otimes \mbb{I}\rangle_\rho=\langle \mbb{I}\otimes Z\rangle_\rho=c$, and $\rho$ is prepared by local CPTP maps on $K$ EPR pairs.
\end{proof}

\subsection{Details for $\ket{\mathrm{C}^{n-1}Z}$}\label{app:CCCZ}
The reduced density matrix of $\ket{\mathrm{C}^{n-1}Z}$~\cite{wei2024noise} on two qubits $\{k,l\}\subset[n]$ is
\begin{equation}
\begin{aligned}
\Tr_{\overline{\{k,l\}}}\ketbra{\mathrm{C}^{n-1}Z}{\mathrm{C}^{n-1}Z}=\big(1-\frac{1}{2^{n-2}}\big)\ketbra{++}{++}+\frac{1}{2^{n-2}}\ketbra{\mathrm{C}Z}{\mathrm{C}Z}.
\end{aligned}
\end{equation}
We have $\bra{\mathrm{C}Z}X\otimes X\ket{\mathrm{C}Z}=0$, thus 
\begin{equation}
\langle X_kX_l\rangle=\Tr\left(X\otimes X\Tr_{\overline{\{k,l\}}}\ketbra{\mathrm{C}^{n-1}Z}{\mathrm{C}^{n-1}Z}\right)=1-\frac{1}{2^{n-2}}.
\end{equation}
On one qubit we have
\begin{equation}
\begin{aligned}
\Tr_{\overline{\{k\}}}\ketbra{\mathrm{C}^{n-1}Z}{\mathrm{C}^{n-1}Z}=\big(1-\frac{1}{2^{n-1}}\big)\ketbra{+}{+}+\frac{1}{2^{n-1}}\ketbra{-}{-},
\end{aligned}
\end{equation}
thus 
\begin{equation}
\langle X_k\rangle=1-\frac{1}{2^{n-1}}-\frac{1}{2^{n-1}}=1-\frac{1}{2^{n-2}}.
\end{equation}

\begin{prop}
$\ket{\mathrm{C}^{n-1}Z}$ state family exhibits LRM.
\end{prop}

\begin{proof}
Suppose on the contrary that there exists $\{n_q\}_{q=1}^\infty\subset\mbb{Z}_{>0}$ such that $\{\ket{\mathrm{C}^{n_q-1}Z}=U_{n_q}\ket{S_{n_q}}\}$ has SRM.
Denote by $M_q=\{j\in[n_q]:\mc{L}(j)\cap \mc{L}(1)\neq\emptyset\}$ the set of qubits whose light-cone with regard to $U_{n_q}$ intersects with that of qubit $1$. We have $\abs{M_q}=O(1)$.
By Theorem~\ref{thm:SRM_2qubits_cor_by_EPRs_app}, for all qubits $k\notin M_q$, the reduced state 
$\Tr_{\overline{\{1,k\}}}\ketbra{\mathrm{C}^{n_q-1}Z}{\mathrm{C}^{n_q-1}Z}$ can be generated by local CPTP maps on $O(1)$ EPR pairs. However, we know that when $n_q$ is sufficiently large, $\Tr_{\overline{\{1,k\}}}\ketbra{\mathrm{C}^{n_q-1}Z}{\mathrm{C}^{n_q-1}Z}$ cannot be generated by local CPTP maps on $O(1)$ EPR pairs, for all $k\neq1$, leading to a contradiction.
Therefore, $\{\ket{\mathrm{C}^{n-1}Z}\}$ has LRM.
\end{proof}

\section{NLTM as a stronger necessary condition for quantum PCP}\label{app:NLTM_qPCP}

Here we provide the rigorous argument that the validity of the NLTM conjecture is necessary for the quantum PCP conjecture to hold (under the standard complexity-theoretic belief that $\mathsf{NP}\neq\mathsf{QMA}$). That is, as remarked in the main text, NLTM is a strictly stronger version of NLTS that moves significantly closer to quantum PCP, since it takes into account more general types of classical witnesses. This statement was also mentioned in a talk~\cite{Nirkhe_2023}.

\begin{prop}\label{prop:NLTM_qPCP-app}
Suppose $\mathsf{NP}\neq\mathsf{QMA}$ and that the quantum PCP conjecture holds, then the NLTM conjecture holds.
\end{prop}

Our proof below builds on the proof for the analogous statement on NLTS in Ref.~\cite{Nirkhe2022Lower}. Various necessary preliminary definitions can be found therein.  

\begin{proof}
Assume on the contrary that NLTM is false, while $\mathsf{NP}\neq\mathsf{QMA}$ and the quantum PCP conjecture holds.

Take an arbitrary $\mathsf{QMA}$ decision problem $L$. Given the validity of quantum PCP, there is a polynomial-time reduction that maps an input instance $x$ of $L$ to a local Hamiltonian $H^{(n)}=\sum_{i=1}^{m^{(n)}}H_i^{(n)}$ acting on $n$ qubits ($H^{(n)}$ depends on $x$, $n$ relies on the size of $x$), with $m^{(n)}=\Theta(n)$ and $\|H_i^{(n)}\|_\infty\le1$.
The promise for $H^{(n)}$ is that there are two energy thresholds $a^{(n)}$ and $b^{(n)}$ satisfying $b^{(n)}-a^{(n)}\ge\gamma n$ for some constant $\gamma>0$ such that
\begin{enumerate}
\item YES instance: If $x\in L$, then $\lambda_{\min}(H^{(n)}) \le a^{(n)}$;\\
\item NO instance: If $x\notin L$, then $\lambda_{\min}(H^{(n)}) \ge b^{(n)}$.
\end{enumerate}

We now show that a consequence of the invalidity of NLTM is that $L\in\mathsf{NP}$.

Consider a family of YES instances $\{x\}$ and the corresponding $\{H^{(n)}\}$.
By the assumption that  NLTM is false, for any constant $\gamma>0$, there exists a family of states $\{\ket{\phi_n}\}$ with 
\begin{equation}
\bra{\phi_n}H^{(n)}\ket{\phi_n}\le\lambda_{\min}(H^{(n)})+\gamma n\le a^{(n)}+\gamma n=b^{(n)}
\end{equation}
that exhibits SRM, meaning that there exist shallow circuit family $\{U_n\}$ and stabilizer states $\{\ket{S_n}\}$ such that $\ket{\phi_n}=U_n\ket{S_n}$.
The classical description of $U_n$ and $\ket{S_n}$ can serve as an $\mathsf{NP}$ witness for a YES instance $H^{(n)}$.
To see this, first notice that $\bra{\phi_n}H^{(n)}\ket{\phi_n}<b^{(n)}$ proves $\lambda_{\min}(H^{(n)})\le a^{(n)}$ due to the promise.
Also, $\bra{\phi_n}H^{(n)}\ket{\phi_n}$ can be computed efficiently classically by the classical description of $U_n$ and $\ket{S_n}$, because
\begin{equation}
\bra{\phi_n}H^{(n)}\ket{\phi_n}=\sum_{i=1}^{m^{(n)}}\bra{S_n}U_n^\dagger H_i^{(n)}U_n\ket{S_n},
\end{equation}
where $\big|\supp\big(U_n^\dagger H_i^{(n)}U_n\big)\big|=O(1)$ implies $\bra{S_n}U_n^\dagger H_i^{(n)}U_n\ket{S_n}$ can be computed efficiently by e.g.~decomposing $U_n^\dagger H_i^{(n)}U_n$ into the sum of $O(1)$ Pauli strings~\cite{Aaronson2004Improved}.

For a NO instance $x\notin L$ (and the corresponding Hamiltonian $H^{(n)}$), the promise implies $\lambda_{\min}(H^{(n)})\ge b^{(n)}$, hence $\bra{\psi}H^{(n)}\ket{\psi}\ge b^{(n)}$ for every $n$-qubit state $\ket{\psi}$. Therefore no witness can produce an energy estimate below $b^{(n)}$, and the verifier rejects.

Consequently, we conclude that $L\in\mathsf{NP}$. Since $L\in\mathsf{QMA}$ is taken arbitrarily, we must have $\mathsf{QMA}\subset\mathsf{NP}$, contradicting the  $\mathsf{NP}\neq\mathsf{QMA}$ assumption. Therefore, NLTM is a necessary condition for quantum PCP to hold unless $\mathsf{NP}=\mathsf{QMA}$. 
\end{proof}

\section{Remark on the relation between LRM and nonlocal magic}\label{app:LRM_and_nonlocalMagic}

Given a partition of the entire system into several subsystems, we say a nonstabilizer state has \emph{entirely nonlocal magic} if its local reduced density matrices on subsystems have no magic. Here we point out:

\begin{prop}
LRM is a strictly stronger condition compared to entirely nonlocal magic.
\end{prop}
\begin{proof}
Suppose $\{\ket{\psi_n}\}$ has LRM, then for an arbitrary disjoint local partition $[n]=R_1\cup\cdots\cup R_{t_n}$ with $\abs{R_i}=O(1)$, we can use a shallow circuit $U^{(n)}=U_1\otimes\cdots\otimes U_{t_n}$ to transform the reduced density matrix on each $R_i$ into a matrix that is diagonal in the computational basis, thereby eliminating the magic on each $R_i$.
Meanwhile, $U^{(n)}\ket{\psi_n}$ remains magical, because $U^{(n)}$ cannot remove all the magic of $\ket{\psi_n}$ when $n$ is sufficiently large.

On the other hand, some families of states with SRM can have entirely nonlocal magic. For example, the logical $\ket{T}^{\otimes k}$ state in 3D color codes can be obtained by performing a single layer of non-Clifford gates on logical $\ket{+}^{\otimes k}$, a stabilizer state, thus exhibiting SRM. However, its magic is completely nonlocal because regions smaller than the code distance have no magic.
\end{proof}

\end{document}